\documentclass[aps,physrev,onecolumn,10pt,amsmath,amssymb,nofootinbib,longbibliography]{revtex4-2}
\usepackage{amsthm,amsfonts,mathrsfs,mathtools}
\usepackage[utf8]{inputenc}
\usepackage[T1]{fontenc}
\usepackage{lmodern,graphicx,array,booktabs}
\usepackage{tikz}
\usetikzlibrary{arrows.meta,positioning,calc}
\usepackage{xcolor}
\usepackage{hyperref}
\hypersetup{colorlinks=true,linkcolor=blue,citecolor=blue,urlcolor=blue,
 pdftitle={Sequential Capacity of Quantum Processes with Finite Memory},pdfauthor={Yibin Wang}}
\theoremstyle{plain}
\newtheorem{theorem}{Theorem}[section]
\newtheorem{lemma}[theorem]{Lemma}
\newtheorem{proposition}[theorem]{Proposition}
\newtheorem{corollary}[theorem]{Corollary}
\theoremstyle{definition}
\newtheorem{definition}[theorem]{Definition}
\newtheorem{remark}[theorem]{Remark}
\DeclareMathOperator{\Tr}{Tr}
\DeclareMathOperator{\rank}{rank}
\DeclareMathOperator{\diag}{diag}

\begin{document}

\title{Sequential Capacity of Quantum Processes with Finite Memory}
\author{Yibin Wang}
\email{yibinw0210@gmail.com}
\affiliation{Graduate School of Mathematics, Nagoya University, Nagoya, 464-8601, Aichi, Japan.}

\begin{abstract}
How complex can the responses of a quantum process become as it runs longer
with fixed internal memory? We measure this complexity by sequential
response capacity, the depth of an adaptive tree of experiments whose
acceptance probabilities separate possible processes by a prescribed gap.
Each experiment uses a fresh complete run, and one fixed process must
realize every response along a path. For any fixed common visible dimension
of at least two and internal memory cap, we prove that capacity has order
$K\log((K+1)/\gamma)$ for $K$ transitions per run and probability gap
$0<\gamma\leq1$. Time-dependent single-qubit phase
channels without internal memory already attain the fixed-gap $K\log K$
growth, realizing each path with zero-or-one response probabilities.
Under the same tester permissions, classical stochastic
processes that measure in a fixed basis at every transition have only
linear capacity at fixed dimensions and gap. For phase sequences selected by a
classical address stored inside the process, we then quantify how known
independent Pauli noise changes this logarithmic enhancement. With ideal tester
controls and weak residual phase noise after correction, we prove matching
capacity bounds at a fixed small gap. At fixed address size and fixed
positive residual noise, the logarithmic enhancement saturates once the
phase-sequence length exceeds the inverse noise scale, and capacity then
grows linearly with sequence length.
\end{abstract}
\maketitle

\section{Introduction}
\label{main:introduction}

A quantum process exchanges a sequence of inputs and outputs while
carrying an internal register between transitions. Fixing the size of this
register sets a persistent-memory budget, but the transition channels may
still vary with time, allowing the number of unknown parameters to grow
with duration at fixed memory. We study how this growth appears in the
class of observable responses.

Within one run, a tester may retain arbitrary finite private quantum
registers and adapt its interventions to earlier outcomes before ending
with a binary measurement. Across fresh runs, the next experiment
may depend on earlier response labels. How long can these experiments
separate possible processes by a fixed probability gap?

Sequential fat-shattering dimension \cite{rakhlin2015sequential}
expresses this question through binary trees of experiments: it is the
supremum of tree depths for which each complete path is realized by one
fixed process, with its responses on the prescribed sides of the thresholds
by the required margin.
Fat-shattering dimension also enters the analysis of learning quantum states
from measurement probabilities~\cite{aaronson2007learnability}.

Sequential state generation uses a finite-dimensional ancillary system
that interacts successively with fresh systems
\cite{Schon2005sequential}. Memory-cost formulations quantify the quantum
register required by a multiround protocol with classical assistance
\cite{bisio2012memory}. Here the internal memory budget includes every
persistent classical record as well as the quantum register.

Write $d$ for the common dimension of the input and output systems forming
the visible interface, $r$ for the dimension cap on the complete register
propagated inside the process, and $K$ for its number of transitions.
For a required probability gap $\gamma$, we denote the sequential capacity
by $C_\gamma(K;d,r)$.

Raza et al.\ construct strategy-norm covers and obtain online guarantees
for a specified bounded-component class of multitime qubit processes
\cite{raza2024online}. Their method approximates local components and adds
the approximation errors, giving the upper order
$O_{d,r}(K\log((K+1)/\gamma))$ when applied to the initial state and the
fixed-dimensional transition channels here.
Appendix~\ref{up:cover-origin} gives the source bound and this connection.
At fixed $d,r,\gamma$, can an allowed response tree attain the $K\log K$
order of this upper bound?

For the noisy family below, channel simulation at fixed known noise with
positive residual noise gives a tester-independent target encoding for an
upper bound.
Ref.~\cite{matsumoto2010metric} establishes parameter-independent channel
simulation with adaptive ancillary processing. Classical simulation of
noisy phase channels underlies dephasing bounds in quantum
metrology~\cite{demkowicz2012elusive}. Raza et al.'s bound for convex mixtures
of known multitime processes~\cite{raza2024online} supplies a learning
antecedent.

For the lower bound, ideal tester controls retain coherence in short blocks. With ideal fast recovery, quantum metrology retains
precision proportional to inverse evolution time over a window set by
weak residual noise~\cite{zhou2018heisenberg}.

Bansal et al.'s full-strategy online regret bound~\cite{bansal2025panoply}
gives a tree benchmark for our full-causal upper bound.

\subsection{Main results}

For every fixed $d\geq2$ and $r\geq1$, we prove the joint
duration--precision law
\[
 C_\gamma(K;d,r)=\Theta_{d,r}\!\left(K\log_2\frac{K+1}{\gamma}\right),
 \qquad K\geq1,\quad0<\gamma\leq1.
\]
Here $\Theta_{d,r}$ denotes two-sided bounds with positive constants
depending only on $d,r$. At fixed gap, the
$K\log(K+1)$ duration dependence is attained by time-dependent qubit
phase channels without propagated memory. Each witness gives zero-or-one
responses along its path. Boolean inversion encodes dyadic data in one
fixed phase vector that is read by queries with coefficients in $\{-1,0,1\}$
and one forward call per transition. The precision dependence comes from
probability bisection, with related one-observable halving trees appearing
in Ref.~\cite{meyer2025pure}.

Under the same tester permissions, classical stochastic processes that
completely measure the input and memory at each transition have linear
capacity in $K$ at fixed visible dimension, internal state count and gap.

For a target fixed before learning and fixed $d,r,\varepsilon$ with
$0<\varepsilon<1/2$, suppose each of $N$ rounds gives only scalar
probability feedback accurate to $\varepsilon/8$. The exact quantum
response tree and a physical cover then give the minimax expected number
of prediction errors larger than $\varepsilon$ as
$\Theta_{d,r,\varepsilon}(\min\{N,K\log(K+1)\})$.
In the single-outcome protocol, each fresh complete run supplies one binary
outcome, and experiment selection and learning use only the public classical
history and randomness independent of the fixed target. For fixed $d,r$,
the minimax expected average squared error relative to the acceptance
probabilities is $\Theta_{d,r}(\min\{1,K\ln(K+N)/N\})$.

To determine how known independent noise after each phase changes
capacity, we next use a measured classical address of dimension
$R=2^h$ to select one of $R$ phase sequences of length $T$ whose
entries vary independently across targets. The address is included in
the propagated memory and loaded in $h$ transitions, giving total
horizon $K=h+T$ for $T\geq1$.
Let $q=(q_I,q_X,q_Y,q_Z)$ be the probability vector for the identity
and three Pauli errors, and put $e=\min(q_I,q_Z)+\min(q_X,q_Y)$, the residual phase-flip
probability after a specified syndrome correction. For fixed
$0<\gamma\leq1/16$ and all $0\leq e\leq1/8$, with ideal tester
controls for the lower construction, the capacity has order
\mbox{$\Theta_\gamma(RT\log_2[1+\min(T,1/e)])$}, uniformly in $R,T,q$.
The convention $1/0=\infty$ includes the noiseless endpoint and pure
transverse bit flips. A phase cover gives an upper bound for all normalized
adaptive testers. For $e>0$, encoding each phase table in a quantum state
for a fixed simulator gives an exact program and a second bound, and we
take the smaller bound.

Appendices~\ref{low:memoryless-section} and~\ref{up:cover-origin} give
the lower constructions and covering bounds.
Appendices~\ref{pred:section} and~\ref{temp:section} treat prediction and
tester resources. Section~\ref{main:memory-regimes} compares growing-memory
bounds and the saturated endpoint. Appendices~\ref{main:growing}--\ref{cost:section}
give the proofs and simulation-storage extensions.

\section{Processes, tests, and sequential capacity}
\label{main:model}

Fix integers $d\geq2$, $r\geq1$, and $K\geq1$. An initialized realization
consists of a state $\rho_{B_0M_0}$ and completely positive trace-preserving
(CPTP) maps
\begin{equation}
 \Phi_t:A_tM_{t-1}\longrightarrow B_tM_t,
 \qquad t=1,\ldots,K,
 \label{main:realization}
\end{equation}
where $\dim A_t=d$ for $1\leq t\leq K$, and $\dim B_t=d$ and
$\dim M_t\leq r$ for $0\leq t\leq K$.
All persistent classical labels are included in $M_t$, while fresh local
environments discarded at a transition are allowed.

A complete run starts by preparing $\rho_{B_0M_0}$ and delivering $B_0$
to the tester while retaining $M_0$ inside the process, after which the
tester supplies $A_t$ and receives $B_t$ for each $t=1,\ldots,K$, using
its retained systems and earlier outcomes to choose each intervention.
The final $M_K$ is discarded. Let $\mathcal P_{d,r}(K)$ be the observable
processes admitting such a realization. One fixed tuple
$(\rho_{B_0M_0},\Phi_1,\ldots,\Phi_K)$ represents a target throughout
all its experiments, so a new run prepares the same initial state and
uses the same transition maps.

Memory one permits time-dependent channels, with no hidden register
passed from one transition to the next. For auxiliary direct-sum statements
we also use $K=0$: the process then consists only of its initialized state,
$B_0$ is delivered and $M_0$ is discarded. The principal horizon laws below
assume $K\geq1$.

Starting with no initially supplied information about the unknown target,
a normalized binary tester may initialize arbitrary finite private
registers, retain earlier outputs, adapt its interventions, and perform
a final binary measurement. Each intervention is a complete instrument,
whose outcome maps are completely positive and sum to a trace-preserving
map, and the final accepting probability includes all branches without
conditioning on a successful intermediate outcome. Denote its accepting
probability on process $P$ by $f_P(E)\in[0,1]$, and define the full
operational distance by
\begin{equation}
 D(P,Q)=2\sup_E|f_P(E)-f_Q(E)|.
 \label{main:distance}
\end{equation}
This is the strategy-norm distance for the complete multiround
interaction~\cite{gutoski2012distance}.
The tester's private workspace is a separate resource from the memory
retained inside the process.

We study capacity information-theoretically, identifying each process $P$
with its response function $E\mapsto f_P(E)$ on this common tester domain.
Each tree node specifies an experiment on a fresh complete run, with its
label indicating a side of a probability threshold rather than necessarily
a sampled outcome. Along each path, every node uses the same fixed target,
including its initialization and channels. Tree depth counts experiments,
while horizon counts transitions within a run.

\begin{definition}[Sequential response capacity]
\label{main:sfat}
For $0<\gamma\leq1$, a depth-$L$ binary tree is $\gamma$-shattered if
every node $u\in\{0,1\}^{<L}$ has a tester $E_u$ and real threshold $s_u$,
and every leaf $x\in\{0,1\}^L$ has one process
$P_x\in\mathcal P_{d,r}(K)$ such that
\begin{equation}
 (2x_j-1)\bigl[f_{P_x}(E_{x_{<j}})-s_{x_{<j}}\bigr]
 \geq\gamma/2\qquad (j=1,\ldots,L).
 \label{main:shattering}
\end{equation}
Here $x_{<j}=(x_1,\ldots,x_{j-1})$ is the prefix preceding bit $j$.
Thus $\gamma$ is the required separation between the two response ranges,
with margin $\gamma/2$ on either side of a node's threshold.
The supremum of such depths is
$C_\gamma(K;d,r):=\operatorname{sfat}_\gamma\mathcal P_{d,r}(K)$.
\end{definition}

\section{The fixed-memory law and its phase mechanism}
\label{main:fixed}

\begin{theorem}[Fixed-memory horizon and precision law]
\label{main:fixed-law}
For every fixed $d\geq2$ and $r\geq1$,
\begin{equation}
 C_\gamma(K;d,r)
 =\Theta_{d,r}\!\left(K\log_2\frac{K+1}{\gamma}\right),
 \qquad K\geq1,\quad0<\gamma\leq1.
 \label{main:fixed-equation}
\end{equation}
The constants can depend on $d,r$, and are independent of $K,\gamma$.
At fixed margin the matching horizon lower bound is attained by
memory-one qubit unitary processes whose responses are zero or one
along the paths they witness.
\end{theorem}
\begin{proof}
Corollary~\ref{low:fixed-memory-joint} combines the signed-phase tree,
probability grids, and the physical cover of
Proposition~\ref{up:physical-cover}.
\end{proof}

The upper bound covers the initial state and each transition channel
at resolution proportional to $\gamma/(K+1)$, using telescoping to bound
every tester's response error so that two tree paths require distinct
cover centers at their first divergence. For fixed $d,r$, the
logarithm of their number has order $K\log((K+1)/\gamma)$
(Proposition~\ref{up:physical-cover}).

The unitary at transition $t$ in the lower construction is
$U_t(\theta_t)=\diag(1,e^{i\theta_t})$, where
$\theta=(\theta_1,\ldots,\theta_K)$ is the target's fixed unknown phase vector.
At each tree node, using only earlier response labels, the tester
chooses the coefficient vector $s\in\{-1,0,1\}^K$ and the final
measurement phase $\phi$. A single control qubit retained by the tester
implements
\begin{equation}
 f_\theta(E_{s,\phi})
 =\frac{1+\cos(\sum_ts_t\theta_t-\phi)}2.
 \label{main:phase-query}
\end{equation}
After initializing the control in $|+\rangle$, the tester realizes each
$s_t\in\{-1,0,1\}$ by choosing a conditional computational-basis preparation.
It sends the input through transition $t$ once in the forward direction and
reverses the known preparation after receiving the output, adding $s_t\theta_t$
to the retained control's relative phase. The final accepting projection
onto $(|0\rangle+e^{i\phi}|1\rangle)/\sqrt2$ gives the response as a function
of the accumulated phase relative to $\phi$.
Figure~\ref{main:phase-figure} links this control to a complete response path.

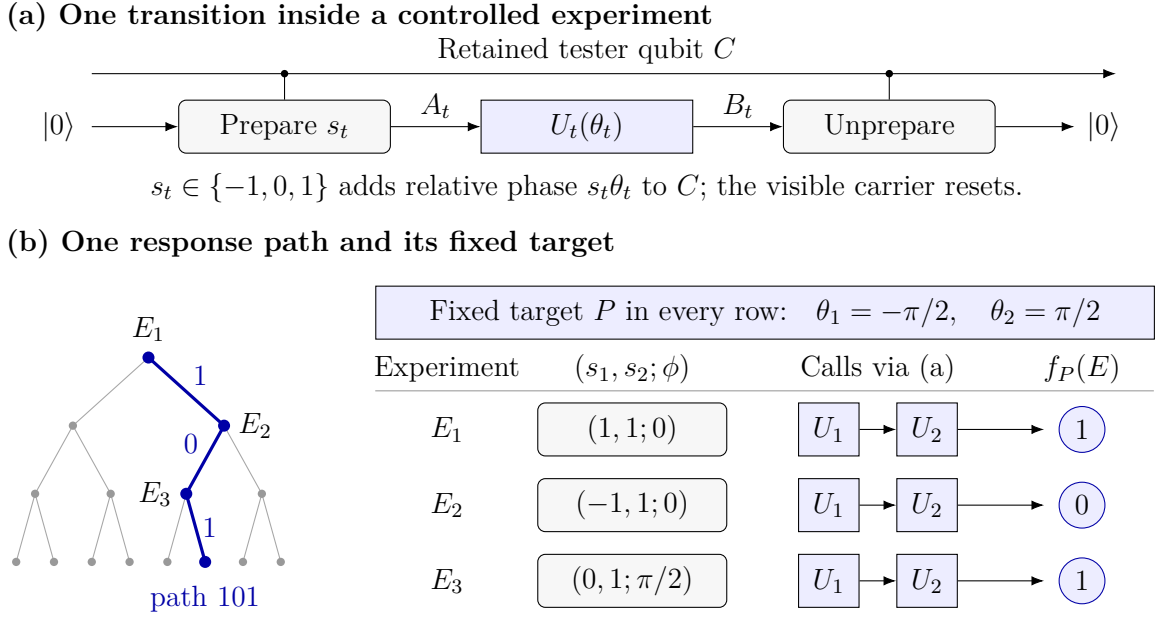
\begin{figure}[t]
\centering
\begin{tikzpicture}[x=1cm,y=1cm,>=Latex,
 font=\fontsize{11}{13}\selectfont,
 known/.style={draw,rounded corners=1mm,fill=black!3,align=center,minimum height=7mm},
 process/.style={draw,fill=blue!7,align=center,minimum height=7mm},
 dot/.style={circle,fill=black!40,inner sep=1.1pt},
 selected/.style={circle,fill=blue!65!black,inner sep=1.6pt}]
\node[anchor=west,font=\bfseries\fontsize{11}{13}\selectfont] at (0,8.1)
 {(a) One transition inside a controlled experiment};
\node[known,minimum width=2.8cm] (prepare) at (3.8,6.65) {Prepare $s_t$};
\node[process,minimum width=2.8cm] (transition) at (7.8,6.65) {$U_t(\theta_t)$};
\node[known,minimum width=2.8cm] (unprepare) at (11.8,6.65) {Unprepare};
\node[anchor=east] at (1.15,6.65) {$|0\rangle$};
\draw[->] (1.25,6.65)--(prepare.west);
\draw[->] (prepare.east)--node[above] {$A_t$}(transition.west);
\draw[->] (transition.east)--node[above] {$B_t$}(unprepare.west);
\draw[->] (unprepare.east)--(14.25,6.65);
\node[anchor=west] at (14.3,6.65) {$|0\rangle$};
\draw[->] (1.25,7.35)--(14.8,7.35);
\fill (3.8,7.35) circle (1.4pt);\draw (3.8,7.35)--(prepare.north);
\fill (11.8,7.35) circle (1.4pt);\draw (11.8,7.35)--(unprepare.north);
\node[fill=white,inner sep=2pt] at (7.8,7.65) {Retained tester qubit $C$};
\node[align=center] at (7.8,5.85)
 {$s_t\in\{-1,0,1\}$ adds relative phase $s_t\theta_t$ to $C$. The visible carrier resets.};

\node[anchor=west,font=\bfseries\fontsize{11}{13}\selectfont] at (0,5.1)
 {(b) One response path and its fixed target};
\coordinate (root) at (2,3.6);
\coordinate (v0) at (1,2.7);\coordinate (v1) at (3,2.7);
\coordinate (v00) at (.5,1.8);\coordinate (v01) at (1.5,1.8);
\coordinate (v10) at (2.5,1.8);\coordinate (v11) at (3.5,1.8);
\coordinate (v000) at (.25,.9);\coordinate (v001) at (.75,.9);
\coordinate (v010) at (1.25,.9);\coordinate (v011) at (1.75,.9);
\coordinate (v100) at (2.25,.9);\coordinate (v101) at (2.75,.9);
\coordinate (v110) at (3.25,.9);\coordinate (v111) at (3.75,.9);
\foreach \a/\b in {root/v0,root/v1,v0/v00,v0/v01,v1/v10,v1/v11,
 v00/v000,v00/v001,v01/v010,v01/v011,v10/v100,v10/v101,v11/v110,v11/v111}
 \draw[black!35] (\a)--(\b);
\foreach \n in {root,v0,v1,v00,v01,v10,v11,v000,v001,v010,v011,v100,v101,v110,v111}
 \node[dot] at (\n) {};
\draw[blue!65!black,very thick] (root)--node[above right,inner sep=2pt] {$1$}(v1)
 --node[above left,inner sep=2pt] {$0$}(v10)--node[right,inner sep=2pt] {$1$}(v101);
\foreach \n in {root,v1,v10,v101}\node[selected] at (\n) {};
\node[above=1mm] at (root) {$E_1$};
\node[right=1.5mm,fill=white,inner sep=1pt] at (v1) {$E_2$};
\node[left=1.5mm,fill=white,inner sep=1pt] at (v10) {$E_3$};
\node[below=2mm,blue!65!black] at (v101) {path $101$};
\node[anchor=north west,process,minimum width=10.3cm] at (5,4.55)
 {Fixed target $P$ in every row:\quad $\theta_1=-\pi/2,\quad\theta_2=\pi/2$};
\node at (5.95,3.45) {Experiment};
\node at (8.4,3.45) {$(s_1,s_2;\phi)$};
\node at (11.65,3.45) {Calls via (a)};
\node at (14.35,3.45) {$f_P(E)$};
\draw[black!40] (5,3.16)--(15.3,3.16);
\foreach \y/\run/\settings/\answer in
 {2.65/1/{(1,1;0)}/1,1.65/2/{(-1,1;0)}/0,0.65/3/{(0,1;\pi/2)}/1}{
 \node at (5.95,\y) {$E_{\run}$};
 \node[known,minimum width=2.5cm] at (8.4,\y) {$\settings$};
 \node[process,minimum width=8mm] (first\run) at (11,\y) {$U_1$};
 \node[process,minimum width=8mm] (second\run) at (12.3,\y) {$U_2$};
 \draw[->] (first\run)--(second\run);
 \draw[->] (second\run.east)--(13.85,\y);
 \node[circle,draw=blue!65!black,fill=blue!7,inner sep=3pt] at (14.35,\y) {$\answer$};
}
\end{tikzpicture}
\caption{A controlled query and a complete response path for a process
with no hidden memory. In (a), known
tester operations select the signed phase in Eq.~\eqref{main:phase-query}
using one forward call to the process transition. The visible carrier
starts and ends in $|0\rangle$. The retained qubit $C$ starts in $|+\rangle$
once per complete run and is measured after all transitions, with accepting
projector onto $(|0\rangle+e^{i\phi}|1\rangle)/\sqrt2$.
In (b), the highlighted
leaf has responses $1,0,1$ at threshold $1/2$ and is witnessed by one
fixed phase pair. Each row is a fresh complete controlled experiment
with the same $U_1,U_2$, using the controls in (a) at every call and
settings determined by the earlier responses.}
\label{main:phase-figure}
\end{figure}

For a two-transition example, set
$\alpha=\theta_1+\theta_2\in\{0,\pi\}$ and
$\beta=\theta_2\in\{0,\pi/2,\pi,3\pi/2\}$ independently.
After the signed query $(1,1)$ reveals the bit in $\alpha$, the query
$(-1,1)$ uses the now known offset $-\alpha$ to measure $\cos(2\beta)$
and reveal the parity $b$ of $2\beta/\pi$ through acceptance probability
$1-b$.
Finally $(0,1)$, with the now known offset $\pi b/2$,
reveals the remaining bit of $\beta$. All three responses are exact zero
or one, and no transition is used twice in any query.
For longer horizons, Boolean inversion converts independent dyadic
subset sums into a fixed phase vector, and signed queries isolate each
coordinate using known offsets from proper subsets already read.
For $K=2^\ell$ with $\ell\in\mathbb Z_{\geq0}$, this gives exactly
$K(1+\ell/2)$ independent bits, with padding handling other horizons. Phase resolution
increases with $K$ (Remark~\ref{low:finite-phase-alphabet}). Appendix~\ref{low:memoryless-section}
gives the construction using the allowed coefficients and proves that
each leaf has one fixed process realizing its entire path.

The same horizon order can be attained without an additional private
qubit: the tester coherently recycles the visible interface and uses
known $X$ gates to implement signs $\pm1$ around forward phase calls.
In contrast, repeating one unknown phase for all $K$ transitions has
capacity $\Theta_\gamma(\log(K+1))$. Both statements are proved in
Appendix~\ref{temp:section} and distinguish independent time-dependent
parameters from coherent access to the interface as resources.

The phase tree organizes responses across fresh runs and already has
superlinear depth at memory one. Information retained after a single run
obeys a different bound: if a finite classical label selects the target
independently of the tester's initial state, the mutual information between
that label and the tester's final quantum--classical register after one
complete run is at most \mbox{$2(K+1)\log_2d$} bits
(Proposition~\ref{inf:trajectory}).

\section{Classical comparison and prediction}
\label{main:prediction}

Fix bases for the visible system and the internal register, and let
$\mathcal C_{d,r}(K)$ denote the processes with a diagonal initial state
whose transitions measure the input--memory pair in these bases and sample
a new output--memory pair from an arbitrary stochastic kernel. These kernels
may vary with time and have arbitrary real probabilities, with all persistent
classical records fitting within the $r$-state internal register. The tester
retains the full quantum permissions of Section~\ref{main:model}.
Proposition~\ref{pred:classical-capacity} gives, for fixed $d,r$ and
$0<\gamma\leq1$,
\begin{equation}
 \operatorname{sfat}_\gamma\mathcal C_{d,r}(K)
 =\Theta_{d,r,\gamma}(K),
 \label{main:classical-order}
\end{equation}
compared with $\Theta_{d,r,\gamma}(K\log(K+1))$ for the quantum
class, whose lower bound uses $r=1$. Each measured stochastic transition
is an exact convex mixture of finitely many deterministic tables,
which provides its linear information bound.

\paragraph*{A common information bound.}
The classical comparison and the noisy-family upper bounds use the
same encoding argument. Suppose every target $P$ in a response class
$\mathcal F$ has a state $\sigma_P$ on a common $D$-dimensional space,
and every tester $E$ has a target-independent effect $0\leq A_E\leq I$.
If the response error satisfies $|f_P(E)-\Tr(A_E\sigma_P)|\leq\eta/2$
for every $P,E$, with $0\leq\eta<\gamma$, then
\begin{equation}
 \operatorname{sfat}_\gamma\mathcal F
 \leq\left\lfloor\frac{\log_2D}{c_{\gamma-\eta}}\right\rfloor,
 \qquad c_z=1-h_2((1-z)/2),
 \label{main:program-bound}
\end{equation}
where $h_2(u)=-u\log_2u-(1-u)\log_2(1-u)$ and $0\log_2 0=0$.
Each node uses its own effect on the same complete-target encoding,
including when testers depend on earlier labels.
The serial information argument~\cite{aaronson2019online}, related
to random-access coding bounds~\cite{nayak1999optimal}, supplies
one conditional information increment per tree level, as proved in
Lemma~\ref{up:information}. This auxiliary program storage is separate from
the process's propagated memory and the tester's workspace.

\paragraph*{Probability feedback.}
The exact response tree also has a scalar prediction interpretation with
unrestricted computation and one target fixed before learning begins.
In round $t$, a learner receives the specification of a tester $E_t$,
whose choice may depend on earlier scalar feedback, and predicts a scalar
$\widehat y_t$. The learner then receives a scalar $b_t$ satisfying
$|b_t-f_P(E_t)|\leq\varepsilon/8$, even when the prediction is correct,
with no additional target queries, internal tester records or retained
quantum systems. Let $M_\varepsilon$ count rounds
with $|\widehat y_t-f_P(E_t)|>\varepsilon$ among $N$ rounds.
For every fixed $d,r,\varepsilon$, Theorem~\ref{pred:mistake-law}
proves that the minimax expected count is
\begin{equation}
 \Theta_{d,r,\varepsilon}\!\left(
       \min\{N,K\log(K+1)\}\right),
 \qquad 0<\varepsilon<1/2,
 \label{main:prediction-order}
\end{equation}
using a finite physical cover and median prediction for the upper bound.
For the lower bound, a uniform leaf of the exact phase tree fixes the
target before the first prediction, while each next response remains
a fair bit conditional on the previous responses. This argument
also covers randomized learners and exact probability feedback.
Appendix~\ref{pred:section} counts the complete process runs needed to
estimate these feedback probabilities with simultaneous confidence,
separately from the mistake count.

\paragraph*{Bernoulli feedback: one outcome per run.}
Fix one target and perform one fresh complete run per round, revealing
one binary outcome. Experiment selection and learning use only the public
classical history of announced experiments, predictions and outcomes,
and randomness independent of the target.
For $N$ rounds, Corollary~\ref{pred:bernoulli-risk} gives the minimax
expected average squared error relative to the true acceptance probabilities,
\[
 \Theta_{d,r}\!\left(\min\{1,K\ln(K+N)/N\}\right).
\]
Multiplication by $N$ gives the cumulative squared error. The experiment
uses $N$ complete process runs and $KN$ unknown transitions.

\section{Known noise in a classical-address family}
\label{main:noise}

To study how noise changes the phase-tree mechanism, we consider known
independent Pauli errors after each phase, with a measured classical
address retained by the process and ideal tester controls. The dephasing
specialization admits a finite exact classical simulation at positive strength.

Let $h\geq0$ and $T\geq1$ be integers, $R=2^h$, and $K=h+T$.
Initialize the visible qubit and the memory in $|0\rangle$ and the
zero $h$-bit string. At loading transition $\ell\leq h$, measure the visible input in the
computational basis, overwrite address bit $\ell$ with the outcome,
and return the visible state $|0\rangle$.
On the entire input algebra this is the channel
\begin{equation}
 \mathcal L_\ell(X)=\sum_{a\in\{0,1\}^h}\sum_{b=0}^1
 \langle b,a|X|b,a\rangle\,
 |0,w_\ell(a,b)\rangle\langle0,w_\ell(a,b)|,
 \label{main:noisy-loader}
\end{equation}
where $w_\ell(a,b)$ overwrites bit $\ell$ of $a$. Each address is
therefore accessible by supplying its $h$ ordinary basis bits in $h$ transitions.

\paragraph*{Capacity under known Pauli noise.}
With this initialization and loader, let
$q=(q_I,q_X,q_Y,q_Z)$ be a known probability vector. With
$U_\theta=\diag(1,e^{i\theta})$ and the Pauli matrices $X,Y,Z$,
replace each local phase channel by
\begin{equation}
 \mathcal E_{\theta,q}(\rho)
 =\sum_{P\in\{I,X,Y,Z\}}q_P P U_\theta\rho U_\theta^*P^*.
 \label{main:pauli-channel}
\end{equation}
Each target is specified by a phase table $\Theta=(\theta_{tj})$, with
$1\leq t\leq T$ and $j\in\{0,1\}^h$, whose entries vary independently
across targets, while the noise vector $q$ and loading maps remain fixed
and known. The number $RT$ counts phase parameters, $R$ is the propagated
memory dimension, and $h=\log_2R$ is its bit count.

Phase slot $t$ is overall transition $h+t$, at which the full addressed
channel has Kraus operators
$\sqrt{q_P}P U_{\theta_{tj}}\otimes|j\rangle\langle j|$, measuring and
preserving the address while independent noise environments are discarded
at each slot. Every experiment includes the $h$ loading and $T$ phase
transitions, after which the final address is discarded. Write
$\mathcal N^{\mathrm P}_{R,T,q}$ for this initialized family, with
one fixed table per target and all normalized testers in
Section~\ref{main:model}.

Define the residual phase-flip probability
\begin{equation}
 e(q)=\min(q_I,q_Z)+\min(q_X,q_Y).
 \label{main:pauli-residual}
\end{equation}

\begin{theorem}[Capacity under known Pauli noise]
\label{pauli:law}
For every fixed $0<\gamma\leq1/16$, the initialized family $\mathcal N^{\mathrm P}_{R,T,q}$ defined above satisfies
\begin{equation}
 \operatorname{sfat}_\gamma\mathcal N^{\mathrm P}_{R,T,q}
 =\Theta_\gamma\!\left(
 RT\log_2\bigl[1+\min\{T,1/e(q)\}\bigr]\right)
 \label{pauli:matching}
\end{equation}
simultaneously for all $R=2^h$, integers $T\geq1$, and known $q$
with $0\leq e(q)\leq1/8$. The constants are independent of $R,T,q$,
and $1/0=\infty$.
\end{theorem}

\begin{proof}
Appendix~\ref{pauli:section} proves the upper bound by a phase cover,
combined with an exact quantum program when $e(q)>0$.
The lower bound follows by syndrome recovery followed
by the noisy Boolean phase trees of Appendix~\ref{noise:section}.
\end{proof}

The tester uses known controlled-NOT and sign controls to
encode its private qubit jointly with a visible carrier initialized in
$|0\rangle$, and reverses that encoding after the transition.
After this reversal, measuring the carrier in the computational basis
distinguishes the error pairs
$\{I,Z\}$ and $\{X,Y\}$, allowing recovery to correct the more probable
phase error in each pair so that the recovered control qubit has phase-flip
probability exactly $e(q)$ and realizes short-block tests with coherence factor
$1-2e(q)$ per active transition.

For the lower bound, park the control during inactive slots and extract
Boolean phase digits within short coherent blocks, which retain a fixed
margin up to the inverse residual phase-error scale. Concatenating their
trees over all addresses gives the lower bound in Theorem~\ref{pauli:law}.
Appendix~\ref{noise:sampling} gives the response-estimation costs.

For pure transverse flips, $q_X=a$, $q_I=1-a$ and $q_Y=q_Z=0$,
one has $e(q)=0$ at every $0\leq a\leq1$: the logarithmic enhancement
survives. Fixed positive $e(q)$ truncates the logarithmic enhancement, so the
relevant quantity in this family is the phase noise remaining after the
specified correction.

For the upper bound with $e(q)>0$, a doubled-angle phase program supplies
an auxiliary information encoding using at most
$\lceil\log_2(1/e(q))\rceil$ program qubits per table entry. The simulator
keeps every failure branch and allocates it to a Pauli pair, obtaining an
exact deterministic simulation as an identity of complete channels.
The complete-table program supplies the target encoding in
Eq.~\eqref{main:program-bound} and gives the capacity bound for all adaptive
testers. Appendix~\ref{pauli:section} gives
the circuit and processor.

For the single-address family $R=1$, $K=T$, with known
$q=(1-2p,p,p,0)$ and fixed $0<p\leq1/8$, the capacity is
$\Theta_{\gamma,p}(T)$ at fixed $0<\gamma\leq1/16$.
To simulate this family in the complete-process dictionary model of
Appendix~\ref{cost:section}, a single random classical label is supplied
once at the start to select a fixed causal quantum simulator that may use
arbitrary memory and correlations between slots. At fixed total response
error $0<\epsilon\leq1/64$, uniformly over normalized adaptive testers,
Theorem~\ref{pauli:whole-program-cost} shows that such a classical program
requires $\Theta(T\log T)$ bits as $T$ grows, whereas an exact quantum
program uses $O_p(T)$ qubits.

\paragraph*{Dephasing as a classical-program example.}
For $q_Z=p/2$, $q_I=1-p/2$ and $q_X=q_Y=0$, with
$0\leq p\leq1/4$, the residual probability is $e=p/2$.
For $p>0$, this specialization has an exact finite classical simulation:
a regular polygon with $\lceil\pi/\arccos(1-p)\rceil$ unitary vertices
contains the coherence circle of radius $1-p$, and independent vertex
choices for the $RT$ entries give a common classical program for the
complete process. Because the channel identity holds on reference-entangled
inputs, this program gives a capacity upper bound valid for all adaptive
testers. Appendix~\ref{noise:section} gives the explicit constants and
finite covering and program estimates.

\begin{table}[t]
\caption{Capacity of the measured-address family under known Pauli noise,
at fixed $0<\gamma\leq1/16$ and residual probability $0\leq e(q)\leq1/8$.
Each experiment uses $K=h+T$ transitions. The first two rows have constants
uniform in the allowed noise vector. The last row fixes $e(q)>0$, and its
constants may depend on this value.}
\label{main:noisy-regimes}
\centering
\begin{tabular}{ll}
\toprule
Residual noise range & Sequential capacity order\\
\midrule
$e=0$ or $e\leq1/T$ & $RT\log_2(1+T)$\\
$1/T<e\leq1/8$ & $RT\log_2(1+1/e)$\\
Fixed $e>0$, growing $T$ & $RT$\\
\bottomrule
\end{tabular}
\end{table}

\section{Growing memory and the saturated endpoint}
\label{main:memory-regimes}

The fixed-memory law of Theorem~\ref{main:fixed-law} allows its constants
to depend on $r$, so we consider separately the case $r=r(K)$, keeping
the selected memory cap at every transition. Table~\ref{main:regimes}
compares the fixed-memory law with two growing-memory regimes.

For qubit processes with $r=K^2$, Corollary~\ref{low:quadratic-memory}
gives, for every integer $K\geq8$ and $0<\gamma\leq1/16$,
\[
 C_\gamma(K;2,K^2)
 \geq\frac{K^4}{4096\sqrt2\,
                    \lceil\log_2(K+1)\rceil^2}.
\]
The unitary-column and addressing constructions in
Appendix~\ref{main:growing} give this lower bound. At fixed
$0<\gamma\leq1/16$, its growth along $r=K^2$ rules out any bound of
the form $CKr[\log(K+1)+\log(r+1)]^a$ that holds uniformly in $K,r$
for $d=2$, with constants $C,a\geq0$: such a bound grows as
$O(K^3\log^aK)$ on this slice.

For every such fixed $\gamma$, the quantum-program bound of
Theorem~\ref{up:optimized-program} gives
$C_\gamma(K;2,K^2)=O_\gamma(K^{11/2}\log K)$ using a simulator that
encodes the initial state and all $K$ transition channels in a common
program space with response error controlled uniformly over adaptive
testers. The common information bound in Eq.~\eqref{main:program-bound}
then converts the program dimension into the capacity estimate. Appendix~\ref{main:upper} gives the
finite program parameters and the complete simulation proof.

To obtain an upper bound valid at every memory cap,
Appendix~\ref{up:causal-section} uses the normalized Choi representation
of the full causal class as a convex set of positive trace-one matrices
satisfying linear causal constraints and applies a log-determinant
potential on that set.
For $d\geq2$, $K\geq1$,
$r\geq d^{2K-1}$ and $0<\gamma\leq1$, Theorem~\ref{sat:matching} gives
\[
 d^{2K}\lfloor\log_2d\rfloor
 \leq C_\gamma(K;d,r)
 \leq\left\lfloor
 \frac{4d^{2K+1}\ln(4/\gamma)}{\gamma^2}
 \right\rfloor.
\]
The dense-table construction attains the lower bound, and
Lemma~\ref{sat:realization} shows that memory dimension $d^{2K-1}$
suffices to realize every causal process on this interface.
For fixed $d$ and $\gamma$, the capacity therefore has horizon
order $\Theta_{d,\gamma}(d^{2K})$ in this regime.

\begin{table}[!htb]
\caption{Capacity regimes for fixed visible dimension $d\geq2$.
Constants in the first row may depend on $d,r$. The other rows fix
the margin. The quadratic-memory lower bound requires $d=2$ and
$0<\gamma\leq1/16$. The last row gives a sufficient memory dimension for saturation.}
\label{main:regimes}
\centering
\begin{tabular}{lll}
\toprule
Regime & Sequential capacity & Status\\
\midrule
Fixed $r$ & $\Theta_{d,r}(K\log((K+1)/\gamma))$ & Joint $K,\gamma$ order\\
$d=2$, $r=K^2$ & $\Omega_\gamma(K^4/\log^2K)$ to $O_\gamma(K^{11/2}\log K)$ & Gap remains\\
$r\geq d^{2K-1}$ & $\Theta_{d,\gamma}(d^{2K})$ & Fixed-margin horizon order\\
\bottomrule
\end{tabular}
\end{table}

\section{Discussion}
\label{main:discussion}

At fixed resolution, the sequential response capacity can have order
$K\log K$ even without hidden process memory, as the phase construction
shows using independent temporal parameters and coherent access to the
visible interface. This realizes the growth allowed by fixed internal
memory, while the linear bound for fully measured processes identifies
the role of coherent access. The two prediction protocols translate the
quantum capacity and covering bounds at fixed memory into losses under
different feedback interfaces. In the specified noisy address family, short
coherent blocks give the lower bound. A phase cover and, for $e>0$, an
exact complete-process program give the matching upper bound governed by
the residual phase-error timescale.

The optimal joint dependence on duration,
memory and margin, and the least memory attaining saturation,
remain further questions. For the dephasing law in
Theorem~\ref{main:noisy-law}, the optimal small-margin dependence
also remains open.

\begin{remark}[Information and computation]
Open computational questions concern program preparation, constrained
optimization and learning from an experimental record. The computational complexity of proper prediction, which requires a
predicted response function in the process class, is a further question
for specified experiment and loss restrictions.
Ref.~\cite{chittoor2022online} studies online program optimization for
changing noisy channels, and Ref.~\cite{cotler2026noisy} studies noisy
quantum learning with ancillary learner memory.
\end{remark}

\section*{Data Availability}
The proofs are contained in this article and its appendices. No experimental data was generated.

\begin{acknowledgments}
Yibin Wang is the sole author. OpenAI Codex assisted with language editing, formatting, diagram preparation, technical checks and improvement, and code for numerical experiments and counterexample searches. The author takes responsibility for the content of this work.
\end{acknowledgments}

\appendix
\section{Perfect-margin trees without process memory}
\label{low:memoryless-section}

Using coherent control across successive transitions to realize bounded
signed phase queries, we first construct a family with one-dimensional
process memory whose sequential fat dimension is of order $K\log K$,
even at the largest margin.  A second, classical family supplies
the dependence on the margin.  Throughout this section, a query means a fresh
run of the complete initialized process.

For an integer $K\geq 1$, define
\begin{equation}
 L(K)=\sum_{n:\,\operatorname{bit}_n(K)=1}
       2^n\left(1+\frac n2\right),
 \qquad L(0)=0.
 \label{low:binary-depth}
\end{equation}
Here $\operatorname{bit}_n(K)$ is the $n$th digit of the binary expansion of
$K$, indexed from $n=0$ at the least significant bit.

\begin{theorem}[Perfect-margin phase trees]
\label{low:phase-tree}
For integers $d\geq2$, $r\geq1$, and $K\geq1$, and every
$0<\gamma\leq1$,
\begin{equation}
 \operatorname{sfat}_{\gamma}\bigl(\mathcal P_{d,r}(K)\bigr)
 \geq L(K).
 \label{low:phase-lower}
\end{equation}
The lower bound is realized by a finite family of memory-one unitary
processes with exactly $2^{L(K)}$ distinct response functions.  Its sequential
fat dimension is exactly $L(K)$ at every such margin. Each witness gives
zero-or-one responses along its path. For $K=2^n$ with $n\in\mathbb Z_{\geq0}$,
the depth is $K(1+n/2)$.
\end{theorem}

\begin{proof}
We begin with qubit transitions
\begin{equation}
 U_t(\theta_t)=\operatorname{diag}(1,e^{i\theta_t}),
 \qquad t=1,\ldots,K,
 \label{low:phase-process}
\end{equation}
and fixed initial output $|0\rangle$. The phases are part of the fixed
transition descriptions, and the process propagates no private register.

For any $s\in\{-1,0,1\}^K$ and any known phase $\phi$, a causal tester can
produce the response
\begin{equation}
 f_{\theta}(E_{s,\phi})
 =\frac{1+\cos(\sum_t s_t\theta_t-\phi)}2.
 \label{low:signed-response}
\end{equation}
To see this, prepare a private control qubit in $|+\rangle$ and use temporary visible carriers belonging to the tester. Before transition $t$, prepare a fresh visible input conditionally on that control, with the two computational input values on branches $0,1$
chosen as $(0,1)$ if $s_t=1$, $(1,0)$ if $s_t=-1$, and $(0,0)$ if $s_t=0$.
Since these preparations extend to unitaries, after receiving the output
reverse the same preparation to return its carrier to $|0\rangle$ and
accumulate a relative phase $e^{is_t\theta_t}$ on the control.
After all transitions, the control state, up to a global phase, is
\[
 \frac{|0\rangle+e^{i\sum_t s_t\theta_t}|1\rangle}{\sqrt2}.
\]
The effect projecting onto
$(|0\rangle+e^{i\phi}|1\rangle)/\sqrt2$ gives
Eq.~\eqref{low:signed-response}. This experiment uses every transition
once and involves no postselection.

In the dyadic case, index the transitions by subsets $Z\subseteq[n]$.
For every $T\subseteq[n]$, put
\begin{equation}
 \alpha_T=\sum_{Z\supseteq T}\theta_Z.
 \label{low:mobius-forward}
\end{equation}
These coordinates are independent: their inverse transformation is
\begin{equation}
 \theta_Z=\sum_{T\supseteq Z}(-1)^{|T|-|Z|}\alpha_T.
 \label{low:mobius-inverse}
\end{equation}
Indeed, substitution into Eq.~\eqref{low:mobius-forward} gives the
coefficient
$\sum_{T\subseteq Z\subseteq U}(-1)^{|U|-|Z|}
=(1-1)^{|U|-|T|}$ for $\alpha_U$.  Choose independently
\begin{equation}
 m_T\in\{0,\ldots,2^{|T|+1}-1\},
 \qquad \alpha_T=\frac{2\pi m_T}{2^{|T|+1}}.
 \label{low:phase-labels}
\end{equation}
Each complete choice fixes the whole process through
Eq.~\eqref{low:mobius-inverse} and encodes
\begin{equation}
 M_n=\sum_{T\subseteq[n]}(|T|+1)
     =2^n+n2^{n-1}=K\left(1+\frac n2\right)
 \label{low:phase-bit-count}
\end{equation}
bits. For $n=0$ the sum is one.

To obtain dyadic multipliers with physical coefficients in $\{-1,0,1\}$,
we use signed sums whose proper-subset terms become known offsets when
smaller subsets are read first. With $z_i(Z)=\mathbf 1_{\{i\in Z\}}$,
define the coefficients for $J\subseteq T$ by
\begin{equation}
 q_{T,J}(Z)=
 \prod_{i\in J}(2z_i(Z)-1)
 \prod_{i\in T\setminus J}z_i(Z).
 \label{low:signed-polynomial}
\end{equation}
Expanding the first product and using
Eq.~\eqref{low:mobius-forward}, with $j=|J|$, gives
\begin{align}
 &\sum_Zq_{T,J}(Z)\theta_Z=2^j\alpha_T+\beta_{T,J},
 \label{low:dyadic-isolation}\\
 &\beta_{T,J}
 =\sum_{A\subsetneq J}(-1)^{j-|A|}2^{|A|}
       \alpha_{(T\setminus J)\cup A}.
 \label{low:known-phase}
\end{align}
Every index in the last sum is a proper subset of $T$, and the sum is zero
when $J$ is empty.

Order the sets $T$ by increasing cardinality, with a fixed order within
each cardinality, and in the block for $T$ reveal the bits of $m_T$ from
least to most significant.  At the node for bit $k\in\{0,\ldots,|T|\}$,
the preceding branches specify $r_k=m_T\bmod2^k$ and all coordinates
indexed by proper subsets of $T$. Choose a fixed subset
$J\subseteq T$ with $|J|=|T|-k$ and define the tester
$E_{q_{T,J},\phi}$ using only these preceding branches by setting
\begin{equation}
 \phi=\beta_{T,J}+\frac{2\pi r_k}{2^{k+1}}.
 \label{low:node-phase}
\end{equation}
If
$m_T=r_k+2^k b_k+2^{k+1}h$, with $b_k\in\{0,1\}$ and
$h\in\mathbb Z_{\geq0}$, then
\begin{equation}
 \sum_Zq_{T,J}(Z)\theta_Z-\phi
 =\pi b_k\pmod{2\pi},
 \qquad f_\theta(E_{q_{T,J},\phi})=1-b_k.
 \label{low:perfect-bit}
\end{equation}
Use threshold $1/2$, assigning the positive branch to $b_k=0$ and giving signed response margin $1/2$.
Every complete branch assignment specifies all integers $m_T$, hence one
fixed phase vector for which Eq.~\eqref{low:perfect-bit} holds at every
node of the complete path.

For a general horizon, partition its transitions into the distinct dyadic
blocks appearing in the binary expansion of $K$. With the parameters of
different blocks chosen independently, query each block with coefficient
zero at all other transitions and concatenate these block trees to obtain
depth $L(K)$. For $d>2$, extend every qubit unitary by the identity on the
remaining levels and keep the tester inside its qubit subspace.

The construction using all dyadic blocks has $2^{L(K)}$ complete parameter
assignments, whose response functions are distinct because any two paths
first diverge at a common tester and threshold. Conversely, a
positive-margin tree of depth $m$ requires at least $2^m$ different
response functions by the same first-divergence argument. This proves the
asserted exact dimension of the finite family.

Alternatively, keeping only the largest dyadic block and padding by
identity transitions gives the same asymptotic order, with depth
$N(1+\log_2N/2)$, where $N=2^{\lfloor\log_2K\rfloor}$.
\end{proof}

\begin{remark}
\label{low:finite-phase-alphabet}
The construction uses fixed, unrestricted transition descriptions,
with gate synthesis cost treated as a separate resource.
For a dyadic block of length $K$, every phase in
Eq.~\eqref{low:phase-labels} has denominator dividing $2K$ in units of
$2\pi$, so Eq.~\eqref{low:mobius-inverse} makes every physical
phase a multiple of $\pi/K$, modulo $2\pi$. Thus a finite set of at most
$2K$ diagonal gates per transition suffices for this family.
\end{remark}

Probability-bisection trees also appear in sequential-fat lower bounds
for pure-state measurements~\cite{meyer2025pure}.  The finite grid below
gives an exact count for independent outputs of an initialized process.

\begin{remark}
\label{low:probability-semantics}
The grid compares response probabilities with thresholds. A single
Bernoulli measurement gives a binary outcome, while the small-margin
levels measure separation between probability parameters.
The exact zero-or-one responses in Theorem~\ref{low:phase-tree} are a
stronger property of that particular family.
\end{remark}

\begin{proposition}[Probability-grid trees]
\label{low:grid}
For $d\geq2$, $r\geq1$, $K\geq1$, and $0<\gamma\leq1$,
\begin{equation}
 \operatorname{sfat}_{\gamma}\bigl(\mathcal P_{d,r}(K)\bigr)
 \geq (K+1)\left\lfloor\log_2(1+1/\gamma)\right\rfloor.
 \label{low:grid-lower}
\end{equation}
The right-hand side is the exact sequential fat dimension of a finite
family of memory-one replacer processes.
\end{proposition}

\begin{proof}
Let $B=\lfloor\log_2(1+1/\gamma)\rfloor$ and $Q=2^B\geq2$.
Set $p_\ell=\ell/(Q-1)$ for $\ell\in\{0,\ldots,Q-1\}$ and
\[
 \rho(p)=(1-p)|0\rangle\langle0|+p|1\rangle\langle1|.
\]
Choose $K+1$ independent labels $\ell_0,\ldots,\ell_K$ and emit
$\rho(p_{\ell_0})$ initially. At transition $t$, discard the input
and emit $\rho(p_{\ell_t})$, so these states and replacer channels define
a memory-one process on the full $d$-dimensional input spaces.
A tester can measure the effect $|1\rangle\langle1|$ on any selected
output, retain its binary result privately, and complete the remaining
transitions, giving exactly that coordinate's probability as its response.

For one coordinate, construct a depth-$B$ tree by repeatedly bisecting
the ordered label set, placing each node's threshold halfway between the
largest probability in the lower half and the smallest probability in
the upper half. Adjacent probabilities differ by
$1/(Q-1)\geq\gamma$, so every response in either half has the required
margin $\gamma/2$. Concatenating these trees over all $K+1$ coordinates
fixes all labels on each complete path and hence gives one process satisfying every
inequality on that path. The family has exactly $Q^{K+1}$ distinct
response functions, and the first-divergence count gives exactness.
\end{proof}

\begin{corollary}[Joint horizon and margin order at fixed memory]
\label{low:fixed-memory-joint}
For integers $d\geq2$, $r\geq1$, $K\geq1$, and $0<\gamma\leq1$,
\begin{equation}
 \frac K8\log_2\frac{K+1}{\gamma}
 \leq\operatorname{sfat}_{\gamma}\bigl(\mathcal P_{d,r}(K)\bigr)
 \leq5\bigl[(dr)^4-1\bigr]K\log_2\frac{K+1}{\gamma}.
 \label{low:joint-lower}
\end{equation}
\end{corollary}

\begin{proof}
Let $N=2^n\leq K<2^{n+1}$.  Theorem~\ref{low:phase-tree} gives
$\operatorname{sfat}_\gamma\geq N(1+n/2)
\geq(K/4)\log_2(K+1)$, since $N>K/2$ and
$\log_2(K+1)\leq n+1$.
For $y=\log_2(1/\gamma)\geq0$,
$\lfloor\log_2(1+2^y)\rfloor\geq y/2$:
when $y\leq2$, its value is at least one. When $y\geq2$, it is at least
$y-1$.  Proposition~\ref{low:grid} therefore also gives
$\operatorname{sfat}_\gamma\geq Ky/2$.
The capacity is at least the maximum of these two lower bounds and
hence at least half their sum, proving the lower bound in
Eq.~\eqref{low:joint-lower} using two subfamilies of the same process class.

For the upper bound, the physical cover of
Proposition~\ref{up:physical-cover}, with $a=dr$, gives
\[
 \operatorname{sfat}_\gamma
 \leq\left\lfloor
 [a^2-1+K(a^4-a^2)]\log_2\bigl(1+8(K+1)/\gamma\bigr)
 \right\rfloor.
\]
For $K\geq1$, the bracket is at most $K(a^4-1)$.
For $x=(K+1)/\gamma\geq2$,
$\log_2(1+8x)\leq\log_2(9x)\leq5\log_2x$.
This proves the upper bound in Eq.~\eqref{low:joint-lower} with constants independent
of the margin.
\end{proof}

\section{Physical covers and common program bounds}
\label{up:section}

The upper bounds use two finite descriptions of a process. A physical cover
approximates the initial state and each transition separately, supplying
the fixed-memory upper bound and the prediction guarantees. A quantum
program represents the same complete process by a state on which every
tester induces an ordinary binary measurement, yielding a common-program
inequality that applies to the classical comparison and both noisy phase
families.

Throughout this section, $d\geq2$, $r\geq1$, and $K\geq1$ are integers, and
$0<\gamma\leq1$. We write $n=dr$ and use base-two logarithms unless
$\ln$ is displayed. The operational distance $D(P,Q)$ is defined in
Eq.~\eqref{main:distance}, with the supremum over all normalized binary
testers with arbitrary finite private workspaces. The physical horizon $K$ is distinct from the
depth of an online tree.

\subsection{A cover with physical centers}

\label{up:cover-origin}

The componentwise covering method is already present in
Ref.~\cite[arXiv v1, Lemma~25]{raza2024online}. For its class of multitime
processes built from $G\geq1$ basic components on $q\geq2$ qubits and
$0<\delta<1$, the strategy-norm covering number satisfies
\[
 N(\delta)\leq \binom{q}{2}^{G}(6G/\delta)^{1024G}.
\]
Its telescoping estimate adds local approximation errors and
controls the probability of every allowed multitime test. Taking
$\delta=\gamma/4$ and counting witnesses at their first divergent
tree node gives, for their response class and $0<\gamma\leq1$,
\[
 \operatorname{sfat}_\gamma
 \leq\left\lfloor G\log_2\binom{q}{2}
       +1024G\log_2(24G/\gamma)\right\rfloor.
\]
For the present memory model, we apply the same method to one initial
state and $K$ arbitrary fixed-dimensional CPTP channels. The construction
below keeps the covering centers inside $\mathcal P_{d,r}(K)$, making
the cover proper, and gives the fixed-memory upper order used in the main law.

To put all realizations on a common memory carrier, choose isometries
$i_j:M_j\longrightarrow\mathbb C^r$, let $P_j=i_ji_j^*$, and embed the initial
state by $I_d\otimes i_0$. For a transition $\Phi_j$, set
$L_j=I_d\otimes i_{j-1}^*$, $O_j=I_d\otimes i_j$, and
$Q_j=I_d\otimes P_{j-1}$. Its extension to an $n$-dimensional input and output is
\begin{equation}
 \widetilde\Phi_j(X)
 =O_j\Phi_j(L_jXL_j^*)O_j^*
   +\operatorname{Tr}[(I-Q_j)X]\tau_j,
 \label{up:padding}
\end{equation}
with any fixed output density matrix $\tau_j$. Both summands are completely positive, and their traces add to
$\operatorname{Tr}X$, so the extension is trace preserving on the whole
input algebra. On the embedded memory sector it agrees with the original map and outputs
in the next embedded sector, with both identities holding also with any
untouched reference. Consequently every process
in $\mathcal P_{d,r}(K)$ has one initial state on $\mathbb C^n$ and $K$
channels from $\mathbb C^n$ to itself, with the prescribed visible-memory
split and final memory discard.

\begin{proposition}[Physical covering bound]
\label{up:physical-cover}
Put
\begin{equation}
 p_{n,K}=(n^2-1)+K(n^4-n^2).
 \label{up:physical-dimension}
\end{equation}
For every $\varepsilon>0$, the class $\mathcal P_{d,r}(K)$ has a proper
$D$-cover of radius $\varepsilon$ satisfying
\begin{equation}
 \log_2\mathcal N(\varepsilon,\mathcal P_{d,r}(K),D)
 \leq p_{n,K}\log_2\!\left(1+\frac{4(K+1)}{\varepsilon}\right).
 \label{up:cover-entropy}
\end{equation}
In particular,
\begin{equation}
 \operatorname{sfat}_\gamma(\mathcal P_{d,r}(K))
 \leq
 \left\lfloor p_{n,K}\log_2\!\left(1+\frac{8(K+1)}\gamma\right)\right\rfloor.
 \label{up:cover-capacity}
\end{equation}
\end{proposition}

\begin{proof}
The affine space of Hermitian trace-one $n\times n$ matrices has real
dimension $n^2-1$. The affine space of Hermiticity-preserving,
trace-preserving maps on $n\times n$ matrices has dimension $n^4-n^2$,
as follows from the Hermitian Choi representation with
$n^4$ real coordinates and $n^2$ independent partial-trace constraints.
Each physical set has diameter at most two when state differences use
the trace norm and channel differences use the diamond norm.

A compact subset of diameter at most two in a real normed affine space of
dimension $a$ has a proper $\delta$-net with at most
$(1+4/\delta)^a$ points. To see this, choose a maximal set of points whose pairwise distances exceed
$\delta$. The balls of radius $\delta/2$ around these points have disjoint
interiors and lie in a ball of radius $2+\delta/2$ around any fixed point
of the set, so comparing the $a$-dimensional volumes proves both the
cardinality bound and finiteness. Maximality then gives the covering property. A zero-dimensional set requires
one point.

Apply the proper-net construction to the initial-state and channel
factors, so every tuple of centers defines a physical process with memory
cap $r$. With $\delta=\varepsilon/(K+1)$, fix a tester and replace the
factors of a target realization by their centers one at a time.
The initial-state replacement contributes at most $\delta$ in trace norm,
and each channel replacement contributes at most $\delta$ by its
diamond-norm bound on the joint visible-memory input, with all other
retained registers as the untouched reference. Because the physical maps
before and after each replacement contract Hermitian trace distance, the
terminal states differ by at most $(K+1)\delta=\varepsilon$ in trace norm,
so a binary outcome probability differs by at most half this amount, proving
the uniform $D$-bound and \eqref{up:cover-entropy}.

Choose one complete process witness for each leaf of a finite
$\gamma$-shattered tree. At the first node where two leaves diverge,
their tester responses differ by at least $\gamma$, so their
$D$-distance is at least $2\gamma$. A ball of radius $\gamma/2$ contains at most one of these witnesses,
so a tree of depth $L$ requires at least $2^L$ cover elements. Substituting $\varepsilon=\gamma/2$ in
\eqref{up:cover-entropy} and taking the integer floor gives the conclusion
for every finite shattered depth and hence for their supremum.
\end{proof}

\subsection{Information carried by one program for a complete leaf}

A common program assigns one state to each complete target, with each
node evaluated separately by a target-independent effect.
Using the binary entropy $h_2$, define the information increment
\begin{equation}
 c_z=1-h_2\!\left(\frac{1-z}{2}\right).
 \label{up:information-increment}
\end{equation}
The next argument is the serial random-access
encoding bound applied to whole-process programs
\cite{aaronson2019online}.

\begin{lemma}[Whole-process program bound]
\label{up:information}
Suppose every process $P$ in a response class has a density matrix
$\sigma_P$ on a common $D_0$-dimensional Hilbert space. Suppose also that,
for some fixed $0\leq\eta<\gamma$, every tester $E$ has an effect
$0\leq A_E\leq I$ on that space, independent of $P$, such that
\begin{equation}
 \sup_E|f_P(E)-\operatorname{Tr}(A_E\sigma_P)|\leq\eta/2
 \label{up:program-accuracy}
\end{equation}
for every $P$ in the class. Then
\begin{equation}
 \operatorname{sfat}_\gamma
 \leq\left\lfloor\frac{\log_2D_0}{c_{\gamma-\eta}}\right\rfloor.
 \label{up:program-information-bound}
\end{equation}
The same bound holds for uniform error at most $\eta/2$ when approximating
by the convex hull of $D_0$ fixed $[0,1]$-valued response functions.
\end{lemma}

\begin{proof}
For a finite depth-$L$ $\gamma$-shattered tree, choose its complete
leaf witnesses $P_x$, $x\in\{0,1\}^L$, and encode each witness once as
$\sigma_x=\sigma_{P_x}$. At every node on that leaf's path,
\eqref{up:program-accuracy} reduces the signed margin by at most
$\eta/2$, so the same thresholds and effects shatter the program
states at scale $z=\gamma-\eta$.

We use $H$ for classical entropy, $S$ for von Neumann entropy, and
$I$ for the corresponding mutual information, all in bits.
Each node uses a separate measurement in the following information
calculation.
Let $X$ be a uniformly distributed leaf label and let $Q$ hold
$\sigma_X$. Conditional on any prefix $X_{<j}=u$, the next bit is
uniform, and the states for its two values are the averages of
$\sigma_x$ over the two child subtrees. Measuring the common node
effect on these averages gives probabilities $p_0,p_1$ with
$p_1-p_0\geq z$. Guessing $X_j$ from the measured binary outcome $Y$
has error
\[
 e=\frac{1-p_1+p_0}{2}\leq\frac{1-z}{2}.
\]
Since the error indicator $Z=X_j\mathbin{\oplus}Y$ determines $X_j$
given $Y$, we have $H(X_j\mid Y,u)=H(Z\mid Y,u)\leq h_2(e)$,
giving $I(X_j:Y\mid u)\geq c_z$ and, by quantum data processing for this measurement,
$I(X_j:Q\mid u)\geq c_z$.

Averaging over prefixes and summing over $j$, we obtain from the
classical-quantum chain rule and the entropy bound for a $D_0$-dimensional
state
\begin{equation}
 Lc_z \leq I(X:Q) =S\!\left(2^{-L}\sum_x\sigma_x\right) -2^{-L}\sum_xS(\sigma_x) \leq\log_2D_0.
 \label{up:serial-chain}
\end{equation}
To check the chain rule, expand each conditional mutual information as
the entropy of the parent average minus the average entropy of its two
children, so that all intermediate entropies cancel.
The final inequality uses nonnegativity of each component entropy and
$S(\rho)\leq\log_2D_0$. Taking the floor proves the claim for every
finite $L$.

For a convex hull, represent the mixture weights as a diagonal density
matrix and take $A_E$ diagonal with the fixed response values on its
diagonal. These are valid effects, so the same argument applies.
\end{proof}

\subsection{Exact component programs}

Suppose a response class is implemented by fixed processors, with an
initial program of dimension $D_0$ and a program of dimension $D_t$ at
transition $t$. For each complete target, assume that the processors use one fixed
product state of these program inputs at every node on that target's
response path to simulate the corresponding components exactly as
complete channels, including propagated memory and reference-entangled
inputs. Composing the processors with a tester
then induces a single effect on that
whole-process program, so Lemma~\ref{up:information}, with $\eta=0$, gives
\begin{equation}
 \operatorname{sfat}_\gamma
 \leq\left\lfloor
 \frac{\log_2D_0+\sum_{t=1}^K\log_2D_t}{c_\gamma}
 \right\rfloor.
 \label{inf:exact-program}
\end{equation}

For example, consider component input and output spaces $\mathcal A_t$
and $\mathcal B_t$, including any propagated process memory.
Fix an orthonormal basis $\{|x\rangle\}$ of each $\mathcal A_t$,
and let $\boldsymbol\sigma$ denote the target's family of density
matrices $\sigma_{t,x}$ on $\mathcal B_t$.
For an arbitrary reference register $R$, consider transitions of the address form
\begin{equation}
(\operatorname{id}_R\otimes\Phi_{t,\boldsymbol\sigma})(\rho_{R\mathcal A_t})
 =\sum_{x=1}^{a_t}
 \langle x|\rho_{R\mathcal A_t}|x\rangle
 \otimes\sigma_{t,x}^{\mathcal B_t},
 \qquad \dim\mathcal A_t=a_t,\quad\dim\mathcal B_t=b_t.
 \label{inf:address-form}
\end{equation}
The fixed processor measures the address, returns the selected factor
of $\bigotimes_x\sigma_{t,x}$, and discards the others.
This gives an exact simulation for arbitrary $R$ with program dimension $b_t^{a_t}$,
also when fixed coherent wiring and an auxiliary system carried past a component
with target-independent processing are composed with this channel action. Thus
\begin{equation}
 \operatorname{sfat}_\gamma
 \leq\frac{\log_2D_0+\sum_ta_t\log_2b_t}{c_\gamma}.
 \label{inf:address-bound}
\end{equation}
If $a_t,b_t,D_0\leq dr$, the capacity bound is
$O_d(Kr\log(dr)/\gamma^2)$.

\begin{remark}[Static response representations]
For fixed $d,r$ and $0<\xi<1/2$, the smallest logarithmic dimension
$\log_2D$ of a common state encoding of $\mathcal P_{d,r}(K)$ with target-independent effects
and uniform response error at most $\xi$ has order
$K\log(K+1)$. Indeed, the perfect phase tree and
Lemma~\ref{up:information}, with $\eta=2\xi$ and $\gamma=1$, give
$\log_2D\geq[1-h_2(\xi)]L(K)$, where $L(K)$ is the phase-tree
depth in Eq.~\eqref{low:binary-depth}. For the upper bound, encode each
center of a proper $D(P,Q)$-cover of radius $2\xi$ as an orthogonal label
and put its tester responses on the diagonal of each effect.
Proposition~\ref{up:physical-cover} bounds the label count.
\end{remark}

\subsection{Information in one physical trajectory}
\label{inf:section}

The depth of a response tree counts distinguishability across experiments
on fresh process runs, with leaf labels indexing complete process witnesses.
The information retained after one run satisfies a different dimension bound.

\begin{proposition}[One-trajectory information]
\label{inf:trajectory}
Let a finite classical variable $X$ select a complete process of visible
dimension $d$ and horizon $K$. Suppose the tester's initial retained state
is independent of $X$. For any normalized adaptive tester, its final
quantum-classical register $F$ obeys
\begin{equation}
 I(X:F)\leq2(K+1)\log_2d.
 \label{inf:trajectory-bound}
\end{equation}
The same bound holds for the classical transcript of its final measurement.
\end{proposition}
\begin{proof}
Let $W_t$ contain everything retained by the tester immediately before
receiving the $d$-dimensional output $B_t$, including its earlier classical
records. The information increment is
\[
 I(X:W_tB_t)-I(X:W_t) = I(X:B_t\mid W_t) =S(B_t\mid W_t)-S(B_t\mid XW_t) \leq2\log_2d.
\]
The inequality follows from the dimension bounds
$-\log_2d\leq S(B_t\mid V)\leq\log_2d$ for any conditioning
system $V$. Local processing, discarding, and sending part of the tester's
state to the process cannot increase its retained mutual information.
Since the initial information is zero, summing over $B_0,\ldots,B_K$
and using data processing proves the claim.
\end{proof}

\subsection{Sampling adaptive response probabilities}
\label{inf:sampling}
For at most $L\geq1$ adaptively chosen experiments, let $a>0$ be the
required probability accuracy and $0<\delta<1$ the total failure
probability. With independent complete resets, use
\begin{equation}
 m(a,L,\delta)=\left\lceil\frac{1}{2a^2}
                         \ln\frac{2L}{\delta}\right\rceil
 \label{inf:sample-count}
\end{equation}
fresh complete runs per experiment. Conditional on the past,
Hoeffding's inequality~\cite{hoeffding1963probability} bounds the
failure probability by $2e^{-2ma^2}\leq\delta/L$, so a union bound gives
simultaneous accuracy $a$ for all $L$ experiments.
The cost is $Lm$ complete runs and $KLm$ process transitions, together
with reset, tester-control, readout and classical computation costs.

\section{Phase-block tools and capacity with dephasing}
\label{noise:section}

The noisy Boolean phase trees in this appendix also supply the lower
bound for general Pauli noise after syndrome recovery. We first give
the dephasing specialization. Its upper bound uses a phase cover and,
for positive dephasing, an exact classical program.
All logarithms have base two except for $\ln$.

\subsection{Dephasing model and quantitative statement}
Use the initialization and loading maps of Section~\ref{main:noise}
for integers $h\geq0$ and $T\geq1$, with $R=2^h$ and $K=h+T$.
Specialize to known $0\leq p\leq1/4$ with $q_Z=p/2$, $q_I=1-p/2$
and $q_X=q_Y=0$, so $e=p/2$. For a fixed phase table
$\Theta=(\theta_{tj})_{1\leq t\leq T,\ j\in\{0,1\}^h}$,
the remaining transitions are
\begin{align}
 &V_\theta=\diag(e^{i\theta},1),\qquad Z=\diag(1,-1),\qquad \mathcal D_{\theta,p}(\rho)
  =(1-p/2)V_\theta\rho V_\theta^\dagger
    +(p/2)V_\theta Z\rho ZV_\theta^\dagger,\label{main:noisy-local}\\
 &\Lambda_{t,\Theta,p}(X)
  =\sum_j\mathcal D_{\theta_{tj},p}
       (\langle j|_M X|j\rangle_M)\otimes|j\rangle\langle j|_M.
 \label{main:noisy-address}
\end{align}
Thus the channel multiplies the visible off-diagonal entry by
$(1-p)e^{i\theta}$ and leaves the diagonal entries unchanged. The unitary convention satisfies $V_\theta=e^{i\theta}U_{-\theta}$
for the phase unitary of Section~\ref{main:fixed}
[Eq.~\eqref{low:phase-process}], so negating the phase
and measurement offset used there gives this response convention.
Equation~\eqref{main:noisy-address} defines the channel also on inputs
with cross-address coherence by measuring the address at every
transition and preserving its classical value.

Denote this initialized process family by $\mathcal N_{R,T,p}$,
with response functions evaluated on all testers in
Section~\ref{main:model}. Its entire propagated memory has dimension
$R$, so $\mathcal N_{R,T,p}\subseteq\mathcal P_{2,R}(h+T)$.

\begin{theorem}[Capacity with known dephasing]
\label{main:noisy-law}
For every $R=2^h$, integer $T\geq1$, $0\leq p\leq1/4$, and
$0<\gamma\leq1/16$, put $H=\min(T,1/p)$ with $1/0=\infty$.
Then
\begin{equation}
 \frac{RT}{16}\log_2(1+H)
 \leq\operatorname{sfat}_\gamma\mathcal N_{R,T,p}
 \leq\kappa_\gamma RT\log_2(1+H),
 \label{main:noisy-matching}
\end{equation}
where one may take
\begin{equation}
 c_\gamma=1-h_2((1-\gamma)/2),\qquad
 \kappa_\gamma=\max\{1+\log_2(1+4\pi/\gamma),\ 2/c_\gamma\}.
 \label{main:noisy-constants}
\end{equation}
These constants are independent of $R,T,p$.
\end{theorem}
\begin{proof}
Sections~\ref{noise:simulation}--\ref{noise:upper-bounds} give a
phase-cover upper bound for every $p$ and an exact classical program for $p>0$.
Sections~\ref{noise:parking}--\ref{noise:rounding} concatenate short-block
phase trees for the lower bound, with one fixed table per leaf.
\end{proof}

\subsection{Complete channel model}
\label{noise:model}

Identify the address strings $a\in\{0,1\}^h$ with
$j\in\{1,\ldots,R\}$. The initialized state is
\begin{equation}
 \rho_{B_0M_0}
 =|0\rangle\langle0|_{B_0}\otimes
   |0^h\rangle\langle0^h|_{M_0}.
 \label{noise:initial-state}
\end{equation}
At loading transition $\ell=1,\ldots,h$, the
channel~\eqref{main:noisy-loader} has Kraus operators from
$A\otimes M$ to $B\otimes M$
\begin{equation}
 L^{(\ell)}_{a,b}
 =|0\rangle_B|w_\ell(a,b)\rangle_M
      \langle b|_A\langle a|_M,
 \qquad a\in\{0,1\}^h,\quad b\in\{0,1\}.
 \label{noise:loading-kraus}
\end{equation}
Each output ket is normalized and the input bras form an orthonormal
basis, giving $\sum_{a,b}(L^{(\ell)}_{a,b})^*L^{(\ell)}_{a,b}=I_{AM}$
and hence a CPTP map on the entire input space, including arbitrary memory
and reference-entangled inputs. Starting from
\eqref{noise:initial-state}, the $h$ measured input bits form the
address, so supplying the basis bits of $j$ selects it deterministically.
For $h=0$, the memory is one-dimensional and no loading is needed.

Write $\Phi_{\theta,p}=\mathcal D_{\theta,p}$ for the local
channel~\eqref{main:noisy-local}, with the same $V_\theta$ and $Z$.
The rate $p\in[0,1/4]$ is known and fixed throughout the class, and
the complete table belongs to $(\mathbb R/2\pi\mathbb Z)^{T\times R}$.
The full phase transition~\eqref{main:noisy-address} has Kraus operators
\begin{equation}
 K_{j,z}=\sqrt{q_z}\,V_{\theta_{tj}}Z^z\otimes
             |j\rangle\langle j|,
 \qquad q_0=1-p/2,\quad q_1=p/2,
 \label{noise:phase-kraus}
\end{equation}
whose squared products sum to $I_{AM}$. For any reference $W$,
the extension of \eqref{main:noisy-address} applies
$\operatorname{id}_W\otimes\Phi_{\theta_{tj},p}$ inside each address block,
deleting cross-address coherence and preserving the classical address
while allowing reference-entangled visible inputs within each block.
Measurement records and noise environments are discarded locally,
and the process discards $M_K$ at the end.

The address is the entire propagated memory of this realization,
including for degenerate tables that admit a smaller realization.
The normalized tester domain and fixed-target convention are those
of Section~\ref{main:noise}.

\subsection{Exact finite simulation of the noisy channels}
\label{noise:simulation}

In \eqref{main:noisy-address}, the address selects a channel that
transmits the visible quantum input, so the simulation requires the
full-channel interface \eqref{inf:exact-program}.
Equation~\eqref{inf:address-form} instead describes measure-and-prepare
transitions. The next lemma expresses the
off-diagonal coefficient as a convex combination of fixed phases.

\begin{lemma}[A common polygon of phase channels]
\label{noise:polygon}
For $0<p\leq1/4$, set
\begin{equation}
 m=\left\lceil\frac{\pi}{\arccos(1-p)}\right\rceil,
 \qquad v_k=e^{2\pi ik/m},\quad k=0,\ldots,m-1.
 \label{noise:polygon-size}
\end{equation}
For every $\theta$, there are nonnegative weights $w_k(\theta)$
summing to one such that
\begin{equation}
 \Phi_{\theta,p}
 =\sum_{k=0}^{m-1}w_k(\theta)
     \operatorname{Ad}_{V_{2\pi k/m}},
 \qquad \operatorname{Ad}_V(X)=VXV^*.
 \label{noise:polygon-identity}
\end{equation}
The equality is an identity of channels and holds after adjoining
any reference system.
\end{lemma}

\begin{proof}
The regular polygon with vertices $v_k$ has inradius $\cos(\pi/m)$,
which, by the definition of $m$, satisfies
\begin{equation}
 \cos(\pi/m)\geq1-p.
 \label{noise:inradius}
\end{equation}
Consequently it contains the disk $|z|\leq1-p$, and in particular
$z=(1-p)e^{i\theta}$. To choose the weights explicitly, intersect the ray
through $z$ with the polygon boundary at $y$, choosing either adjacent
edge by a deterministic rule when the intersection is a vertex.
Express $y$ as a convex combination of the two vertices of its edge
and the center as the uniform average of all vertices. Substituting these representations into the decomposition
\[
 z=\frac{|z|}{|y|}y+
       \left(1-\frac{|z|}{|y|}\right)0
\]
gives $\sum_kw_k(\theta)v_k=(1-p)e^{i\theta}$.

Both sides of \eqref{noise:polygon-identity} preserve the two diagonal
matrix units. On $|0\rangle\langle1|$, the right side multiplies
by $\sum_kw_kv_k=(1-p)e^{i\theta}$, and on its adjoint it
multiplies by the complex conjugate. These four equalities prove the linear-map identity, whose tensor product
with $\operatorname{id}_W$ gives the reference extension.
\end{proof}

For a complete discrete table
$\boldsymbol k=(k_{tj})\in\{0,\ldots,m-1\}^{T\times R}$, define
\begin{equation}
 W_\Theta(\boldsymbol k)
 =\prod_{t=1}^T\prod_{j=1}^R
       w_{k_{tj}}(\theta_{tj}).
 \label{noise:table-weights}
\end{equation}
Let $P_{\boldsymbol k}$ be the known complete process with
initialization \eqref{noise:initial-state} and loading
\eqref{noise:loading-kraus}, followed at slot $t$ by the address
measurement and the selected unitary $V_{2\pi k_{tj}/m}$.
Averaging its full transition over all $k_{tj}$ at a fixed $t$
gives \eqref{main:noisy-address}, because in each address block
the selected coordinate gives \eqref{noise:polygon-identity} while
all other coordinate weights sum to one, including when the block is
correlated with the tester.

For a fixed adaptive tester, the response is multilinear in the
successive transitions, so expanding those transitions in time gives
\begin{equation}
 f_\Theta(E)=\sum_{\boldsymbol k}
       W_\Theta(\boldsymbol k)f_{P_{\boldsymbol k}}(E)
 \quad\text{for every normalized tester }E.
 \label{noise:whole-process-mixture}
\end{equation}
Each coordinate $(t,j)$ occurs in only one transition of a run.
The product in \eqref{noise:table-weights} thus implements the
independent local noise of the original process.

For each complete target, define one program for all node effects on
a common auxiliary space of dimension $D=m^{RT}$, separate from the
propagated memory:
\begin{equation}
 \Omega_\Theta=\sum_{\boldsymbol k}W_\Theta(\boldsymbol k)
       |\boldsymbol k\rangle\langle\boldsymbol k|.
 \label{noise:whole-program}
\end{equation}
The fixed processor consists of a measurement of the program in this basis
followed by the controlled circuit that runs $P_{\boldsymbol k}$ according
to the measured complete table, making it CPTP on every program input.
Composing this processor with a normalized tester produces the
target-independent effect
\begin{equation}
 A_E=\sum_{\boldsymbol k}f_{P_{\boldsymbol k}}(E)
              |\boldsymbol k\rangle\langle\boldsymbol k|,
 \qquad 0\leq A_E\leq I,
 \qquad \operatorname{Tr}(A_E\Omega_\Theta)=f_\Theta(E).
 \label{noise:program-effect}
\end{equation}
Basis enumeration and exact-weight preparation for this encoding
require additional computational resources.

\subsection{Information and covering bounds}
\label{noise:upper-bounds}

Applying Lemma~\ref{up:information} with $\eta=0$ to the complete
target programs in \eqref{noise:whole-program} and the target-independent
effects in \eqref{noise:program-effect}, and using $D=m^{RT}$ and
$c_\gamma$ as in \eqref{up:information-increment}, we obtain for $p>0$
\begin{equation}
 \operatorname{sfat}_\gamma(\mathcal N_{R,T,p})
 \leq\left\lfloor\frac{RT\log_2m}{c_\gamma}\right\rfloor.
 \label{noise:program-upper}
\end{equation}

\begin{lemma}[Phase cover]
\label{noise:cover}
For every $p\in[0,1/4]$ and $0<\gamma\leq1$, let
$N=\lceil4\pi T/\gamma\rceil$. Then
\begin{equation}
 \operatorname{sfat}_\gamma(\mathcal N_{R,T,p})
 \leq\lfloor RT\log_2N\rfloor.
 \label{noise:cover-upper}
\end{equation}
\end{lemma}

\begin{proof}
Address measurement averages the block errors, and telescoping adds
the errors over phase slots. Rounding every table entry to a nearest point among $N$ equally spaced
points on the phase circle gives circular distance at most
$\delta=\pi/N\leq\gamma/(4T)$, and choosing corresponding real lifts
then gives $\|V_\theta-V_{\theta'}\|\leq\delta$ in operator norm. Expand the difference
of unitary conjugations and apply the trace-norm multiplication bound,
also on a reference extension. The common dephasing postprocessing
preserves the resulting bound, giving
\begin{equation}
 \|\Phi_{\theta,p}-\Phi_{\theta',p}\|_\diamond
 \leq2\delta.
 \label{noise:local-distance}
\end{equation}

For the full address transition, measuring the address of an
arbitrary normalized input state with reference produces positive,
possibly subnormalized blocks $\rho_j^{WA}$ with
$\sum_j\operatorname{Tr}\rho_j^{WA}=1$, for which the difference between
the outputs of the two rounded transitions obeys
\begin{equation}
 \left\|\sum_j (\operatorname{id}_W\otimes [\Phi_{\theta_{tj},p}-\Phi_{\theta'_{tj},p}]) (\rho_j^{WA})\otimes|j\rangle\langle j|\right\|_1 \leq 2\delta\sum_j\operatorname{Tr}\rho_j^{WA}=2\delta.
 \label{noise:address-distance}
\end{equation}
For a difference of channels, optimizing over normalized states
with a reference suffices for the diamond norm. Thus the same
$2\delta$ bound holds for the complete transition, independently
of the number of address blocks.

The loading maps are identical and contribute no replacement error, so for an arbitrary adaptive tester we replace only the $T$ phase transitions, one at a time, including its retained workspace in the reference for each replacement. Since all preceding and succeeding operations are CPTP, trace-distance contractivity bounds the final trace-norm error by $2T\delta$ and hence the binary probability error by $T\delta\leq\gamma/4$.

There are $N^{RT}$ rounded physical phase tables. At the first
divergence of two distinct leaf witnesses in a shattered tree,
their responses at the common node differ by at least $\gamma$.
If both witnesses had the same rounded center, their difference
would be at most $2(\gamma/4)=\gamma/2$. Hence different leaves
require different centers, and $2^L\leq N^{RT}$ for every finite
shattered depth $L$.
\end{proof}

Combining \eqref{noise:program-upper} and \eqref{noise:cover-upper}
gives the finite estimate
\begin{equation}
 \operatorname{sfat}_\gamma(\mathcal N_{R,T,p})
 \leq\left\lfloor RT\min\left\{
      \log_2\left\lceil\frac{4\pi T}{\gamma}\right\rceil,
      \frac{\log_2m}{c_\gamma}\right\}\right\rfloor
 \quad(p>0).
 \label{noise:finite-upper}
\end{equation}
At $p=0$, the covering term still applies.

Put $H=\min(T,1/p)$, with $1/0=\infty$. Since
$N\leq(1+4\pi/\gamma)(1+T)$ and $\log_2(1+T)\geq1$,
\begin{equation}
 \log_2N\leq
   [1+\log_2(1+4\pi/\gamma)]\log_2(1+T).
 \label{noise:cover-simplification}
\end{equation}
For $p>0$, the inequality $1-\cos x\leq x^2/2$ gives
$\arccos(1-p)\geq\sqrt{2p}$, and therefore
\begin{equation}
 m\leq1+\frac{\pi}{\sqrt{2p}}\leq\frac4{\sqrt p},
 \qquad \log_2m\leq2\log_2(1+1/p).
 \label{noise:polygon-simplification}
\end{equation}
For the last inequality, set $u=1/p\geq4$ and use
$2+\tfrac12\log_2u\leq\tfrac32\log_2u
\leq2\log_2(1+u)$. Use the cover for $p=0$ or $T\leq1/p$, and the
program bound for $p>0$ and $T>1/p$. The resulting upper constant is
\begin{equation}
 \kappa_\gamma=
 \max\left\{1+\log_2(1+4\pi/\gamma),\frac2{c_\gamma}\right\},
 \qquad
 \operatorname{sfat}_\gamma(\mathcal N_{R,T,p})
 \leq\kappa_\gamma RT\log_2(1+H).
 \label{noise:uniform-upper}
\end{equation}

\subsection{A retained control and a reset carrier}
\label{noise:parking}

The lower bound uses a tester with one coherent control qubit $C$
and one reusable visible carrier $A$, operated by deterministic known
controls without postselection or error correction. Active slots
transfer phase and attenuation to the control, while inactive slots
leave it in the tester's private space. After deterministically loading address $j$, prepare $C$ in $|+\rangle$ and $A$ in
$|0\rangle$. For $s\in\{-1,0,1\}$, define the known preparation
unitaries, in the order $C,A$, by
\begin{equation}
 W_+=\operatorname{CNOT}_{C\to A},\qquad
 W_-=X_A\operatorname{CNOT}_{C\to A},\qquad W_0=I.
 \label{noise:preparation-gates}
\end{equation}
The carrier values on control branches $0,1$ are respectively
$(0,1)$, $(1,0)$, and $(0,0)$. At each phase slot the tester applies
$W_s$, sends only the carrier once forward through the process,
receives the output, and applies the known preparation inverse $W_s^*$.
Pauli $X_A$ on the CNOT target qubit commutes with
this CNOT, so $W_-$ may also be reversed with the same gates.

To compute the reduced action on $C$, let $a_c$ be the prepared
carrier value on branch $c$. A noise Kraus branch
$V_\theta Z^z$ multiplies that branch by
\begin{equation}
 e^{i\theta(1-a_c)}(-1)^{za_c}.
 \label{noise:kraus-phase}
\end{equation}
Undoing $W_s$ returns the carrier to $|0\rangle$ in each Kraus
branch, decoupled from the control, and leaves the control diagonals
unchanged. Applying \eqref{noise:kraus-phase} to the control's $01$
entry and summing the two noise branches gives the multipliers
\begin{equation}
 \begin{cases}
 (1-p)e^{i\theta},&s=1,\\
 (1-p)e^{-i\theta},&s=-1,\\
 1,&s=0
 \end{cases}.
 \label{noise:control-multiplier}
\end{equation}
This is a channel identity on the control, so it can be composed
over slots after the fresh noise environment at each slot is
discarded. For a complete string $s\in\{-1,0,1\}^T$, put
$S=\sum_ts_t\theta_{tj}$ and
$v_s=(1-p)^{|\operatorname{supp}s|}$, where
$\operatorname{supp}s=\{t:s_t\ne0\}$ is the set of active slots.
The final control state is
\begin{equation}
 \rho_C=\frac12
  \begin{pmatrix}1&v_se^{iS}\\v_se^{-iS}&1\end{pmatrix}.
 \label{noise:control-state}
\end{equation}
Projection onto $(|0\rangle+e^{-i\phi}|1\rangle)/\sqrt2$ gives
the normalized binary response
\begin{equation}
 f_\Theta(E_{j,s,\phi})
 =\frac{1+v_s\cos(S-\phi)}2.
 \label{noise:signed-response}
\end{equation}

\subsection{A two-slot block}
\label{noise:two-block-example}

Consider $h=1$, $R=2$ and $T=4$, with blocks $(1,2)$ and $(3,4)$.
Each experiment loads the selected address bit and runs all four phase
slots, while the tester retains its control qubit throughout.
For block $(1,2)$, the three queries have signed patterns
$(1,1,0,0)$, $(-1,1,0,0)$ and $(0,1,0,0)$, while block
$(3,4)$ uses the same patterns shifted to the last two slots.
For either block $(a,b)$ at the selected address $j$, choose independently
\begin{equation}
 \alpha=\theta_{aj}+\theta_{bj}=\pi\varepsilon,
 \qquad \beta=\theta_{bj}=\pi m_0/2,
 \qquad \varepsilon\in\{0,1\},\quad m_0\in\{0,1,2,3\}.
 \label{noise:two-block-labels}
\end{equation}
Querying the sum with offset $0$ determines $\alpha$ from the first bit, so
the difference can then be queried with offset $-\alpha$.
The second bit determines $m_0\bmod2$, making $\pi(m_0\bmod2)/2$ a known
offset for the final query of the second slot. The cosine factors are
\begin{equation}
 (-1)^\varepsilon,\qquad
 (-1)^{m_0\bmod2},\qquad
 (-1)^{\lfloor m_0/2\rfloor}.
 \label{noise:two-block-responses}
\end{equation}
The three responses are $(1\pm(1-p)^2)/2$, $(1\pm(1-p)^2)/2$ and
$(1\pm(1-p))/2$. With threshold $1/2$, assign the positive branch
to exponent bit zero. A leaf fixes both phases through
$\theta_{aj}=\alpha-\beta$ and $\theta_{bj}=\beta$.

Even at $p=1/4$, a two-slot
visibility is $9/16$, above every $\gamma\leq1/16$ considered here.
Making these choices independently for both blocks at both addresses
gives one $4\times2$ phase table for each depth-$12$ path. The construction
below extends this concatenation to general block lengths.

\subsection{Boolean phase trees in noisy blocks}
\label{noise:boolean}

Noise changes only the visibility in Eq.~\eqref{noise:signed-response}.
For a dyadic block of length $n=2^\ell$ at one address, use
Theorem~\ref{low:phase-tree} with slots $Z\subseteq[\ell]$ and
subset coordinates indexed by $A\subseteq[\ell]$.
The subset-sum inversion
\eqref{low:mobius-forward}--\eqref{low:mobius-inverse} assigns one
fixed phase table to every choice of the labels
$m_A\in\{0,\ldots,2^{|A|+1}-1\}$, with
$\alpha_A=2\pi m_A/2^{|A|+1}$. Their total bit count is
\begin{equation}
 d(n)=\sum_{A\subseteq[\ell]}(|A|+1)
     =n\left(1+\frac\ell2\right).
 \label{noise:block-depth}
\end{equation}
This includes the singleton $\ell=0$, with depth one.

Order $A$ by increasing size, fix an order within each size, and read
each label from least to most significant bit, as in the phase tree.
At bit $k$ the earlier labels determine $r_k=m_A\bmod2^k$ and
every proper-subset coordinate. For a choice of $J\subseteq A$ of size $|A|-k$,
Eqs.~\eqref{low:signed-polynomial} and \eqref{low:known-phase} supply
the sign string $q_{A,J}$ and known offset $\beta_{A,J}$, respectively. With
$\phi=\beta_{A,J}+2\pi r_k/2^{k+1}$, the bit-isolation identity is
\begin{equation}
 \sum_Zq_{A,J}(Z)\theta_Z-\phi=\pi b_k\pmod{2\pi},
 \label{noise:bit-isolation}
\end{equation}
where $b_k$ is the $k$th digit of $m_A$. The identity holds for every
completion of the prefix. For $J=\varnothing$, we have $\beta_{A,J}=0$.

Use $q_{A,J}$ as the sign string on this block and zero at every
other phase slot. Equation~\eqref{noise:signed-response} becomes
\begin{equation}
 f_\Theta(E_{j,q_{A,J},\phi})
 =\frac{1+(-1)^{b_k}v_q}{2},
 \qquad v_q=(1-p)^{|\operatorname{supp}q_{A,J}|}.
 \label{noise:noisy-bit}
\end{equation}
With threshold $1/2$, let the response branch label be $x=1-b_k$.
Its signed margin is $v_q/2$, so the shattering condition is
exactly $v_q\geq\gamma$. When $np\leq1/4$, Bernoulli's
inequality and $|\operatorname{supp}q_{A,J}|\leq n$ give
\begin{equation}
 v_q\geq(1-p)^n\geq1-np\geq\frac34.
 \label{noise:visibility-bound}
\end{equation}
For $p=0$, every block has visibility one regardless of its length.

Choose parameters independently across addresses and blocks, and
concatenate the trees in a fixed order of address, block, coordinate,
and bit. A complete path specifies every independent $m_A$, hence
fixes all $RT$ phase entries by applying the same inverse transform in
Eq.~\eqref{low:mobius-inverse} separately to each block at each address,
while the loading maps and noise rate are common to all paths. Because each node's signs and offset depend
only on the earlier labels and \eqref{noise:bit-isolation} holds for
every completion of those labels, a single complete target witnesses
all node inequalities on each leaf.

\subsection{Block partition and uniform constants}
\label{noise:rounding}

To keep a fixed margin with dyadic blocks within the visibility bound,
put $b=T$ for $p=0$ and $b=\min(T,\lfloor1/(4p)\rfloor)$ for $p>0$.
Since the assumption $p\leq1/4$ ensures $b\geq1$, define
\begin{equation}
 q=2^{\lfloor\log_2b\rfloor},\qquad
 a=\lfloor T/q\rfloor,\qquad r_0=T-aq.
 \label{noise:block-partition}
\end{equation}
Partition the first $aq$ slots into $a$ consecutive blocks of
length $q$, and the remaining $r_0$ slots into disjoint dyadic
blocks according to the binary expansion of the remainder. Every block has
length $n\leq q\leq b$ and hence meets
\eqref{noise:visibility-bound} when $p>0$. An empty remainder
contributes no block, and a singleton uses the $\ell=0$ case
of \eqref{noise:block-depth}.

The complete construction gives depth
\begin{equation}
 D_*=R\left[
       aq\left(1+\frac12\log_2q\right)
       +\sum_{\ell:\,\operatorname{bit}_\ell(r_0)=1}
           2^\ell\left(1+\frac\ell2\right)\right].
 \label{noise:constructed-depth}
\end{equation}
To obtain a uniform lower constant, it suffices to retain the
first term of \eqref{noise:constructed-depth}. For $p>0$, the integer $q$ is the largest
power of two at most $\min(T,1/(4p))$, and consequently
\begin{equation}
 q>\frac12\min(T,1/(4p))\geq\frac H8,
 \qquad H=\min(T,1/p).
 \label{noise:block-size-bound}
\end{equation}
For $p=0$, the inequality $q>T/2=H/2$ gives the same last bound.
Since $T/q\geq1$, we have $a\geq T/(2q)$. It follows that
\begin{equation}
 D_*\geq\frac{RT}{2}\left(1+\frac12\log_2q\right)
      \geq\frac{RT}{4}\log_2(2q).
 \label{noise:depth-lower}
\end{equation}
For every $q\geq1$, we have
$1+H<1+8q\leq(2q)^4$. The second inequality holds at $q=1$
because $9\leq16$, and $16q^4-8q-1$ is increasing for $q\geq1$.
Thus $\log_2(1+H)\leq4\log_2(2q)$. Since every constructed node has
visibility at least $3/4$ for $p>0$ and visibility one for $p=0$,
the construction gives, for every fixed $0<\gamma\leq1/16$,
\begin{equation}
 \operatorname{sfat}_\gamma(\mathcal N_{R,T,p})
 \geq D_*\geq\frac{RT}{16}\log_2(1+H).
 \label{noise:uniform-lower}
\end{equation}
Together with \eqref{noise:uniform-upper}, this proves
Theorem~\ref{main:noisy-law}, with constants uniform in $R,T,p$.

The endpoint cases are included in the construction. When $T=1$,
the partition has $q=1$ and $D_*=R$, with one phase bit per address
and visibility at least $3/4$. When $R=1$, loading is empty.
At $p=0$, the resulting order is
$\Theta_\gamma(RT\log_2(1+T))$. At $p=1/4$, the partition has
$q=1$ and $D_*=RT$, and the polygon bound permits $m=5$.

\subsection{Sampling threshold responses}
\label{noise:sampling}
When $|f_P(E)-s|\geq\gamma/2$ at each selected node $(E,s)$, use
Eq.~\eqref{inf:sample-count} with $a=\gamma/4$ to obtain
\begin{equation}
 n_{\rm shot}=\left\lceil\frac{8}{\gamma^2}
                         \ln\frac{2L}{\delta}\right\rceil.
 \label{main:noisy-shots}
\end{equation}
The estimates then determine the response side at all $L$ adaptively selected nodes with total
confidence $1-\delta$. Each shot is a complete run of all
$K=h+T$ transitions, so the cost is $Ln_{\rm shot}$ runs and
$KLn_{\rm shot}$ transitions, with the controls and reset costs
specified in Appendix~\ref{inf:sampling}.

\section{Capacity with known Pauli noise}
\label{pauli:section}

We prove Theorem~\ref{pauli:law} for the measured-address model of
Section~\ref{main:noise}. All logarithms in this appendix have base two.

Use the initial state and loading channels in
Eqs.~\eqref{noise:initial-state}--\eqref{noise:loading-kraus}, with
$R=2^h$, $T\geq1$ and $K=h+T$. Loading is empty for $R=1$.

Fix a known probability vector $q=(q_I,q_X,q_Y,q_Z)$ and write
$U_\theta=\diag(1,e^{i\theta})$. The local phase channel is
\begin{equation}
 \mathcal E_{\theta,q}(\rho)
 =\sum_{P\in\{I,X,Y,Z\}}q_P\,P U_\theta\rho U_\theta^*P^*,
 \label{pauli:local}
\end{equation}
where $I,X,Y,Z$ are the qubit Pauli matrices. For a fixed table
$\Theta=(\theta_{tj})\in(\mathbb R/2\pi\mathbb Z)^{T\times R}$,
the full transition at phase slot $t$ is
\begin{equation}
 \Lambda^{\mathrm P}_{t,\Theta,q}(\rho)
 =\sum_j\mathcal E_{\theta_{tj},q}
       (\langle j|_M\rho|j\rangle_M)\otimes|j\rangle\langle j|_M.
 \label{pauli:address}
\end{equation}
Its Kraus operators
$\sqrt{q_P}\,P U_{\theta_{tj}}\otimes|j\rangle\langle j|$
have squared products summing to $I$, so the map is CPTP on the entire
input algebra, including inputs entangled with a reference. The map
measures the address and preserves its classical value while the
Pauli errors are independent between slots and their environments are
discarded locally. The final address is discarded.
The resulting family is $\mathcal N^{\mathrm P}_{R,T,q}$ from
Section~\ref{main:noise}, with ideal tester operations on the same
normalized domain. Write
\begin{equation}
 e=e(q)=\min(q_I,q_Z)+\min(q_X,q_Y).
 \label{pauli:residual}
\end{equation}
The recovery below gives $e$ its residual phase-flip interpretation.

Syndrome recovery gives the noisy Boolean phase tree and the lower bound
$RT\log_2[1+\min\{T,1/e(q)\}]/32$.
The upper bound combines a finite quantum program at positive
residual error with a phase cover that includes zero error.

\subsection{Syndrome recovery and the phase tree}
\label{pauli:recovery-section}

After loading a selected address, retain a control qubit $C$ in
the tester and use a visible carrier $S$ initialized in $|0\rangle$
at each slot. For an active sign $s\in\{-1,1\}$, set
\begin{equation}
 A_+=\operatorname{CNOT}_{C\to S},\qquad
 A_-=X_S\operatorname{CNOT}_{C\to S}.
 \label{pauli:encoding}
\end{equation}
Applying $A_s$, sending $S$ through one phase transition, and applying
$A_s^*$ leaves the bit flip recorded in the carrier and the phase flip
retained by the control. Write a Pauli as $X^uZ^v$ up to scalar phase. On the
whole logical input space, direct action on its two basis vectors
gives
\begin{equation}
 A_s^*(I_C\otimes X^uZ^vU_\theta)A_s
       (|\psi\rangle_C|0\rangle_S)
 =\lambda_{s,u,v,\theta}
       (Z_C^vU_{s\theta}|\psi\rangle_C)|u\rangle_S,
 \qquad |\lambda_{s,u,v,\theta}|=1.
 \label{pauli:decoded-isometry}
\end{equation}
The scalar is independent of $|\psi\rangle$. For $s=1$ the encoded carrier value is the control bit, whereas for
$s=-1$ it is the complement of that bit, which contributes only the common factor
$e^{i\theta}(-1)^v$ and changes the phase sign. This also proves
Eq.~\eqref{pauli:decoded-isometry} after adjoining any reference.

Measure $S$ in the computational basis to obtain an outcome $u$
distinguishing the error pairs $\{I,Z\}$ and $\{X,Y\}$. For each outcome,
use the known vector $q$ to choose the more likely phase bit $v_u$ in that pair,
using a fixed rule for ties, and apply $Z_C^{v_u}$.
Discard the carrier and the measured bit. Including every syndrome
outcome, the complete recovered channel on $C$ is exactly
\begin{equation}
 \mathcal R_{s,\theta,q}
 =(1-e(q))\operatorname{Ad}_{U_{s\theta}}
       +e(q)\operatorname{Ad}_{ZU_{s\theta}}.
 \label{pauli:recovery}
\end{equation}
Indeed, the wrong phase correction has probability
$\min(q_I,q_Z)$ in the first pair and $\min(q_X,q_Y)$ in the
second.

At an inactive slot $s=0$, keep $C$ untouched, send a separate carrier
in $|0\rangle$ through the transition, and discard its output, leaving
the identity channel on $C$. Composing these
complete channel identities from the initial state $|+\rangle_C$
and measuring the projection onto
$(|0\rangle+e^{i\phi}|1\rangle)/\sqrt2$ gives
\begin{equation}
 f_\Theta(E_{j,s,\phi})
 =\frac{1+(1-2e(q))^{|\operatorname{supp}s|}
       \cos(\sum_t s_t\theta_{tj}-\phi)}2,
 \qquad \operatorname{supp}s=\{t:s_t\ne0\}.
 \label{pauli:response}
\end{equation}
In particular, noise in inactive slots does not attenuate the
retained coherence.

In the parameter range of Theorem~\ref{pauli:law},
Eq.~\eqref{pauli:response} is the signed response of
Eq.~\eqref{noise:signed-response} with $p=2e(q)\leq1/4$, so we use the Boolean phase coordinates and
short-block partition from
Sections~\ref{noise:boolean} and~\ref{noise:rounding}, including singleton
and tail blocks. At $e>0$, each block has length at most
$\min(T,\lfloor1/(8e)\rfloor)$, so its visibility is at least $3/4$.
At $e=0$, the dyadic partition has no noise constraint on block length.
With threshold $1/2$, the signed margin is therefore at least $3/8$, and
we concatenate the independent block trees over all addresses. In the resulting tree,
each node's signs and offset depend only on earlier labels, while every
complete path fixes all table entries through the same Boolean inverse
as in Eq.~\eqref{low:mobius-inverse}.

By Eq.~\eqref{noise:uniform-lower}, this gives
\begin{equation}
 \operatorname{sfat}_\gamma\mathcal N^{\mathrm P}_{R,T,q}
 \geq\frac{RT}{16}\log_2\bigl[1+\min\{T,1/(2e)\}\bigr]
 \geq\frac{RT}{32}\log_2\bigl[1+\min\{T,1/e\}\bigr].
 \label{pauli:lower}
\end{equation}
For the last inequality, the first minimum is at least half the
second and $(1+x/2)^2\geq1+x$ for $x\geq0$. The convention at
$e=0$ makes the first bound the noiseless bound.

Each active slot uses one forward target call, at most two CNOTs,
known one-qubit gates and one transient syndrome bit, while a single
private control qubit is retained between slots.

\subsection{An exact quantum program and uniform upper bounds}
\label{pauli:program-section}

The program uses the stochastic phase construction of Vidal,
Masanes and Cirac~\cite{vidal2002storing}. Its unsuccessful branch
is incorporated into the prescribed noise, following the programming
method of Ji et al.~\cite{ji2008parameter}. We give the full
channel calculation for the Pauli weights in Eq.~\eqref{pauli:local}.

For the upper bound in Theorem~\ref{pauli:law}, assume
$e=e(q)>0$. To split the probability $\delta$ of $b$ consecutive ones
in the procedure below equally between two Pauli errors, choose
$W\in\{I,X\}$ so that $\{W,WZ\}$ has the larger minimum weight $\beta$,
treating Pauli labels modulo scalar phase, and put
\begin{equation}
 \beta=\max\{\min(q_I,q_Z),\min(q_X,q_Y)\},\qquad
 b=\left\lceil\log_2\frac1{2\beta}\right\rceil,
 \qquad \delta=2^{-b}.
 \label{pauli:program-size}
\end{equation}
Then $e/2\leq\beta\leq e\leq1/8$ and $\delta\leq2\beta$,
so each error in the selected pair can absorb weight $\delta/2$.
The program length satisfies $2\leq b\leq\lceil\log_2(1/e)\rceil$,
and the $b$-qubit program is
\begin{equation}
 |\Pi_\theta\rangle
 =\bigotimes_{a=0}^{b-1}|+_{2^a\theta}\rangle,
 \qquad |+_\alpha\rangle
       =\frac{|0\rangle+e^{i\alpha}|1\rangle}{\sqrt2}.
 \label{pauli:program}
\end{equation}

For one data qubit and a program $|+_\alpha\rangle$, apply a
CNOT with the data as control and measure the program in the
computational basis. The data Kraus operators are
\begin{equation}
 M_0(\alpha)=\frac{\diag(1,e^{i\alpha})}{\sqrt2},
 \qquad
 M_1(\alpha)=\frac{\diag(e^{i\alpha},1)}{\sqrt2}.
 \label{pauli:program-instrument}
\end{equation}
On every data input, including inputs entangled with a reference, the two
outcomes each have probability $1/2$: outcome zero implements $U_\alpha$,
whereas outcome one implements $U_{-\alpha}$ up to scalar phase.
Branches containing a zero have total probability $1-\delta$ and leave
$U_\theta$ on the data at the first zero outcome: after an outcome one,
repeating with the next doubled angle leaves the net phase coefficient
$2^a-(2^a-1)=1$ at the first subsequent zero outcome.
If all $b$ outcomes are one, an event of probability $\delta$, the data
undergoes $U_{-(2^b-1)\theta}$.

On a branch containing a zero, apply a Pauli drawn from the known
probability vector
\begin{equation}
 r_P=\frac{q_P-(\delta/2)\mathbf1_{\{P=W,WZ\}}}{1-\delta}.
 \label{pauli:success-weights}
\end{equation}
The numerators are nonnegative and sum to $1-\delta$. On the
all-one branch, apply the complete computational-basis dephaser
$\mathcal D_Z=(\operatorname{Ad}_I+\operatorname{Ad}_Z)/2$
followed by $W$. Since
$\mathcal D_Z\operatorname{Ad}_{U_\alpha}=\mathcal D_Z$ for
every $\alpha$, this branch erases the residual phase, and its
contribution supplies $\delta/2$ of each channel with errors $W$ and $WZ$.
Averaging all outcomes gives the exact channel identity
\begin{equation}
 (1-\delta)\sum_P r_P\operatorname{Ad}_{PU_\theta} +\delta\operatorname{Ad}_W\mathcal D_Z =\sum_Pq_P\operatorname{Ad}_{PU_\theta} =\mathcal E_{\theta,q}.
 \label{pauli:exact-program}
\end{equation}
All operations form one fixed CPTP processor
on arbitrary joint program and data inputs. Equation~\eqref{pauli:exact-program}
gives its action on the target programs and holds after adjoining
any reference.

Each complete target table specifies one program for all of its node
effects:
\begin{equation}
 |\Pi_\Theta\rangle=\bigotimes_{t=1}^T\bigotimes_{j=1}^R
       |\Pi_{\theta_{tj}}\rangle,
 \qquad D_{\mathrm{prog}}=2^{bRT}.
 \label{pauli:table-program}
\end{equation}
After running the fixed loaders, at slot $t$ measure the address,
use the program factor assigned to $(t,j)$, and preserve the measured
address value. Each slot uses a new program factor, and unused factors
can be discarded. Since the channel identity holds inside every address block for inputs
correlated with the tester, induction over the slots reproduces the
complete initialized process.

The exact-component interface in Eq.~\eqref{inf:exact-program}
therefore supplies, for each normalized tester, an effect satisfying
\begin{equation}
 0\leq A_E\leq I,\qquad
 \operatorname{Tr}(A_E|\Pi_\Theta\rangle\langle\Pi_\Theta|)
       =f_\Theta(E).
 \label{pauli:program-effect}
\end{equation}
Lemma~\ref{up:information} gives
\begin{equation}
 \operatorname{sfat}_\gamma\mathcal N^{\mathrm P}_{R,T,q}
 \leq\left\lfloor\frac{bRT}{c_\gamma}\right\rfloor,
 \qquad c_\gamma=1-h_2((1-\gamma)/2).
 \label{pauli:program-upper}
\end{equation}

A phase cover also applies when $e=0$. Use
$N=\lceil8\pi T/\gamma\rceil$ equally spaced phases. The local
half-diamond error is at most $\pi/N$ because the common Pauli
postprocessing preserves the $p=0$ bound in
Eq.~\eqref{noise:local-distance}. The measured-address block estimate
\eqref{noise:address-distance} and the telescoping argument in
Lemma~\ref{noise:cover} then give response error at most
$T\pi/N\leq\gamma/8$ for every adaptive tester, independently of $R$.

There are $N^{RT}$ rounded tables. At the first divergent node of
two leaf witnesses, their responses differ by at least $\gamma$,
whereas witnesses sharing one rounded table would differ by at most
$\gamma/4$. Each leaf needs a distinct center, so
\begin{equation}
 \operatorname{sfat}_\gamma\mathcal N^{\mathrm P}_{R,T,q}
 \leq\lfloor RT\log_2N\rfloor.
 \label{pauli:cover-upper}
\end{equation}
For fixed $\gamma$, the grid term is at most a constant times
$RT\log_2(1+T)$. For $e>0$, the program term is at most
$(2/c_\gamma)RT\log_2(1+1/e)$, since
$b\leq\lceil\log_2(1/e)\rceil$. Taking the smaller of the two
bounds, and using the cover alone at $e=0$, proves the uniform upper
bound in Theorem~\ref{pauli:law}. Equation~\eqref{pauli:lower}
completes the proof.

\section{Classical comparison and fixed-target prediction}
\label{pred:section}

\paragraph*{Fully measured classical processes.}
Fix computational bases for the visible space and the complete internal
memory, and put $n=dr$. Let $\mathcal C_{d,r}(K)$ consist of processes
with a diagonal initial state on $B_0M_0$ and transition kernels
$p_t(v\mid u)$ on the $n$ visible--memory pairs. On every joint input
with an untouched reference $Z$, the transition is
\begin{equation}
 (\operatorname{id}_Z\otimes\Phi_t)(\rho)
 =\sum_{u,v=1}^n p_t(v\mid u)
       \langle u|\rho|u\rangle_Z\otimes|v\rangle\langle v|.
 \label{pred:classical-model}
\end{equation}
Here $\langle u|\rho|u\rangle_Z$ is the subnormalized operator on $Z$
obtained by contracting the visible--memory factor against $|u\rangle$.
Since $p_t(v\mid u)\geq0$ and $\sum_vp_t(v\mid u)=1$, the Kraus operators $\sqrt{p_t(v\mid u)}|v\rangle\langle u|$
give a CPTP map on the entire input algebra, with the memory component
of $v$ passed to the next transition. The class thus includes
time-dependent stochastic memory updates with arbitrary real probabilities. Every persistent random record belongs to the $r$-state memory,
while fresh local randomness is discarded. The tester class is the
same as for $\mathcal P_{d,r}(K)$.

\begin{proposition}[Capacity of fully measured processes]
\label{pred:classical-capacity}
For integers $d\geq2$, $r\geq1$, $K\geq1$, and $0<\gamma\leq1$,
\begin{equation}
 K+1\leq\operatorname{sfat}_\gamma\mathcal C_{d,r}(K)
 \leq\left\lfloor
       \frac{(1+Kdr)\log_2(dr)}{c_\gamma}
      \right\rfloor,
 \label{pred:classical-bound}
\end{equation}
where $c_\gamma$ is defined in Eq.~\eqref{up:information-increment}.
Consequently, at fixed $d,r,\gamma$, this capacity is
$\Theta_{d,r,\gamma}(K)$.
\end{proposition}

\begin{proof}
For each function $g:[n]\to[n]$, assign the weight
\[
 w_t(g)=\prod_{u=1}^n p_t(g(u)\mid u).
\]
These weights sum to one, and the sum over $g$ satisfying $g(u)=v$
equals $p_t(v\mid u)$. Thus each transition is a convex mixture of
$n^n$ fixed channels that measure $u$ and prepare $|g(u)\rangle$.
This identity holds in Eq.~\eqref{pred:classical-model}, including
its reference blocks.

Encode the initial basis pair and one deterministic table per
transition in a diagonal program state on a common space. The target's initial
distribution and the product of the table distributions $w_t$
specify this state, which a fixed processor measures to deliver
the indicated initial visible--memory pair and apply the indicated
measured table at each transition. This circuit is CPTP on every
program input, and its equality with the target as a complete
channel at each step persists under arbitrary adaptive tester
interventions. Composing the circuit with any tester therefore
gives an ordinary effect on the common program space, whose dimension is
\[
 D_{\rm cl}=n(n^n)^K=n^{1+Kn}.
\]
This auxiliary space stores an information encoding of the target,
while the propagated process memory has dimension $r$.
Lemma~\ref{up:information}, with $\eta=0$, proves the upper bound.
This is the diagonal specialization of the component program described
after Eq.~\eqref{inf:address-form}.

For the lower bound, choose $K+1$ independent bits, emit the first
as the initial output, and replace the input at each transition by
the corresponding next bit. To query one bit, measure that output,
retain the classical result, and complete the remaining transitions.
These are memory-one processes of the form
\eqref{pred:classical-model}. One fixed bit list realizes each leaf of a
depth-$K+1$ tree with zero-or-one responses.
\end{proof}

\begin{remark}[Binary endpoint and resources]
Using the common tester domain with $d=2$, $r=1$, and $\gamma=1$,
the upper bound is $2K+1$ and is attained. Choose the initial bit and both
outputs of every
deterministic binary transition independently. A fresh complete run
queries any one of these bits by preparing the appropriate input
at its transition and recording its output. Hence
$\operatorname{sfat}_1\mathcal C_{2,1}(K)=2K+1$.
The phase-tree witnesses in Theorem~\ref{low:phase-tree} also have
one-dimensional process memory. At $K=8$, that tree has depth
$L(8)=20$, while this classical maximum is $17$.
\end{remark}

\paragraph*{Prediction from scalar responses.}
For one fixed target, we compare approximate-probability feedback and
large-error counts with single-outcome feedback and cumulative squared
probability error. Prediction rounds are distinct from transitions
within a process run.

\begin{definition}[Fixed-target scalar feedback]
\label{pred:scalar-protocol}
Fix a target $P\in\mathcal P_{d,r}(K)$ before round one and a
precision $\varepsilon>0$. The learner's random coins are independent
of $P$, and its access to the target consists of scalar feedback, with
no additional target queries, internal tester records, or retained
quantum systems. The learner may predict any scalar in $[0,1]$
for each announced tester. The announced testers and feedback errors may depend on the
preceding history. At round $t$, the learner receives the classical
specification of a tester $E_t$, predicts $\widehat y_t\in[0,1]$, and
receives only a scalar $b_t\in[0,1]$ such that
\begin{equation}
 |b_t-f_P(E_t)|\leq\varepsilon/8.
 \label{pred:feedback-accuracy}
\end{equation}
For $N$ rounds, define
$M_\varepsilon=\sum_{t=1}^N
 \mathbf1_{\{|\widehat y_t-f_P(E_t)|>\varepsilon\}}$.
Let $\mathfrak M_\varepsilon(N;d,r,K)$ be the infimum, over
randomized learners, of the supremum of $\mathbb E M_\varepsilon$
over fixed targets and permitted tester and feedback histories.
\end{definition}

\begin{theorem}[Fixed-precision prediction law]
\label{pred:mistake-law}
For integers $d\geq2$, $r,K,N\geq1$, and
$0<\varepsilon<1/2$, put $n=dr$ and let $p_{n,K}$ be as in
Eq.~\eqref{up:physical-dimension}. Then
\begin{equation}
 \frac12\min\{N,L(K)\}
 \leq\mathfrak M_\varepsilon(N;d,r,K)
 \leq\min\left\{N,
 \left\lfloor p_{n,K}\log_2
       \left(1+\frac{16(K+1)}{\varepsilon}\right)\right\rfloor
       \right\}.
 \label{pred:mistake-bounds}
\end{equation}
In particular, for fixed $d,r,\varepsilon$,
\begin{equation}
 \mathfrak M_\varepsilon(N;d,r,K)
 =\Theta_{d,r,\varepsilon}
       \bigl(\min\{N,K\log(K+1)\}\bigr).
 \label{pred:mistake-order}
\end{equation}
\end{theorem}

\begin{proof}
For the upper bound, use a halving argument on a finite cover of
processes, maintaining its surviving centers and evaluating their responses.
By Proposition~\ref{up:physical-cover}, there is a proper $D$-cover
$V_0$ of radius $\varepsilon/4$ with
\[
 \log_2|V_0|\leq p_{n,K}\log_2
                 \left(1+\frac{16(K+1)}{\varepsilon}\right).
\]
It contains a fixed center $g_*$ whose response differs from the
target's by at most $\varepsilon/8$ on every tester. Starting with
$V_0$, predict a median of $(f_g(E_t))_{g\in V_{t-1}}$, with one entry
per surviving center, and retain
\[
 V_t=\{g\in V_{t-1}:|f_g(E_t)-b_t|\leq\varepsilon/4\}.
\]
The center $g_*$ survives every update, and every surviving center
has response within $3\varepsilon/8$ of $f_P(E_t)$. On a round
with error greater than $\varepsilon$, all survivors therefore
lie strictly on one side of that round's median prediction, leaving at most
half of the centers. Since $V_t$ remains nonempty, the number
of such rounds is at most $\lfloor\log_2|V_0|\rfloor$.

For the lower bound, choose a uniformly random leaf
$B\in\{0,1\}^{L(K)}$ of Theorem~\ref{low:phase-tree} before
prediction begins, independently of the learner's coins, fix its
witness $P_B$, and label the tree by its exact response bits.
At round $t\leq\min\{N,L(K)\}$, announce the tester determined
by $B_{<t}$ and give exact feedback $b_t=B_t$, while any remaining
rounds can use an arbitrary tester. On each round $t\leq\min\{N,L(K)\}$,
both the tester and all earlier scalar feedback depend only on the
prefix, so $B_t$ is a fair bit conditional on the preceding history
and the learner's coins.
For $\varepsilon<1/2$, any given scalar prediction lies more than
$\varepsilon$ from at least one of $0,1$, so this conditional fairness
gives each such round an error probability of at least $1/2$.
Averaging over leaves supplies one fixed target with
the asserted expected-error lower bound for the given learner.
The asymptotic statement follows from
$L(K)=\Theta(K\log(K+1))$ and $p_{n,K}=O_{d,r}(K)$.
\end{proof}

\begin{remark}[Sampling scalar feedback]
Set $a=\varepsilon/8$, $L=N$ and $\delta=\zeta$ in
Eq.~\eqref{inf:sample-count}. Then
\begin{equation}
 m=\left\lceil32\varepsilon^{-2}\ln\frac{2N}{\zeta}\right\rceil
 \label{pred:sample-cost}
\end{equation}
independent complete runs per announced tester implement
Eq.~\eqref{pred:feedback-accuracy} simultaneously over all $N$
adaptive rounds with confidence $1-\zeta$. The cost is $Nm$ complete
runs and $KNm$ transitions, plus the reset, control, readout and
classical computation costs of Appendix~\ref{inf:sampling}.
\end{remark}

\paragraph*{Prediction from single outcomes.}
Fix one target $P\in\mathcal P_{d,r}(K)$ before learning begins.
Let $\mathcal H_{t-1}$ contain the previously announced testers,
predictions and outcomes. Using this public history and random coins
independent of $P$, the tester-selection rule accesses the target through
the public outcomes. The learner also uses independent coins and receives
only tester specifications and final outcomes, retaining classical records
between rounds. Computation and finite tester workspace are unrestricted.

In round $t$, announce a tester $E_t$, predict
$\widehat p_t\in[0,1]$, and perform one freshly initialized complete
run, revealing its final binary outcome $Y_t$. Writing
$p_t=f_P(E_t)$, the feedback satisfies
\begin{equation}
 Y_t\mid P,\mathcal H_{t-1},E_t,\widehat p_t
 \sim\operatorname{Bernoulli}(p_t).
 \label{pred:bernoulli-feedback}
\end{equation}
Define
\begin{equation}
 \mathcal R_N(K;d,r)=\inf_{\mathcal A}\sup_{P,\mathcal E}
 \mathbb E\sum_{t=1}^N(\widehat p_t-f_P(E_t))^2,
 \label{pred:bernoulli-risk-definition}
\end{equation}
where the infimum is over learners $\mathcal A$ satisfying the preceding
protocol, and the supremum is over fixed targets $P$ and permitted
tester-selection rules $\mathcal E$. Single-outcome feedback and squared
probability error were also studied for fixed quantum states by
Aaronson et al.~\cite{aaronson2019online}. Here the physical cover,
the exact phase tree and independent replacer channels give the joint
dependence on process length and prediction rounds.

\begin{corollary}[Single-outcome prediction risk]
\label{pred:bernoulli-risk}
For every fixed $d\geq2$ and $r\geq1$, uniformly over integers
$K,N\geq1$,
\begin{equation}
 \mathcal R_N(K;d,r)
 =\Theta_{d,r}\!\left(\min\{N,K\ln(K+N)\}\right).
 \label{pred:bernoulli-order}
\end{equation}
\end{corollary}

\begin{proof}
The upper bound uses a cover and exponential weights. The lower bound takes
the larger of the phase-tree bound for successive response bits and
the independent-replacer bound for repeated probability estimation.
For $a>0$, Proposition~\ref{up:physical-cover} gives physical
centers $P_1,\ldots,P_S$ with one fixed center approximating the
target's responses within $a$ on every tester, and
\[
 \ln S\leq p_{dr,K}\ln[1+2(K+1)/a].
\]
Use exponential weighting for square loss: put
$g_{i,t}=f_{P_i}(E_t)$, start with $w_{i,1}=1/S$, predict
$\widehat p_t=\sum_iw_{i,t}g_{i,t}$, and update weights in
proportion to $w_{i,t}\exp[-(g_{i,t}-Y_t)^2/2]$.
For $y\in\{0,1\}$, the second derivative of
$\exp[-(x-y)^2/2]$ is
$((x-y)^2-1)\exp[-(x-y)^2/2]\leq0$ on $[0,1]$.
Jensen's inequality and the product of the weight normalizers
therefore give, for every realized history and every fixed $i$,
\[
 \sum_t(\widehat p_t-Y_t)^2
 \leq\sum_t(g_{i,t}-Y_t)^2+2\ln S.
\]
Conditional on the information preceding $Y_t$, including the
current tester and realized predictor randomness, Eq.~\eqref{pred:bernoulli-feedback}
implies, for every real $u$ determined by that information,
\begin{equation}
 \mathbb E[(u-Y_t)^2\mid\cdot]
 =(u-p_t)^2+p_t(1-p_t).
 \label{pred:brier-decomposition}
\end{equation}
Equation~\eqref{pred:brier-decomposition} separates probability risk from
the outcome's random variance: applying this identity to the prediction
and the same approximating center in every round, then taking expectations,
cancels the conditional variances and leaves risk at most $2\ln S+Na^2$.
Taking $a=N^{-1/2}$ gives
\[
 \mathcal R_N(K;d,r)
 \leq\min\{N/4,\,
  2p_{dr,K}\ln[1+2(K+1)\sqrt N]+1\},
\]
where the first bound uses the constant prediction $1/2$.
Since $p_{dr,K}=O_{d,r}(K)$, this proves the required upper order.

For the first lower bound, draw a uniformly random leaf of the
depth-$L(K)$ tree in Theorem~\ref{low:phase-tree} and fix its
process before learning. Along this tree the probabilities are
exactly zero or one, so the observed outcomes reveal the response
bits without error. The next tester depends only on the revealed prefix, so the uniformly drawn
leaf's next bit remains fair conditional on the learner's history and coins, giving
every prediction expected square error at least $1/4$ until the leaf or
the round budget is exhausted.
Averaging over leaves yields
$\mathcal R_N\geq\min\{N,L(K)\}/4$, with
$L(K)\geq(K/4)\log_2(K+1)$.

For the second lower bound, fix the initial output and use the memory-one
replacer channels
\[
 \Phi_i(X)=\Tr(X)\bigl[(1-\vartheta_i)|0\rangle\langle0|
                       +\vartheta_i|1\rangle\langle1|\bigr],
 \qquad 1\leq i\leq K,
\]
which are CPTP on the full input algebra. A tester querying coordinate $i$ measures its output,
keeps the resulting bit, completes all remaining transitions,
and reveals only that bit at the end. Cycle deterministically
through the coordinates and draw the parameters independently
and uniformly from $[0,1]$ before round one.

After $m$ observations of one coordinate, of which $s$ are one,
its posterior is $\operatorname{Beta}(s+1,m-s+1)$, with the count $s$
uniform on $\{0,\ldots,m\}$ under this prior. Other coordinates
and the learner's coins supply no information about this parameter,
so every predictor incurs expected square error at least
\[
 \frac1{m+1}\sum_{s=0}^m
 \frac{(s+1)(m-s+1)}{(m+2)^2(m+3)}
 =\frac1{6(m+2)}.
\]
Summing along the cyclic schedule gives
\[
 \mathcal R_N(K;d,r) \geq\frac K6\sum_{t=1}^N\frac1{t-1+2K} \geq\frac K6\ln(1+N/(2K)) \geq\frac K{12}\ln(1+N/K).
\]
Prior averaging again selects one fixed target for each learner.
Both lower-bound families embed into every $d\geq2,r\geq1$.

Finally, put $A=K\ln(K+1)$ and $B=K\ln(1+N/K)$.
The two lower bounds give constant multiples of $\min\{N,A\}$
and $B$. Since
$K\ln(K+N)\leq A+B$ and
$\min\{N,A\}+B\geq\min\{N,A+B\}$, their maximum has the
lower order in Eq.~\eqref{pred:bernoulli-order}.
\end{proof}

The protocol uses $N$ complete target runs, $KN$ unknown transitions and
$N$ initial outputs, together with resets, tester controls and readout.
The cover predictor maintains $S$ weights and updates them in
$O(S)$ arithmetic operations per round, in addition to computing
the $S$ known-center responses. Constructing the physical cover,
representing its centers and evaluating those responses require
further computation and precision.

\section{Temporal parameters and tester workspace}
\label{temp:section}

To separate tester workspace from independent temporal parameters,
we first restrict the tester to the visible qubit while retaining
independent phases, then retain the full tester class and require
one unknown phase channel to be repeated at every transition.

\paragraph*{Coherent reuse of the visible interface.}
Consider the memory-one qubit phase processes of
Eq.~\eqref{low:phase-process}, initialized with $|0\rangle$.
Between transitions, the tester retains the current visible qubit and
classical records, with no additional private quantum register, and may
apply known qubit operations between calls before performing a final binary
measurement. These permissions define the tester domain for the response
class $\mathcal F_{\rm vis}(K)$.

\begin{proposition}[Phase trees on the visible interface]
\label{temp:interface-tree}
For $K=2^\ell$ with $\ell\in\mathbb Z_{\geq0}$, the class
$\mathcal F_{\rm vis}(K)$ has an exact zero-or-one response tree
of depth $1+\ell K/2$. For every fixed $0<\gamma\leq1$ and
all integers $K\geq1$,
\begin{equation}
 \operatorname{sfat}_\gamma\mathcal F_{\rm vis}(K)
 =\Theta_\gamma(K\log(K+1)).
 \label{temp:interface-order}
\end{equation}
\end{proposition}

\begin{proof}
For a vector $s\in\{-1,1\}^K$, prepare the visible qubit in
$|+\rangle$ and use each forward target transition once, conjugating
the call by the known Pauli $X$ at a negative sign. Since
$XU_\theta X=e^{i\theta}U_{-\theta}$, the final equatorial
measurement gives the response in Eq.~\eqref{low:signed-response}
using the interface itself.

We need subset queries with no zero entries and known offsets.
Suppose $K=2^\ell$ and use the Boolean coordinates in
Eqs.~\eqref{low:mobius-forward}--\eqref{low:mobius-inverse}, now
indexed by subsets of $[\ell]$.
Choose $\alpha_\varnothing\in\{0,\pi\}$.
For nonempty $T\subseteq[\ell]$, the full-sign queries below double the
coefficient of $\alpha_T$ in Eq.~\eqref{low:dyadic-isolation}.
We use them to read $|T|$ bits per coordinate, choosing independently
\[
 \alpha_T=\frac{2\pi m_T}{2^{|T|+1}},
 \qquad m_T\in\{0,\ldots,2^{|T|}-1\}.
\]
The inverse formula fixes one complete phase vector for each choice.
For $J\subsetneq T$, put
\begin{equation}
 p_{T,J}(Z)=\prod_{i\in J}(2z_i(Z)-1)
       \left(2\prod_{i\in T\setminus J}z_i(Z)-1\right)
       \in\{-1,1\}.
 \label{temp:full-signs}
\end{equation}
This is $2q_{T,J}$ minus the character
$\prod_{i\in J}(2z_i-1)$, where $q_{T,J}$ is defined in
Eq.~\eqref{low:signed-polynomial}. Expanding that character
and using Eq.~\eqref{low:dyadic-isolation} gives, with $j=|J|$,
\begin{equation}
 \sum_Zp_{T,J}(Z)\theta_Z =2^{j+1}\alpha_T+\widetilde\beta_{T,J},\qquad \widetilde\beta_{T,J} =2\beta_{T,J} -\sum_{A\subseteq J}(-1)^{j-|A|}2^{|A|}\alpha_A.
 \label{temp:full-sign-isolation}
\end{equation}
All coordinates in $\widetilde\beta_{T,J}$ are indexed by proper subsets of $T$.

Reveal $\alpha_\varnothing$ with the all-positive query, then process
nonempty sets by increasing cardinality so each offset is known.
Read the bits of $m_T$ from least to most significant.
For bit $k\in\{0,\ldots,|T|-1\}$, choose a fixed
$J\subset T$ with $|J|=|T|-k-1$. Using the value
$r_k=m_T\bmod2^k$ specified by the preceding branches, set the
measurement phase to $\phi=\widetilde\beta_{T,J}+2\pi r_k/2^{k+1}$.
Writing $m_T=r_k+2^kb_k+2^{k+1}a$ makes the accumulated phase
minus $\phi$ equal to $\pi b_k$ modulo $2\pi$. Thus the
response is $1-b_k$, with threshold $1/2$ and signed margin
$1/2$. The number of bits is
\[
 1+\sum_{\varnothing\ne T\subseteq[\ell]}|T|
 =1+\frac{\ell K}{2}.
\]
Every node depends only on its preceding bits, while the entire
leaf determines one fixed phase vector.

For general $K$, use the largest dyadic prefix and fix all
remaining target phases to zero, completing the visible-qubit circuit
with their calls. The prefix has
length greater than $K/2$, giving the required lower order.
For $K=1$, there is a one-bit tree. The upper bound follows by restricting
the tester domain in Proposition~\ref{up:physical-cover} with
$d=2$ and $r=1$.
\end{proof}

\paragraph*{Repeated transition parameters.}
Let $\mathcal H$ be a family of channels on a fixed visible
space, with a fixed initial output. A process in
$\mathcal H^{[K]}$ applies the same unknown $\Phi\in\mathcal H$
at all $K$ transitions, with one-dimensional process memory.
The testers are the full adaptive class of
Section~\ref{main:model}. A half-diamond cover of radius
$\delta$ approximates every $\Phi$ by a channel $\Psi$ with
$\tfrac12\|\Phi-\Psi\|_\diamond\leq\delta$.
One cover center specifies every transition, at horizon-dependent accuracy.

\begin{proposition}[Repeated-channel cover]
\label{temp:repeated-cover}
Suppose that, for $0<\delta\leq1$, $\mathcal H$ has a
half-diamond cover of size at most $(A/\delta)^v$, with
$A\geq1$ and $v>0$ independent of $K$. Then, for
$K\geq1$ and $0<\gamma\leq1$,
\begin{equation}
 \operatorname{sfat}_\gamma\mathcal H^{[K]}
 \leq\left\lfloor v\log_2\frac{4AK}{\gamma}\right\rfloor.
 \label{temp:repeated-bound}
\end{equation}
\end{proposition}

\begin{proof}
Replacing each use of $\Phi$ by the same cover center $\Psi$ and
telescoping the complete channels bounds every adaptive tester response
error by $K\delta$, including arbitrary retained reference systems. Set $\delta=\gamma/(4K)$.
At the first divergence of two leaves of a
$\gamma$-shattered tree, their responses differ by at least
$\gamma$. Two witnesses with the same cover center would
differ there by at most $2K\delta=\gamma/2$. A depth-$L$
tree therefore requires at least $2^L$ centers, proving
the displayed bound.
\end{proof}

\begin{corollary}[One repeated phase]
\label{temp:repeated-phase}
For a memory-one process initialized with $|0\rangle$ and
applying the same unknown qubit phase channel
$U_\theta=\operatorname{diag}(1,e^{i\theta})$ at each of its
$K$ transitions, the sequential capacity under the full tester
class is $\Theta_\gamma(\log(K+1))$ for every fixed
$0<\gamma\leq1$.
\end{corollary}

\begin{proof}
A grid of $\lceil\pi/\delta\rceil$ equally spaced phases
is a half-diamond $\delta$-cover: the nearest phase has
circular distance at most $\pi/\lceil\pi/\delta\rceil$,
which bounds the half-diamond distance of their unitary
channels. For $\delta\leq1$ the grid size is at most
$(\pi+1)/\delta$, so Proposition~\ref{temp:repeated-cover}
gives the upper bound with $A=\pi+1$, $v=1$.

Put $a=\lfloor\log_2K\rfloor$ and choose
$\theta=2\pi m/2^{a+1}$ with $0\leq m<2^{a+1}$.
Reveal the bits of the chosen integer label from least to most
significant. To query
bit $k$, use the tester of Eq.~\eqref{low:signed-response}
with $2^{a-k}$ positive signs and zero at the other
transitions. The control qubit holds the probe
while inactive transitions receive a basis input.
Subtracting the known phase $2\pi(m\bmod2^k)/2^{k+1}$ makes the
response exactly the complement of bit $k$. One fixed $\theta$
realizes each complete leaf of depth $a+1$.
\end{proof}

\section{Extensions to growing memory}
\label{main:growing}

We now allow $r=r(K)$ and keep the memory cap at every transition.
The following unitary and addressing constructions give the explicit
lower bounds in Table~\ref{main:regimes} for the full process class
$\mathcal P_{d,r}(K)$.

\subsection{Unitary columns and their implementation time}
\label{low:unitary-section}

Larger process memory can store a state on which several visible inputs
act collectively, and we use this storage to query columns of a unitary.
The inverse unitary occurs in specified process transitions, and both
the forward and inverse operations require time to load and return the
state through the visible ports.

\begin{theorem}[Unitary-column lower bound]
\label{low:unitary-lower}
Let $d\geq2$, $\ell\geq1$, $R=d^\ell\leq r$, and $K\geq0$ be integers.
Put
\begin{equation}
 Q=\lceil\sqrt R\rceil,\qquad K_0=2\ell Q,
 \qquad G=\left\lfloor\frac K{K_0}\right\rfloor.
 \label{low:unitary-parameters}
\end{equation}
For every $0<\gamma\leq1/16$,
\begin{equation}
 \operatorname{sfat}_{\gamma}\bigl(\mathcal P_{d,r}(K)\bigr)
 \geq G\frac{R(R-1)}2.
 \label{low:unitary-bound}
\end{equation}
The witnesses have fixed initial visible output $|0\rangle$ and an
$R$-dimensional process memory.  If $K\geq K_0$, then
\begin{equation}
 \operatorname{sfat}_{\gamma}\bigl(\mathcal P_{d,r}(K)\bigr)
 \geq\frac{KR^{3/2}}{32\ell}.
 \label{low:unitary-simplified}
\end{equation}
\end{theorem}

The column family supplies independent signs whose response gap decreases
as the number of signs per column grows. A decoder amplifies this gap
to a constant, and streaming its unitary calls through $d$-dimensional
ports gives the block length in
\eqref{low:unitary-parameters}.

\begin{lemma}[A sequential family of unitary columns]
\label{low:column-family}
In dimension $R\geq2$, there is a family of real orthogonal matrices $U$
indexed by $R(R-1)/2$ independent signs.  When the signs are ordered by
column, a query for any sign in column $j\in\{0,\ldots,R-2\}$ uses a known
unit vector $g$ and input $|j\rangle$ with
\begin{equation}
 |\langle g|U|j\rangle|^2
 =p_\pm=\frac{(1\pm N^{-1/2})^2}{4},
 \qquad N=R-j-1\geq1.
 \label{low:column-probabilities}
\end{equation}
The vector $g$ depends only on earlier columns, and the two probabilities
are independent of all other signs of the current column and all later
columns.
\end{lemma}

\begin{proof}
Given columns $0,\ldots,j-1$, select an ordered orthonormal basis
$e_0,\ldots,e_N$ of their orthogonal complement by a fixed deterministic
rule.  For example, project the standard coordinate vectors in order,
orthogonalize against the vectors already selected, skip zero residuals,
and divide each nonzero residual by its positive norm.  For $j\leq R-2$,
choose signs $s_{j,i}\in\{-1,1\}$ for $1\leq i\leq N$ and set
\begin{equation}
 \psi_j=\frac1{\sqrt2}
       \left(e_0+\frac1{\sqrt N}\sum_{i=1}^N s_{j,i}e_i\right).
 \label{low:column-vector}
\end{equation}
This unit vector is orthogonal to all earlier columns, so choosing the
last column by the same basis rule in the remaining one-dimensional complement
completes an orthogonal matrix $U$ with
$\sum_{j=0}^{R-2}(R-j-1)=R(R-1)/2$ signs.

Order the signs first by $j$, then by $i$.  At the node for $(j,i)$,
the preceding branches determine all earlier columns and therefore
$g=(e_0+e_i)/\sqrt2$.  Equation~\eqref{low:column-vector} gives
$\langle g|U|j\rangle=(1+s_{j,i}/\sqrt N)/2$, proving the claim.
\end{proof}

\begin{lemma}[Constant-gap column decoder]
\label{low:column-decoder}
For each integer $N\geq1$, the two cases in
Eq.~\eqref{low:column-probabilities} can be separated by a tester using
an odd number $k\leq\lceil\sqrt{N+1}\rceil$ of alternating operations
$U,U^\dagger,U,\ldots,U$, interspersed with known reflections.
Its two success probabilities satisfy $q_+-q_->1/16$.
\end{lemma}

\begin{proof}
Let $S_g=I-2|g\rangle\langle g|$ and
$S_j=I-2|j\rangle\langle j|$.  Starting from $U|j\rangle$, apply
$m$ iterations of $-US_jU^\dagger S_g$, each using $S_g$, then $U^\dagger$,
then $S_j$, then $U$ in chronological order, with both reflections
specified by the node.  The identity
\begin{equation}
 US_jU^\dagger=I-2|\psi_j\rangle\langle\psi_j|
 \label{low:reflection-cancellation}
\end{equation}
removes dependence on the other columns from each complete iteration.
This is the amplitude-amplification construction of
Brassard, H\o yer, Mosca, and Tapp~\cite{brassard2002amplitude}.

Write $\psi_j=\sin\theta\,g+\cos\theta\,h$ in its plane with $g$,
where $h$ is a unit vector orthogonal to $g$.  In this ordered basis,
the iteration has matrix
\[
 \begin{pmatrix}
 \cos(2\theta)&\sin(2\theta)\\
 -\sin(2\theta)&\cos(2\theta)
 \end{pmatrix}.
\]
After $m$ iterations, measuring $|g\rangle\langle g|$ therefore gives
\begin{equation}
 q_\pm=\sin^2(k\theta_\pm),\qquad
 k=2m+1,\qquad
 \theta_\pm=\arcsin\frac{1\pm u}{2},\quad u=N^{-1/2}.
 \label{low:amplified-probabilities}
\end{equation}
The endpoint cases with probability zero or one follow directly from
the reflections, or by continuity.

For $1\leq N<256$, choosing $k=1$ gives $q_+-q_-=u>1/16$.

For $N\geq256$, taking $k$ near $\sqrt N/4$ and one modulo six
amplifies the angle difference while keeping the sine of the angle
sum away from zero, since $\theta_++\theta_-$ is close to $\pi/3$. Set
\begin{equation}
 a=\left\lceil\frac{\sqrt N/4-1}{6}\right\rceil,
 \qquad k=6a+1.
 \label{low:decoder-call-count}
\end{equation}
The chosen call count is positive and odd, and
\begin{equation}
 \frac{\sqrt N}{4}\leq k<\frac{\sqrt N}{4}+6,
 \qquad \frac14\leq ku<\frac14+6u\leq\frac58.
 \label{low:call-count-bounds}
\end{equation}
In particular, $k<\sqrt N\leq\lceil\sqrt{N+1}\rceil$.

For $f(x)=\arcsin(x/2)$ on $0\leq x\leq5/4$,
\[
 \frac12\leq f'(x)=\frac1{\sqrt{4-x^2}}\leq\frac23,
 \qquad
 0\leq f''(x)=\frac{x}{(4-x^2)^{3/2}}\leq\frac13.
\]
For the last bound, $f''$ is increasing and
$f''(5/4)=80/(39\sqrt{39})<1/3$, as is seen by squaring the positive
quantities: $57600<59319$.  Since $u\leq1/16$, these derivative bounds
apply on $[1-u,1+u]$.  Integration yields
\begin{equation}
 u\leq\delta:=\theta_+-\theta_-\leq\frac{4u}{3},
 \qquad
 0\leq\sigma:=\theta_++\theta_--\frac\pi3\leq\frac{u^2}{3}.
 \label{low:decoder-angle-bounds}
\end{equation}
Indeed, the second difference is
$\int_0^u(u-t)[f''(1+t)+f''(1-t)]\,dt$.
Equations~\eqref{low:call-count-bounds} and
\eqref{low:decoder-angle-bounds} give
\[
 \frac14\leq k\delta\leq\frac56<\frac\pi2,
 \qquad 0\leq k\sigma\leq\frac5{384}.
\]
Because $k\equiv1\pmod6$, the angle $k(\theta_++\theta_-)$ modulo
$2\pi$ lies in $[\pi/3,\pi/3+5/384]\subset[\pi/3,\pi/2]$.
The sine subtraction identity now gives
\begin{equation}
 q_+-q_- =\sin\bigl(k(\theta_++\theta_-)\bigr)\sin(k\delta) \geq\sin(\pi/3)\sin(1/4) >\frac34\frac15=\frac3{20}>\frac1{16}.
 \label{low:decoder-gap}
\end{equation}
Here $\sin(1/4)\geq1/4-(1/4)^3/6=95/384>1/5$.

In either case, the midpoint threshold $(q_++q_-)/2$ depends only on $N$
and the chosen $k$ and gives a signed margin greater than $1/32$ on both
branches.
\end{proof}

\begin{lemma}[Streaming a scheduled unitary]
\label{low:streaming}
For integers $d\geq2$ and $\ell\geq1$, set $R=d^\ell$ and let $V$ be any
fixed unitary on $\mathbb C^R$.
A process with an $R$-dimensional memory can implement one use of $V$
on a tester-held register, including its correlations with arbitrary
references, in $2\ell$ transitions of port dimension $d$.  At the end
of these transitions its memory is reset to a fixed state on every input.
\end{lemma}

\begin{proof}
Identify the process's propagated memory with $\ell$ registers $M_0,\ldots,M_{\ell-1}$,
each of dimension $d$, initially in $|0\rangle$, and factor the tester's
register $E$ in the same way. During loading transition $a$, swap the
visible input with $M_a$, emitting the post-swap visible carrier, and at
the final loading transition also apply $V$ to the whole memory after
the swap. Writing $S_a$ for the swap, that final channel is the unitary
$(I_B\otimes V_M)S_a$, so sending the factors of $E$ in order loads its
whole state, including its reference entanglement. All loading outputs
are zero in this experiment.

During the next $\ell$ transitions, send zero inputs and swap them with
$M_0,\ldots,M_{\ell-1}$ in the same order, retaining the outputs to
recover $V$ applied to $E$.  At the final unloading transition, follow
the swap by discarding the remaining memory and preparing $|0\rangle_M$.
Explicitly, its Kraus operators are
\begin{equation}
 F_z=(I_B\otimes|0\rangle_M\langle z|)S_{\ell-1},
 \qquad z=0,\ldots,R-1,
 \label{low:reset-kraus}
\end{equation}
and $\sum_zF_z^\dagger F_z=I$.  In the specified zero-input experiment,
the memory is already zero after the swaps, so this last operation
preserves the implemented channel.  For arbitrary other inputs it
still resets the memory and removes all its correlations.  Thus each
transition is CPTP on its entire domain, with arbitrary references.
\end{proof}

\begin{proof}[Proof of Theorem~\ref{low:unitary-lower}]
A block contains $Q$ consecutive streaming slots, each of duration
$2\ell$. Fix a matrix from Lemma~\ref{low:column-family} and use
odd slots to apply $U$ and even slots to apply $U^\dagger$, as in
Lemma~\ref{low:streaming}. This schedule is the same for every query
to that matrix.  At the node for a column sign, the tester executes
Lemma~\ref{low:column-decoder} during the first $k$ slots, applying all
known reflections to its returned register between slots.
The bound $N+1\leq R$ ensures that these slots are available.
After its measurement, the tester retains the result privately and
completes the remaining slots using zero inputs and discarded outputs.
Normalized operations on the other systems leave the retained result's
probability unchanged.

Use $G$ such blocks with independent matrices and pad the remaining
$K-GK_0$ transitions by fixed discard-and-prepare channels.  Every slot
resets the process memory, so every block begins in its required state.
To query a bit in one block, the tester sends zero inputs and discards
the outputs of preceding blocks, executes the decoder in the selected
block, and retains its outcome through the rest of the process.
The initial visible output is fixed to zero and ignored.  All block
indices and all forward or inverse choices occur in the fixed
time-dependent channel descriptions.

Order the tree bits by block, column, and sign coordinate. Each complete
path determines every column of every matrix, so installing these
matrices and their full alternating schedules once gives a single
complete process.  The cancellation in Eq.~\eqref{low:reflection-cancellation}
and Lemma~\ref{low:column-family} ensure that each node's response is
independent of all suffix signs.  Thus this same process realizes every
constraint on its path, with margin at least $\gamma/2$ for
$0<\gamma\leq1/16$.  There are $GR(R-1)/2$ levels, proving
Eq.~\eqref{low:unitary-bound}.

For $K\geq K_0$, use
$\lfloor K/K_0\rfloor\geq K/(2K_0)$,
$R(R-1)/2\geq R^2/4$, and $Q\leq2\sqrt R$.  Then
\[
 G\frac{R(R-1)}2
 \geq\frac{KR^2}{16\ell Q}
 \geq\frac{KR^{3/2}}{32\ell},
\]
which proves Eq.~\eqref{low:unitary-simplified}.
\end{proof}

\subsection{Addressing independent processes within one memory cap}
\label{low:direct-sum-section}

An address can select one of several independent processes while using
the same visible ports, with its storage cost incorporated by a direct
sum of the branch memories.  The dense-coding construction reduces the loading time and accounts for
the preparation and storage of the initial entangled pair.

For integers $A\geq1$, define
\begin{equation}
 a_d(A)=\left\lceil\log_{d^2}A\right\rceil.
 \label{low:address-time}
\end{equation}
In particular, $a_d(1)=0$.

\begin{theorem}[Direct sums of initialized processes]
\label{low:direct-sum}
Let $d\geq2$, $H\geq0$, and $A\geq1$ be integers, and let
$R_1,\ldots,R_A$ be positive integers.  Suppose that, for each
$x\in\{1,\ldots,A\}$, a nonempty family
$\mathcal F_x\subseteq\mathcal P_{d,R_x}(H)$ has a finite
$\gamma$-shattered tree of depth $L_x\geq0$, where $0<\gamma\leq1$.
Then
\begin{equation}
 \operatorname{sfat}_{\gamma}
 \left(\mathcal P_{d,\sum_xR_x}(H+a_d(A))\right)
 \geq\sum_{x=1}^A L_x.
 \label{low:direct-sum-bound}
\end{equation}
Consequently, for a common branch cap $R$,
\begin{equation}
 \operatorname{sfat}_{\gamma}
       \bigl(\mathcal P_{d,AR}(H+a_d(A))\bigr)
 \geq A\operatorname{sfat}_{\gamma}\bigl(\mathcal P_{d,R}(H)\bigr).
 \label{low:equal-direct-sum}
\end{equation}
\end{theorem}

\begin{proof}
Dense coding loads the address before the selected initialized branch
runs, with the memory bound including both the decoded prefix and the
retained Bell half. The case $A=1$ requires no loading.  Suppose that $A\geq2$, and put
$q=\min(d,A)$ and $b=dq$.  Let $a$ be the smallest positive integer
with $b^a\geq A$.  If $2\leq A<d$, both $a$ and $a_d(A)$ equal one.
If $A\geq d$, then $b=d^2$.  Thus $a=a_d(A)$ in every case.

Prepare initially the visible-memory state
\begin{equation}
 |\Phi_q\rangle=\frac1{\sqrt q}
           \sum_{j=0}^{q-1}|j\rangle_{B_0}|j\rangle_M.
 \label{low:initial-bell}
\end{equation}
The retained Bell half has dimension $q$ and is part of the process
memory.  A visible phase on the first $q$ levels, extended by the identity on
the unused levels and followed by a cyclic shift, can encode
$(u,e)\in\{0,\ldots,d-1\}\times\{0,\ldots,q-1\}$ as
\begin{equation}
 |\beta_{u,e}\rangle=\frac1{\sqrt q}
 \sum_{j=0}^{q-1}e^{2\pi i e j/q}|j+u\bmod d\rangle|j\rangle.
 \label{low:rectangular-bell}
\end{equation}
These $dq$ vectors are an orthonormal basis:
the second register in an overlap forces the same $j$, the first
register then forces the same $u$, and the roots-of-unity sum forces
the same $e$.  This rectangular dense-coding alphabet is due to
Yan and Wang~\cite{yan2004dense}.

Select the first $A$ length-$a$ words over this alphabet in lexicographic
order.  Let $V_t$ be their distinct prefixes of length $t$.  For
$0\leq t<a$,
\begin{equation}
 |V_t|=\left\lceil\frac{A}{b^{a-t}}\right\rceil.
 \label{low:prefix-count}
\end{equation}
At nonfinal loading transition $t\in\{1,\ldots,a-1\}$, the memory
consists of the old prefix and one Bell half.  Measure the incoming
visible register and Bell half in the complete basis
Eq.~\eqref{low:rectangular-bell}, together with the old prefix basis.
Retain the extended prefix if it belongs to $V_t$. Otherwise retain
a fixed member of $V_t$.  Emit a fresh visible half of
$|\Phi_q\rangle$ and retain its partner. For an old prefix
$p\in V_{t-1}$ and a Bell outcome $z=(u,e)$, write $g_t(p,z)$
for the updated prefix, including the stated fallback for an invalid extension.
On the supported loading subspace the Kraus maps have the form
\begin{equation}
 E_{p,z}=
 \bigl(|g_t(p,z)\rangle\otimes|\Phi_q\rangle\bigr)
 \bigl(\langle p|\otimes\langle\beta_z|\bigr),
 \label{low:prefix-kraus}
\end{equation}
up to the ordering of the output registers.  Their output vectors
are normalized and their input bras form an orthonormal basis, so their
adjoint products sum to the identity on that subspace.

The largest memory dimension before the final write is
$q\lceil A/b\rceil\leq A$.  Write $A=bu+v$ with
$0\leq v<b$.  If $v=0$, the inequality is $qu\leq bu$.
If $u=0$, it is $q\leq A$.  If $u\geq1$ and $v>0$, then
\[
 A-q(u+1)=q(d-1)u+v-q\geq q(d-2)+1\geq1.
\]
Earlier prefix counts are smaller, and the initial memory costs
$q\leq A$ as well.

Choose one branch process $C_x\in\mathcal F_x$ for each address.
Pad its realization, if necessary, to a memory space $M_x$ of
dimension $R_x$ and write $\sigma_x$ for its initialized state on
$\mathbb C^d\otimes M_x$. With $J_x$ the canonical inclusion of $M_x$, set
\begin{equation}
 M=\bigoplus_{x=1}^A M_x,
 \qquad J_x:M_x\longrightarrow M,
 \qquad\dim M=\sum_xR_x.
 \label{low:memory-direct-sum}
\end{equation}
Before the final loading transition, only the prefix and Bell-half
subspaces of $M$ are used. At that transition, perform the prefix/Bell
measurement to decode $x$, assigning invalid words to $x=1$, and
prepare $(I\otimes J_x)\sigma_x(I\otimes J_x^\dagger)$ so the
direct-sum block carries the address.
For a spectral decomposition
$\sigma_x=\sum_\nu\lambda_{x,\nu}|v_{x,\nu}\rangle
\langle v_{x,\nu}|$, explicit Kraus maps on the loading subspace are
\begin{equation}
 \sqrt{\lambda_{x,\nu}}
 (I_B\otimes J_x)|v_{x,\nu}\rangle
 (\langle p|\otimes\langle\beta_z|),
 \label{low:branch-initialization-kraus}
\end{equation}
where $x$ is the address decoded from $(p,z)$.
For each input basis vector, the sum of squared output norms is one.
Complete the input basis outside the loading subspace and map each
remaining basis vector to one fixed normalized output state by
discard-and-prepare Kraus maps.  Use the same type of completion for
Eq.~\eqref{low:prefix-kraus} and for padded branch channels. The completed
maps are CPTP on the full input space, including in the presence of an
arbitrary reference.

Let the Kraus maps of transition $t$ of branch $x$ be
$K_{x,t,\alpha}:\mathbb C^d\otimes M_x\to\mathbb C^d\otimes M_x$.
Use global Kraus maps
\begin{equation}
 L_{x,t,\alpha}=(I_B\otimes J_x)K_{x,t,\alpha}
                              (I_A\otimes J_x^\dagger).
 \label{low:branch-channel-kraus}
\end{equation}
Their adjoint products sum to
$I_A\otimes\sum_xJ_xJ_x^\dagger=I_A\otimes I_M$.
Thus the channel dephases distinct blocks and applies the branch
channel within each block, with a CPTP action on coherent inputs
and arbitrary references.  Once address $x$ has been loaded, its
state remains in the $x$th block and every subsequent transition
agrees exactly with that of $C_x$.

The tester loads any selected address deterministically by applying
the known encodings of its word to the successive emitted Bell halves.
The last loading output is the selected branch's initial visible
output and is followed by the branch's $H$ transitions, giving a combined
horizon of exactly $a+H$, including when $H=0$.

To query a node of the $x$th tree, first load $x$, then run the
original branch tester on its initialized visible output and $H$
transitions, preserving its response. Concatenate the $A$ trees and,
for each complete concatenated path, choose a witness $C_x$ for each
complete branch path, then install all their initialized states and
transitions in the single process just constructed. This process
realizes every constraint on the entire concatenated path, proving
Eq.~\eqref{low:direct-sum-bound}.  Taking arbitrary finite
shattered depths in identical branch classes gives
Eq.~\eqref{low:equal-direct-sum}. Finiteness follows from
\eqref{up:cover-capacity} for positive horizons and from a finite trace-norm
cover of the initial visible states at horizon zero, when only $B_0$
is delivered.
\end{proof}

\begin{corollary}[Addressed unitary-column lower bound]
\label{low:addressed-unitary}
Let $d\geq2$, $\ell\geq1$, $R=d^\ell$, and $A\geq1$ be integers
with $AR\leq r$ and $K\geq a_d(A)$.  Set
$H=K-a_d(A)$.  For $0<\gamma\leq1/16$,
\begin{equation}
 \operatorname{sfat}_{\gamma}\bigl(\mathcal P_{d,r}(K)\bigr)
 \geq A\left\lfloor\frac{H}{2\ell\lceil\sqrt R\rceil}\right\rfloor
                         \frac{R(R-1)}2.
 \label{low:addressed-bound}
\end{equation}
\end{corollary}

\begin{proof}
Apply Theorem~\ref{low:direct-sum} to $A$ copies of the family in
Theorem~\ref{low:unitary-lower}, each on the remaining $H$ transitions.
Their initialized visible and memory states are all fixed to zero.
Thus the entire loader is independent of the column signs.
\end{proof}

\begin{theorem}[A uniform lower bound when the loading time is available]
\label{low:sufficient-regime}
For integers $d\geq2$, $r\geq d$, and $K\geq0$, define
\begin{equation}
 t=\lceil\log_{d^2}r\rceil,\quad h=K-t,\quad
 s=\lfloor\log_dr\rfloor,\quad
 \lambda_u=\lceil\log_d(u+1)\rceil\quad(u\geq1).
 \label{low:regime-parameters}
\end{equation}
If $h\geq2\lceil\sqrt d\rceil$ and $0<\gamma\leq1/16$, then
\begin{equation}
 \operatorname{sfat}_{\gamma}\bigl(\mathcal P_{d,r}(K)\bigr)
 \geq\frac{rh}{1024\sqrt d}
       \min\left\{\frac{\sqrt r}{s},\frac h{\lambda_h^2}\right\}.
 \label{low:uniform-h-bound}
\end{equation}
In particular, if $2t\leq K$, $K\geq4\lceil\sqrt d\rceil$, and
$0<\gamma\leq1/16$, then
\begin{equation}
 \operatorname{sfat}_{\gamma}\bigl(\mathcal P_{d,r}(K)\bigr)
 \geq\frac{rK}{4096\sqrt d}
       \min\left\{\frac{\sqrt r}{s},\frac K{\lambda_K^2}\right\}.
 \label{low:uniform-K-bound}
\end{equation}
\end{theorem}

\begin{proof}
For $\ell\in\mathbb Z_{\geq1}$, set
$R_\ell=d^\ell$ and $T_\ell=2\ell\lceil\sqrt{R_\ell}\rceil$.
Since the streaming-time sequence is strictly increasing and unbounded
and the hypothesis on $h$ makes the following set nonempty, define
the finite maximum $m$ and choose the construction parameters by
\[
 m=\max\{\ell\in\mathbb Z_{\geq1}:T_\ell\leq h\},\qquad
 \ell=\min(s,m),\qquad R=d^\ell,\qquad A=\lfloor r/R\rfloor.
\]
Here $s$ and $m$ impose the memory and streaming-time limits
met by the choice $\ell=\min(s,m)$, and $A$ counts the addressed copies.
Since $R\leq r$, these
choices satisfy $A\geq1$ and $AR\leq r$.
Equation~\eqref{low:address-time} is monotone in $A$, so
$a_d(A)\leq t$ and $H=K-a_d(A)\geq h\geq T_\ell$.
These memory and loading-time bounds allow the addressed construction
to give the integer depth
\[
 F=A\left\lfloor\frac H{T_\ell}\right\rfloor\frac{R(R-1)}2.
\]
Since all ratios to which the floor bound is applied are at least one,
we may use $\lfloor x\rfloor\geq x/2$ in the following estimates:
\[
 A\geq\frac r{2R},\qquad
 \left\lfloor\frac H{T_\ell}\right\rfloor\geq\frac h{2T_\ell},
 \qquad \frac{R(R-1)}2\geq\frac{R^2}{4},\qquad
 \lceil\sqrt R\rceil\leq2\sqrt R.
\]
Their product yields
\begin{equation}
 F\geq\frac{rh\sqrt R}{64\ell}.
 \label{low:rounded-depth}
\end{equation}

If $s\leq m$, then $\ell=s$ and $R=d^s>r/d$, giving
\begin{equation}
 F\geq\frac{r^{3/2}h}{64\sqrt d\,s}.
 \label{low:memory-limited-depth}
\end{equation}
If $s>m$, then $\ell=m$.  Maximality of $m$ gives
\[
 h<2(m+1)\lceil\sqrt{d^{m+1}}\rceil
   \leq4\sqrt d\,(m+1)\sqrt{d^m}.
\]
Substitution in Eq.~\eqref{low:rounded-depth} implies
\begin{equation}
 F\geq\frac{rh^2}{256\sqrt d\,m(m+1)}.
 \label{low:time-limited-depth}
\end{equation}
Since $2m d^{m/2}\leq h$,
\[
 m\leq2\log_d\frac h{2m}
   <2\log_d(h+1)\leq2\lambda_h.
\]
Integrality therefore gives $m+1\leq2\lambda_h$, and hence
$m(m+1)\leq4\lambda_h^2$.  Applying this to
Eq.~\eqref{low:time-limited-depth}, and taking the smaller scale
of the two cases, proves Eq.~\eqref{low:uniform-h-bound}.

Under the second set of hypotheses, $h\geq K/2\geq2\lceil\sqrt d\rceil$
and $\lambda_h\leq\lambda_K$.  Applying
$h/\lambda_h^2\geq K/(2\lambda_K^2)$ and
$\min(x,y/2)\geq\min(x,y)/2$ in
Eq.~\eqref{low:uniform-h-bound} gives Eq.~\eqref{low:uniform-K-bound}.
\end{proof}

\begin{corollary}[Quadratically growing qubit memory]
\label{low:quadratic-memory}
For every integer $K\geq8$ and $0<\gamma\leq1/16$,
\begin{equation}
 \operatorname{sfat}_{\gamma}\bigl(\mathcal P_{2,K^2}(K)\bigr)
 \geq\frac{K^4}{4096\sqrt2\,
                   \lceil\log_2(K+1)\rceil^2}.
 \label{low:quadratic-memory-bound}
\end{equation}
\end{corollary}

\begin{proof}
Take $d=2$ and $r=K^2$ in Theorem~\ref{low:sufficient-regime}.
Then $t=\lceil\log_2K\rceil$, and $2t\leq K$ for $K\geq8$.
Indeed, if $n=\lceil\log_2K\rceil$, the case $n=3$ is $K=8$.
For $n\geq4$, $K\geq2^{n-1}+1\geq2n$.
The other horizon condition is $K\geq4\lceil\sqrt2\rceil=8$.
Writing $\lambda=\lceil\log_2(K+1)\rceil$ and
$s=\lfloor2\log_2K\rfloor$, we have
$\lambda\geq2$ and
$\lambda^2\geq2\lambda\geq2\log_2K\geq s$.
Thus the smaller term in Eq.~\eqref{low:uniform-K-bound} is
$K/\lambda^2$, which yields Eq.~\eqref{low:quadratic-memory-bound}.
\end{proof}

For $d=2$, $r=K^2$ and fixed $0<\gamma\leq1/16$, the program
bound in Theorem~\ref{up:optimized-program} and the binomial estimate
in Eq.~\eqref{up:optimized-simplified}, evaluated at $n=2K^2$ and
$\eta=\gamma/2$, also give
\begin{equation}
 C_\gamma(K;2,K^2)=O_\gamma(K^{11/2}\log K).
 \label{main:quadratic-upper}
\end{equation}
The lower bound in Eq.~\eqref{low:quadratic-memory-bound} rules out
a uniform estimate $CKr[\log(K+1)+\log(r+1)]^a$ at these fixed
$d,\gamma$ for any constants $C,a\geq0$: along $r=K^2$ the proposed
bound is $O(K^3\log^aK)$, whereas the lower bound is
$\Omega(K^4/\log^2K)$.

\section{General quantum-program bounds}
\label{main:upper}

The constructions in this appendix supply the growing-memory upper
estimates. Program storage is an auxiliary encoding resource, separate
from the target's physical memory. Throughout, $d\geq2$, $r\geq1$
and $K\geq1$ are integers and $0<\gamma\leq1$. Write $n=dr$ and use the
information increment $c_z$ from Eq.~\eqref{up:information-increment}.
The Kraus rank is the minimum number of operators in a Kraus representation.

\subsection{Finite mixtures of channels with small Kraus rank}

The use of low-rank channel mixtures and common program spaces has
antecedents in quantum programming
\cite{gschwendtner2021programmability,yoshida2026isometry}.
The rank decomposition gives each component an $n$-dimensional
environment. Programs in one common space then encode channel mixtures
as mixed states.

\begin{lemma}[Finite rank decomposition]
\label{up:rank-decomposition}
Every channel from $\mathbb C^n$ to itself is a finite convex combination
of channels with Kraus rank at most $n$.
\end{lemma}

\begin{proof}
Let $J\geq0$ be its unnormalized Choi matrix, ordered with the output
factor first, so that $\operatorname{Tr}_{\rm out}J=I_n$. If its rank is $a>n$, the real space of Hermitian matrices supported
on $\operatorname{ran}J$ has dimension $a^2>n^2$, so the partial-trace
map has a nonzero Hermitian kernel vector $D$ in that space.
Since $\operatorname{Tr}D=0$, this vector is indefinite, so on
$\operatorname{ran}J$ the matrix
$W=J^{-1/2}DJ^{-1/2}$ has eigenvalues of both signs. Set
\[
 t_+=-\lambda_{\min}(W)^{-1},\qquad
 t_-=\lambda_{\max}(W)^{-1}.
\]
The matrices $J_+=J+t_+D$ and $J_-=J-t_-D$ are positive,
have partial trace $I_n$, and have ranks strictly smaller than $a$.
The original Choi matrix is their convex combination,
\[
 J=\frac{t_-}{t_++t_-}J_+
     +\frac{t_+}{t_++t_-}J_-.
\]
Repeating this splitting on each term whose rank exceeds $n$ gives
a finite decomposition, since every branch has at most $a-n$ rank decreases. Spectral factorization of a terminal Choi
matrix gives at most $n$ Kraus operators $A_j$, satisfying
$\sum_jA_j^*A_j=I_n$. Padding with zero operators realizes the channel
by the isometry
\[
 V:\mathbb C^n\longrightarrow\mathbb C^n\otimes\mathbb C^n,
 \qquad V|u\rangle=\sum_{j=1}^n A_j|u\rangle\otimes|j\rangle.
\]
Tracing the second factor gives the required channel.
\end{proof}

\subsection{A finite port-based program}

Port-based teleportation transfers an input to one selected output port
without a port-dependent correction on that port. To improve the finite
error bound with a suitable correlated resource, we use the finite
covariant construction of Refs.~\cite{yang2020optimal,yoshida2026pbt}.

\begin{theorem}[Finite quantum-program upper bound]
\label{up:optimized-program}
For $0<\eta<\gamma$, define
\begin{equation}
 t=\sqrt{K/\eta},\qquad
 m=\lceil24n^2t\rceil,\qquad N=m+1,\qquad
 B=\binom{N+n^3-1}{N}.
 \label{up:program-parameters}
\end{equation}
Then
\begin{equation}
 \operatorname{sfat}_\gamma(\mathcal P_{d,r}(K))
 \leq
 \left\lfloor\frac{\log_2n+K\log_2B}{c_{\gamma-\eta}}\right\rfloor.
 \label{up:optimized-capacity}
\end{equation}
The binomial factor obeys the finite estimate
\begin{equation}
 \log_2B\leq25n^2t\log_2\!\left[e\left(1+\frac{n}{24t}\right)\right].
 \label{up:optimized-simplified}
\end{equation}
\end{theorem}

Lemmas~\ref{up:pbt-error} and~\ref{up:compression} bound the full
diamond error and show that the common program dimension is at most $B$,
including for channel mixtures. After simulating $K$ transitions, Lemma~\ref{up:information}
gives the capacity bound. Equation~\eqref{main:quadratic-upper} applies it
to the growing-memory regime in Table~\ref{main:regimes}.

\begin{lemma}[Covariant finite-port error]
\label{up:pbt-error}
For the integers $n,m,N$ in \eqref{up:program-parameters}, there is a
finite pure resource on $N$ Alice ports and $N$ Bob ports, each of
dimension $n$, and a complete deterministic correction-free
port-based teleportation channel $\mathcal T_N$ such that
\begin{equation}
 \|\mathcal T_N-\operatorname{id}\|_\diamond
 \leq\frac{4\pi^2(n-1)^4(3n-2)^2}{n^2m^2}
 <\frac{576n^4}{m^2}\leq\frac\eta K.
 \label{up:pbt-diamond}
\end{equation}
\end{lemma}

\begin{proof}
We bound partition-weight overlaps, transfer their score to teleportation
fidelity, and use covariance to obtain full diamond error.
Let $\mathcal Y_a^{(n)}$ denote the partitions of $a$ into at most
$n$ rows, with zero padding. Define the incidence matrix $R$ with rows
in $\mathcal Y_m^{(n)}$ and columns in $\mathcal Y_{m+1}^{(n)}$ by
$R_{\alpha\lambda}=1$ when $\lambda$ is obtained by adding one box
to $\alpha$, and zero otherwise. For every nonnegative unit vector $w$ indexed by
$\mathcal Y_N^{(n)}$, the finite covariant $N$-port construction has
normalized entanglement fidelity at least
\begin{equation}
 F_w=\frac{w^{\mathsf T}R^{\mathsf T}Rw}{n^2}.
 \label{up:pbt-primitive}
\end{equation}
This is the representation-theoretic finite resource and fidelity formula of
Ref.~\cite{yoshida2026pbt}, Lemma~10 and Appendix~D, in the
at-most-$n$-row convention. Let $P_\lambda$ be the orthogonal projectors
onto the joint representation sectors of collective unitaries
$U^{\otimes N}$ and port permutations on Alice's $N$ ports,
labelled by $\lambda\in\mathcal Y_N^{(n)}$. The resource is
\begin{equation}
 |\psi_w\rangle=(O_w\otimes I)|\Omega_n\rangle^{\otimes N},
 \qquad
 O_w=n^{N/2}\sum_\lambda
       \frac{w_\lambda}{\sqrt{\rank P_\lambda}}P_\lambda,
 \label{up:pbt-resource}
\end{equation}
where $|\Omega_n\rangle=n^{-1/2}\sum_i|ii\rangle$.
Its squared norm is
$n^{-N}\operatorname{Tr}(O_w^*O_w)=\sum_\lambda w_\lambda^2=1$.
Writing $C$ for Alice's input register, the measurement is the
square-root measurement for the positive signal operators
\[
 \sigma_i=n^{-(N-1)}\Omega_{A_iC}\otimes I_{A_{\ne i}}.
\]
Here $\Omega_{A_iC}$ is the normalized maximally entangled projector
on those two registers. With
$\sigma=\sum_i\sigma_i$ and $P$ its support projector, take
\begin{equation}
 M_i=\sigma^{-1/2}\sigma_i\sigma^{-1/2}+\frac{I-P}{N},
 \label{up:pbt-completion}
\end{equation}
using the inverse on the support. This is a positive measurement
summing to $I$. The added measurement terms cannot reduce fidelity, since their
contributions are traces of products of positive operators. The
resource and measurement retain unitary covariance. Bob completes
the protocol by selecting the indicated port and tracing the others. Thus
\eqref{up:pbt-primitive} applies to a complete finite deterministic
protocol.

Choose product weights $v$ on partitions of $m$ to bring the score
$v^{\mathsf T}RR^{\mathsf T}v/n^2$ close to one. The normalized incidence
transfer to $w$ on partitions of $N=m+1$ makes this score a lower bound on
teleportation fidelity. Following Ref.~\cite{yang2020optimal}, set
\begin{equation}
 q=\left\lfloor \frac{2m/(n-1)+n-2}{3n-2}\right\rfloor,\qquad m_0=m-\frac{[(3n-2)q-n+2](n-1)}2.
 \label{up:pbt-q}
\end{equation}
The subtracted quantity is an integer by parity, and $m_0\geq0$.
The condition $q\geq2$ is equivalent to
$m\geq(n-1)(5n-2)/2$, which is satisfied by our choice of $m$.
Let $\mu$ be the balanced nonincreasing $n$-row partition of $m_0$.
For $k\in\{0,\ldots,q-1\}^{n-1}$ define
\begin{equation}
 \alpha_i(k)=\mu_i+q(2n-3)+1-(q+1)(i-1)+k_i, \qquad i<n,\qquad \alpha_n(k)=\mu_n+(n-1)(q-1)-\sum_{i<n}k_i.
 \label{up:pbt-partitions}
\end{equation}
The row sum is $m$, the last row is nonnegative, and every successive
row gap is at least two: for $i<n-1$, the gap bound follows from
$q+1+k_i-k_{i+1}\geq2$, whereas the final gap equals
$\mu_{n-1}-\mu_n+2+k_{n-1}+\sum_{i<n}k_i$.
The map $k\mapsto\alpha(k)$ is injective. Each supported partition
has exactly $n$ allowed one-box children, including when its last row
is zero.

For $0\leq j<q$, put
\begin{equation}
 g_j=\frac2q\sin^2\!\left(\frac{\pi(2j+1)}{2q}\right),
 \qquad v_{\alpha(k)}=\prod_{i<n}\sqrt{g_{k_i}},
 \label{up:pbt-weights}
\end{equation}
and set the other entries of $v$ to zero. Since $q\geq2$, the finite
sine sum gives $\sum_jg_j=1$, so $v$ is a nonnegative unit vector.
For
$h=\sum_{j=0}^{q-2}\sqrt{g_jg_{j+1}}$, the product-to-sum identity gives
\begin{equation}
 h=\left(1-\frac1q\right)\cos\frac\pi q+\frac1q,
 \qquad 0\leq1-h\leq\frac{\pi^2}{2q^2}.
 \label{up:pbt-overlap}
\end{equation}
For example, the sum of the cosine terms with arguments
$2\pi(j+1)/q$, $0\leq j\leq q-2$, is $-1$, which gives the
displayed expression directly.

The diagonal entries of $RR^{\mathsf T}$ on this support equal $n$.
Two distinct supported partitions have a common child exactly when
the corresponding $k$-vectors differ by $\pm e_i$ or by $e_i-e_j$, with
$i\ne j<n$. Indeed, a common child differs from each parent by one
box, so the two parents differ by moving one box between two rows,
a move that Eq.~\eqref{up:pbt-partitions} expresses as one of the listed
coordinate differences. The row gaps ensure that the children stay
in the permitted partition set. Summing the factorized overlaps gives
\begin{equation}
 F_{\mathrm{est}}:=\frac{v^{\mathsf T}RR^{\mathsf T}v}{n^2}
 =\frac{n+(n-1)(n-2)h^2+2(n-1)h}{n^2}.
 \label{up:pbt-estimation-score}
\end{equation}
With $\delta=1-h\in[0,1]$, expansion of the numerator yields
\begin{equation}
 1-F_{\mathrm{est}}
 \leq\frac{2(n-1)^2\delta}{n^2}
 \leq\frac{2\pi^2(n-1)^2}{n^2q^2}.
 \label{up:pbt-estimation-error}
\end{equation}
The floor in
\eqref{up:pbt-q} implies
\begin{equation}
 q\geq\frac{2[m-n(n-1)]}{(3n-2)(n-1)}
   \geq\frac{m}{(3n-2)(n-1)},
 \label{up:pbt-floor}
\end{equation}
because $m\geq2n(n-1)$.

Set $X=RR^{\mathsf T}$ and
$w=R^{\mathsf T}v/\|R^{\mathsf T}v\|_2$. Its denominator is positive,
and $w$ is a nonnegative unit vector. Let $F$ be the normalized
entanglement fidelity of the resulting teleportation channel.
The finite transfer gives
\begin{equation}
 F\geq F_w
 =\frac{v^{\mathsf T}X^2v}{n^2v^{\mathsf T}Xv}
 \geq\frac{v^{\mathsf T}Xv}{n^2}=F_{\mathrm{est}},
 \label{up:pbt-transfer}
\end{equation}
where Cauchy--Schwarz was used for the second inequality. Together
with \eqref{up:pbt-estimation-error} and \eqref{up:pbt-floor}, this gives
\[
 1-F\leq
 \frac{2\pi^2(n-1)^4(3n-2)^2}{n^2m^2}.
\]

It remains to convert this fidelity into the full diamond error.
The unitary covariance of the resource and
\eqref{up:pbt-completion} gives
$\mathcal T_N(U\rho U^*)=U\mathcal T_N(\rho)U^*$ for every
$U\in\operatorname{SU}(n)$. Its normalized Choi state consequently
has the form
\begin{equation}
 J_{\mathcal T_N}
 =F\Omega_n+(1-F)\frac{I-\Omega_n}{n^2-1},
 \qquad \Omega_n=|\Omega_n\rangle\langle\Omega_n|.
 \label{up:isotropic}
\end{equation}
Under vectorization, $\overline U\otimes U$ acts on matrices by
conjugation, which is unchanged when the conjugating unitaries are multiplied by
scalar phases chosen to make their determinants one. Scalar matrices form one
invariant subspace, while the traceless subspace is irreducible: to prove this,
consider a nonzero invariant complex subspace of traceless matrices.
If a member has a nonzero off-diagonal entry, conjugation and Fourier averaging
over diagonal unitaries isolate an off-diagonal matrix unit. If instead
a nonzero member is diagonal, a two-coordinate rotation gives it such
an entry, after which the same averaging isolates an off-diagonal matrix unit.
Permutations then give all off-diagonal matrix units, and two-coordinate
rotations give all diagonal differences, which together with those matrix
units span the traceless matrices.
A commuting Hermitian operator is therefore scalar on these two
inequivalent subspaces, giving \eqref{up:isotropic} by its trace and
overlap with $\Omega_n$.

The second state in \eqref{up:isotropic} is positive and has input
marginal $I_n/n$, so it is the normalized Choi state of a channel
$\mathcal Q$. Thus
$\mathcal T_N=F\operatorname{id}+(1-F)\mathcal Q$ and
$\|\mathcal T_N-\operatorname{id}\|_\diamond\leq2(1-F)$.
The maximally entangled input attains the reverse inequality because
the two Choi sectors are orthogonal, making the full diamond error
exactly $2(1-F)$. Finally, $\pi<4$, $n-1<n$, and $3n-2<3n$ give
the strict inequality in \eqref{up:pbt-diamond}. The chosen ceiling
for $m$ gives its last inequality.
\end{proof}

\begin{lemma}[Common support for channel programs]
\label{up:compression}
Fix any finite pure correction-free $N$-port teleportation resource
$|\psi\rangle\in A\otimes(\mathbb C^n)^{\otimes N}$ and its
complete instrument, with induced channel $\mathcal T$. There is a
common program space of dimension
\begin{equation}
 b_\psi\leq\binom{N+n^3-1}{N}
 \label{up:compression-dimension}
\end{equation}
and a fixed physical retrieval circuit such that every channel
$\Phi:\mathbb C^n\longrightarrow\mathbb C^n$ has a program
$\sigma_\Phi$ whose retrieved channel is $\Phi\circ\mathcal T$.
The retrieval circuit is trace preserving on every program input.
\end{lemma}

\begin{proof}
We find a common support for the encoded isometries, then decompress
physically onto it. For an isometry
$V:\mathbb C^n\longrightarrow\mathbb C^n_B\otimes\mathbb C^n_E$,
let $|v(V)\rangle=(I_n\otimes V)|\Omega_n\rangle$, a unit vector
in $\mathbb C^{n^3}$ whose identical tensor power lies
in the symmetric tensor subspace $\operatorname{Sym}^N(\mathbb C^{n^3})$,
which has dimension
$\binom{N+n^3-1}{N}$. An orthonormal occupation basis proves this
dimension formula by counting $n^3$ nonnegative integers summing to
$N$.

Expand the fixed resource as
$|\psi\rangle=\sum_I|a_I\rangle\otimes|I\rangle$, where $I$
ranges over the product basis of Bob's $N$ ports, and define the fixed
linear map $X_\psi:\mathbb C^{n^N}\longrightarrow A$ by
$X_\psi|I\rangle=n^{N/2}|a_I\rangle$. Regrouping factors gives
$|\psi\rangle=(X_\psi\otimes I)|\Omega_n\rangle^{\otimes N}$,
so applying $V$ on all Bob ports yields
\begin{equation}
 |\psi(V)\rangle
 =(I_A\otimes V^{\otimes N})|\psi\rangle
 =L_\psi|v(V)\rangle^{\otimes N},
 \label{up:image-identity}
\end{equation}
where $L_\psi$ is the fixed regrouping followed by $X_\psi$ on
the reference factors. The first expression proves that
$|\psi(V)\rangle$ is normalized, while the second places all these vectors
in the common image
\[
 H_\psi=L_\psi\bigl(\operatorname{Sym}^N(\mathbb C^{n^3})\bigr).
\]
Since $H_\psi$ is a linear image of the symmetric subspace,
$b_\psi=\dim H_\psi$ satisfies \eqref{up:compression-dimension}.
For physical decompression, choose an isometry
$W:\mathbb C^{b_\psi}\longrightarrow H_\psi$ onto that image, so that
$|u(V)\rangle=W^*|\psi(V)\rangle$ is normalized and
$W|u(V)\rangle=|\psi(V)\rangle$.

Use Lemma~\ref{up:rank-decomposition} to write
$\Phi=\sum_a w_a\Phi_a$ and choose an $n$-environment
isometry $V_a$ for each component. Define the program as a mixture on this same
space by
\[
 \sigma_\Phi=\sum_a w_a|u(V_a)\rangle\langle u(V_a)|.
\]
Decompression
and tracing all environment factors give the resource
\begin{equation}
 \sum_a w_a
 (\operatorname{id}_A\otimes\Phi_a^{\otimes N})
          (|\psi\rangle\langle\psi|).
 \label{up:mixed-resource}
\end{equation}
The component index $a$ is shared by all ports in each summand.
Because the channel operations on Bob's ports commute with Alice's instrument,
the trace preservation of $\Phi_a$ lets us omit it on every port discarded in branch $i$,
even when the ports are correlated and the input has an external reference.
Correction-free port selection then leaves the retained port undergoing
$\Phi_a$ after the original teleportation output.
Summing the branches and components therefore implements
\begin{equation}
 \sum_a w_a\Phi_a\circ\mathcal T=\Phi\circ\mathcal T.
 \label{up:retrieval-identity}
\end{equation}
Finally, the circuit consists of isometric decompression, partial
traces, a complete instrument, and port selection, so it is completely
positive and trace preserving on every program and input, including
reference-entangled programs.
\end{proof}

\begin{proof}[Proof of Theorem~\ref{up:optimized-program}]
Use the resource in Lemma~\ref{up:pbt-error} and its common program
space from Lemma~\ref{up:compression}. For each complete process,
use the padded realization in Eq.~\eqref{up:padding} and fix one product
encoding for all its tester queries. Encode the initial state exactly
in a factor of dimension $n$ and each of the $K$ padded transition channels
on $\mathbb C^n$ in a separate factor of dimension $b_\psi$.
The product encoding has total dimension
\[
 D_0=nb_\psi^K\leq nB^K.
\]
A fixed simulator uses the initial factor as the visible output and
memory, then consumes the transition programs in order. Because each retrieval map is physical on every program input, the
simulator defines a physical circuit on the entire $D_0$-dimensional
space, even for inputs entangled between program factors. Any adaptive
binary tester composed with this circuit therefore gives an ordinary
effect $0\leq A_E\leq I$ on that space.

For a target's product encoding, replace its transitions by their
retrieved channels one at a time, acting on the visible system together
with the process memory and leaving all other retained registers as the
untouched reference. Their individual full diamond errors are at most
$\eta/K$ by \eqref{up:pbt-diamond} and postprocessing contractivity.
Since the initial state is exact, the terminal trace error is at most
$\eta$, uniformly in the tester, and the binary probability error is
at most $\eta/2$.
Lemma~\ref{up:information} now gives \eqref{up:optimized-capacity}.

For the finite simplification, put $b=n^3-1$. The integral bound
$\sum_{j=1}^k\ln j\geq\int_1^k\ln x\,dx
=k\ln k-k+1$ proves the factorial estimate
$k!\geq(k/e)^k$, which gives
\begin{equation}
 \log_2\binom{N+b}{N}
 \leq\min\left\{
 N\log_2\!\left[e\left(1+\frac bN\right)\right],
 b\log_2\!\left[e\left(1+\frac Nb\right)\right]\right\}.
 \label{up:binomial-two-sides}
\end{equation}
Since $t>1$ and $n\geq2$, the ceiling in
\eqref{up:program-parameters} implies
\[
 24n^2t\leq N\leq24n^2t+2\leq25n^2t,
 \qquad \frac bN\leq\frac{n}{24t}.
\]
Apply these inequalities to the first term in
\eqref{up:binomial-two-sides} to obtain
\eqref{up:optimized-simplified}.
\end{proof}

\subsection{An alternative bound from heralded teleportation}

The optimized deterministic program used in Theorem~\ref{up:optimized-program}
trades simulation accuracy against the number of ports. An exact
successful branch gives another bound, with a different dependence on
the margin.

\begin{proposition}[Heralded-program capacity bound]
\label{inf:heralded}
Let $d\geq2$, $r\geq1$, $K\geq1$, and $n=dr$. Set
\[
 s=2K(n^2-1),\qquad p=\frac{2K}{2K+1},\qquad
 D_h=n^{1+4K^2(n^2-1)}.
\]
For every $0<\gamma\leq1$,
\begin{equation}
 \operatorname{sfat}_\gamma\mathcal P_{d,r}(K)
 \leq\left\lfloor\frac{\log_2D_h}{c_{p^K\gamma}}\right\rfloor
 <\frac{8\ln2}{\gamma^2}\log_2D_h.
 \label{inf:heralded-bound}
\end{equation}
\end{proposition}
\begin{proof}
Optimized probabilistic port-based teleportation in dimension $n$ with
$s$ ports has perfect successful branches and success probability
$p=s/(s+n^2-1)$~\cite{studzinski2017pbt}. Routing the successful output
port to a common output and summing those branches gives the completely
positive map $p\operatorname{id}_n$: perfect conditional fidelity
makes each Kraus operator preserve every ray. Such a linear
operator is a scalar identity, and the average success probability
fixes the sum of squared scalars to $p$. This equality of complete maps
also holds with arbitrary untouched reference systems.

For a channel $\Phi$, apply $\Phi$ to each receiver port of the
optimized resource state to obtain $\omega_\Phi$, a state of
dimension $n^{2s}$. The sender's measurement commutes with these channel
applications, and trace preservation on unselected ports leaves their
discard unchanged. The fixed processor acting on $\omega_\Phi$ thus
has complete successful map $p\Phi$.

Using the padded realization in Eq.~\eqref{up:padding}, combine its
initial state and $K$ channel programs into a single state $\Omega_P$ of
dimension $D_h=n^{1+2sK}$. Running a tester through the fixed processors and
rejecting whenever a herald reports failure gives an accepting effect
$M_E$ on the program, satisfying
\[
 \operatorname{Tr}(M_E\Omega_P)=p^Kf_P(E).
\]
Reference-complete equality at each transition proves this identity
through arbitrary adaptive prefixes. A $\gamma$-shattered process tree
therefore yields a $p^K\gamma$-shattered state-program tree, with its
thresholds multiplied by $p^K$. Lemma~\ref{up:information} proves the
first inequality.

Since $c_z\geq z^2/(2\ln2)$ for $0\leq z\leq1$, and
Bernoulli's inequality gives
$p^K\geq1-K/(2K+1)>1/2$, the second inequality follows.
\end{proof}

\section{Capacity of the full causal class and memory saturation}
\label{up:causal-section}

For large memory, we work directly with the entire convex set of
causal processes using the normalized Choi description
\cite{chiribella2009theoretical}. With each $A'_t$ a reference copy of
the input $A_t$, let $\mathcal C_d(K)$ be the set of positive trace-one
operators $\rho_K$ on $B_0A'_1B_1\cdots A'_KB_K$ whose prefix marginals
$\rho_t$ on $B_0A'_1B_1\cdots A'_tB_t$ obey the causal constraints
\begin{equation}
 \Tr_{B_t}\rho_t=\rho_{t-1}\otimes I_{A'_t}/d.
 \label{main:causality}
\end{equation}
The full Choi operator $\rho_K$ has matrix dimension $d^{2K+1}$, and the usual
unnormalized process Choi operator is $d^K\rho_K$. For a normalized binary
tester, the process-tester pairing, with the transpose convention absorbed
into the tester, gives positive operators $T_E,T_E^c$ whose sum pairs to
one with every causal state and may differ from the ambient identity.
Their outcome probabilities satisfy
\begin{equation}
 f_P(E)=\Tr(T_E\rho_P),\qquad
 1-f_P(E)=\Tr(T_E^c\rho_P).
 \label{main:positive-pair}
\end{equation}
These are the positive measuring-tester operators
of Ref.~\cite{chiribella2009theoretical}, scaled by $d^K$ and transposed
consistently with the process Choi convention. With the pair in
Eq.~\eqref{main:positive-pair}, write $f_E(\rho)=\Tr(T_E\rho)$ for the
acceptance probability.

Write $\nu$ for the matrix size of a normalized causal Choi state:
\begin{equation}
 \nu=d^{2K+1}.
 \label{up:causal-dimension}
\end{equation}
The causal constraints in Eq.~\eqref{main:causality} define a compact
convex set containing $\rho_*=I_\nu/\nu$ and the Choi state of every
memory-capped process. Conversely, the causal equations have a sequential
channel realization, with a purification proof and a sufficient memory
bound given in Lemma~\ref{sat:realization}.

Online learning over this set has
been studied with matrix-entropy methods
\cite{bansal2025panoply}. We instead use a log-determinant potential
and bound its curvature by the change in a physical tester's
outcome probability. Positive-operator log-determinant inequalities
and maximally mixed shrinkage have established uses in quantum
online learning \cite{tsai2023logdet}. The argument below uses
both outcomes of a complete causal tester.

\begin{theorem}[Full causal upper bound]
\label{up:causal-upper}
For $d\geq2$, $K\geq1$, and $0<\gamma\leq1$,
\begin{equation}
 \operatorname{sfat}_\gamma(\mathcal P_{d,r}(K))
 \leq\operatorname{sfat}_\gamma(\mathcal C_d(K))
 \leq\left\lfloor
       \frac{4d^{2K+1}\ln(4/\gamma)}{\gamma^2}
       \right\rfloor
 \label{up:causal-capacity}
\end{equation}
for every $r\geq1$.
\end{theorem}

\begin{proof}
We choose a path in a shattered tree so that cumulative excess loss
against one fixed comparator grows linearly with depth.
We compare this loss with a square-root online upper bound, obtained
from log-determinant curvature in a tester norm and shrinkage of the
causal domain to bound the potential while keeping the response error
below the shattering margin.
Let $V$ be the real vector space parallel to the causal affine
space, with operational norm
\begin{equation}
 \|\Delta\|_{\mathrm{op}}
 =\sup_E|\operatorname{Tr}(T_E\Delta)|.
 \label{up:causal-norm}
\end{equation}
The supremum is finite: positivity and
$\operatorname{Tr}(T_E\rho_*)\leq1$ imply
$\operatorname{Tr}T_E\leq\nu$, so each pairing is bounded by
$\nu\|\Delta\|_\infty$. To separate nonzero elements of $V$, a tester can prepare the normalized
Choi experiment by supplying one half of a normalized maximally entangled
pair at each input and retaining the other halves and all outputs, then
measure any effect on the collected Choi state, since such effects distinguish
every nonzero Hermitian trace-zero $\Delta$. Homogeneity and the triangle inequality follow from
the supremum. A query, modulo constants on the causal affine
space, therefore has dual norm at most one.

For every $\Delta\in V$,
\begin{equation}
 \operatorname{Tr}(T_E^{\mathrm c}\Delta)
 =-\operatorname{Tr}(T_E\Delta).
 \label{up:complement-tangent}
\end{equation}
Indeed, $\rho_*+s\Delta$ is a positive causal state for all
sufficiently small positive and negative $s$. Differentiate the
sum of the two identities in \eqref{main:positive-pair} along this
segment.

On positive-definite causal states define
\begin{equation}
 R(\rho)=-\ln\det\rho-\nu\ln\nu,
 \qquad
 q_\rho(\Delta)
 =\operatorname{Tr}(\rho^{-1}\Delta\rho^{-1}\Delta).
 \label{up:logdet}
\end{equation}
The second expression is the Hessian of $R$ evaluated twice on
$\Delta$. Fix a tester, set
$p=\operatorname{Tr}(T_E\rho)$ and
$\ell=\operatorname{Tr}(T_E\Delta)$, and put
\[
 C=\rho^{-1/2}\Delta\rho^{-1/2},\quad
 B=\rho^{1/2}T_E\rho^{1/2},\quad
 B^{\mathrm c}=\rho^{1/2}T_E^{\mathrm c}\rho^{1/2}.
\]
The positive matrices $B,B^{\mathrm c}$ have traces $p,1-p$,
and their pairings with $C$ are $\ell,-\ell$.
Their combination recovers the acceptance change with bounded
Hilbert--Schmidt norm: even for noncommuting pairs,
$H=(1-p)B-pB^{\mathrm c}$ satisfies
$\operatorname{Tr}(HC)=\ell$ and
\begin{equation}
 \begin{aligned}
 \operatorname{Tr}H^2
  &=(1-p)^2\operatorname{Tr}B^2
    +p^2\operatorname{Tr}[(B^{\mathrm c})^2]
    -2p(1-p)\operatorname{Tr}(BB^{\mathrm c})\\
  &\leq 2p^2(1-p)^2\leq\frac18.
 \end{aligned}
 \label{up:tester-curvature-matrix}
\end{equation}
The first inequality uses
$\operatorname{Tr}B^2\leq(\operatorname{Tr}B)^2$ and its
complementary version, together with
$\operatorname{Tr}(BB^{\mathrm c})\geq0$. Hilbert--Schmidt Cauchy--Schwarz now gives
$\ell^2\leq q_\rho(\Delta)/8$. The inequality also covers
$p=0,1$, when the relevant positive matrix vanishes and
$\ell=0$. Taking the supremum over all testers proves
\begin{equation}
 q_\rho(\Delta)\geq8\|\Delta\|_{\mathrm{op}}^2.
 \label{up:strong-curvature}
\end{equation}
Integration along a segment of positive-definite causal states gives
\begin{equation}
 R(y)\geq R(x)+DR(x)[y-x]
              +4\|y-x\|_{\mathrm{op}}^2.
 \label{up:strong-convexity}
\end{equation}
Thus $R$ is strongly convex with constant $\mu=8$ in the
convention $\mu\|y-x\|^2/2$.

Fix $0<\epsilon<1$ and consider the compact convex domain
\[
 \mathcal C_\epsilon
 =\{(1-\epsilon)\rho+\epsilon\rho_*:
       \rho\in\mathcal C_d(K)\}.
\]
Its states have eigenvalues at least $\epsilon/\nu$, while the
arithmetic-geometric mean inequality at trace one gives
$R\geq0$ and $R(\rho_*)=0$, so
\begin{equation}
 0\leq R(\rho)\leq A:=\nu\ln(1/\epsilon)
 \quad(\rho\in\mathcal C_\epsilon).
 \label{up:logdet-radius}
\end{equation}
Replacing a comparator $\rho$ by
$u=(1-\epsilon)\rho+\epsilon\rho_*$ changes each tester
probability by at most $\epsilon$, since both original
probabilities lie in $[0,1]$.

The online potential calculation uses constrained strongly convex
regularization
\cite{liu2022ftrl}. For a finite depth $L\geq1$ and rate
$\alpha>0$, define signed linear query losses
$g_j(\rho)=s_jf_{E_j}(\rho)$, $s_j\in\{-1,1\}$, and set
\begin{equation}
 F_j=R+\alpha\sum_{i=1}^jg_i,\qquad
 x_j=\mathop{\rm argmin}_{\rho\in\mathcal C_\epsilon}F_{j-1}(\rho),
 \qquad
 \phi_j=\min_{\rho\in\mathcal C_\epsilon}F_j(\rho).
 \label{up:ftrl-definitions}
\end{equation}
The minimizers exist by compactness and are unique by strong
convexity. For queries and signs allowed to depend on past observations
and on the current $x_j$, each loss is one-Lipschitz for
\eqref{up:causal-norm}. Put $\delta_j=\|x_{j+1}-x_j\|_{\mathrm{op}}$. Optimality at $x_j$ gives a
nonnegative directional derivative of $F_{j-1}$ along the feasible segment
to $x_{j+1}$, including at boundary minimizers. Together with
\eqref{up:strong-convexity}, this gives
\begin{equation}
 \phi_j-\phi_{j-1} \geq\frac\mu2\delta_j^2+\alpha g_j(x_{j+1}) \geq\alpha g_j(x_j)+\frac\mu2\delta_j^2-\alpha\delta_j \geq\alpha g_j(x_j)-\frac{\alpha^2}{2\mu}.
 \label{up:ftrl-step}
\end{equation}
Since $\phi_0=0$ and
$\phi_L\leq R(u)+\alpha\sum_jg_j(u)$, summation yields,
for every $u\in\mathcal C_\epsilon$,
\begin{equation}
 \sum_{j=1}^L[g_j(x_j)-g_j(u)]
 \leq\frac A\alpha+\frac{L\alpha}{2\mu}.
 \label{up:regret}
\end{equation}
Taking $\alpha=\sqrt{2\mu A/L}$ makes the right side
$\sqrt{AL}/2$.

Suppose a complete depth-$L$ tree is $\gamma$-shattered.
At its current node, with tester $E_j$ and threshold $a_j$,
compute $x_j$. If $f_{E_j}(x_j)\geq a_j$, choose the lower
child and set $s_j=1$. Otherwise choose the upper child and
set $s_j=-1$. After this path is fixed, the shattering
definition supplies one complete process $\rho$ witnessing
every chosen branch. For all $j$,
\[
 g_j(x_j)-g_j(\rho)\geq\gamma/2.
\]
For its single smoothed comparator
$u=(1-\epsilon)\rho+\epsilon\rho_*$, the uniform probability error and
\eqref{up:regret} give
\[
 L(\gamma/2-\epsilon)\leq\sqrt{AL}/2.
\]
Choose $\epsilon=\gamma/4$. Squaring and cancelling the
positive $L$ yields
$L\leq4\nu\ln(4/\gamma)/\gamma^2$. The depth-zero case
is immediate, and the integer floor bounds every finite
shattered depth. The inclusion of each $\mathcal P_{d,r}(K)$
in the causal domain gives the same upper bound for every memory cap.
\end{proof}

Theorem~\ref{up:causal-upper} gives the upper estimate used
at the saturated endpoint in Theorem~\ref{sat:matching}.
For smaller caps, one can take the minimum of this estimate, the bounds in
Proposition~\ref{up:physical-cover} and Theorem~\ref{up:optimized-program},
and the heralded estimate~\eqref{inf:heralded-bound}, optimizing the permitted
choice of $\eta$ in the finite quantum-program bound.

\subsection{Dense tables and the saturated memory regime}
\label{sat:section}

The causal upper bound in Eq.~\eqref{up:causal-capacity} has a matching
horizon dependence once the memory can retain a complete quantum prefix.
The memory cap used below is sufficient for every process on this interface.

\begin{lemma}[A sufficient memory for the full causal class]
\label{sat:realization}
Every initialized $K$-transition process with visible dimension $d$ and
$K\geq1$ has a realization with propagated memory at most $d^{2K-1}$.
\end{lemma}
\begin{proof}
Take $\rho_K\in\mathcal C_d(K)$ with marginals $\rho_t$ on
$B_0A'_1B_1\cdots A'_tB_t$, so that
\[
 \operatorname{Tr}_{B_t}\rho_t
 =\rho_{t-1}\otimes I_{A'_t}/d.
\]
Choose a minimal purification of $\rho_0$ using a memory $M_0$ of
dimension $\operatorname{rank}\rho_0\leq d$.
Given a realized purification of $\rho_{t-1}$, tensor it with a
maximally entangled state on $A_tA'_t$. The resulting vector and a
purification of $\rho_t$ purify the same state on the old visible
registers and $A'_t$, so uniqueness of purification gives an
isometry from $M_{t-1}A_t$ to $B_tM_t$, with
\[
 \dim M_t=\operatorname{rank}\rho_t\leq d^{2t+1}.
\]
Minimality makes the input purifying space exactly the support dimension.
On a larger input space, extend the channel by projecting onto the original
input subspace and its orthogonal complement and using a fixed CPTP map on
the complement.
At the last transition the purifying output is discarded locally, so only
$M_0,\ldots,M_{K-1}$, whose dimensions are at most $d^{2K-1}$,
propagate between transitions. This also reconstructs the usual sequential
realization of a causal Choi operator~\cite{chiribella2009theoretical}.
\end{proof}

Lemma~\ref{sat:realization} identifies the class at the sufficient
memory cap $d^{2K-1}$ with the full causal class.

\begin{theorem}[Saturation at sufficient memory]
\label{sat:matching}
Let $d\geq2$, $K\geq1$, $r\geq d^{2K-1}$ and $0<\gamma\leq1$.
Then
\begin{equation}
 d^{2K}\lfloor\log_2d\rfloor
 \leq \operatorname{sfat}_\gamma\mathcal P_{d,r}(K)
 \leq
 \left\lfloor
 \frac{4d^{2K+1}\ln(4/\gamma)}{\gamma^2}
 \right\rfloor.
 \label{sat:bounds}
\end{equation}
For fixed $d$ and $\gamma$, both the full causal class and the indicated
memory-capped class therefore have sequential capacity
$\Theta_{d,\gamma}(d^{2K})$.
\end{theorem}
\begin{proof}
Put $b=\lfloor\log_2d\rfloor$. Assigning an arbitrary $b$-bit word
$z_a$ to each address $a\in\{1,\ldots,d^2\}^K$ gives a table with
$d^{2K}b$ independent binary coordinates.

Initially the process emits half of a maximally entangled $d$-level pair
and retains its mate $E_1$. At transition $t<K$, it measures the pair
$A_tE_t$ in a generalized Bell basis, stores the resulting $d^2$-ary
digit, and emits half of a fresh maximally entangled pair while retaining
$E_{t+1}$. At transition $K$, the last Bell outcome completes $a$, after
which the process emits the computational-basis state labeled by $z_a$
and discards its remaining registers. These measurements, conditional
preparations, and discards define normalized channels on all inputs,
with dephasing of the classical prefix specifying their action on
off-diagonal input blocks. For $t<K$, the stored prefix of Bell outcomes
and the retained half $E_{t+1}$ have joint dimension $d^{2t+1}$,
so the largest propagated memory is $d^{2K-1}$.

To query an address, the tester applies the corresponding generalized
Weyl operator to each received half-pair before returning it as the next
input. With the Bell basis chosen accordingly, each intended address
digit is obtained with probability one. The final measurement reads one chosen bit of $z_a$, making every
table coordinate an exact zero or one response. A full assignment of the coordinates fixes one complete
process, which witnesses the entire path of the coordinate-query tree.
This proves the lower bound at margin one and hence at every
$0<\gamma\leq1$.

Equation~\eqref{up:causal-capacity} gives the upper bound.
\end{proof}

Under direct tree conversion of the stated regret or potential estimates,
the full-causal upper also sharpens the horizon dependence obtained by
combining the full-strategy matrix-entropy estimate of
Ref.~\cite{bansal2025panoply} with a tree argument.
Writing $\nu=d^{2K+1}$, that direct conversion gives order
$\nu^2\ln\nu/\gamma^2$, whereas the causal upper above is
$4\nu\ln(4/\gamma)/\gamma^2$.

\section{Classical and quantum programs for complete processes}
\label{cost:section}

We use the phase-channel programs of Appendix~\ref{pauli:section} to
compare classical and quantum simulation storage. A program is supplied
once to simulate a complete process uniformly over testers.

\subsection{Obstruction to a finite classical channel dictionary}
\label{pauli:dictionary-section}

A classical channel dictionary consists of fixed CPTP maps, with a program
given by a classical label sampled from a parameter-dependent distribution.
Averaging over labels gives the corresponding convex mixture of maps,
each of which may itself act quantum mechanically.
Probabilistic channel and state synthesis study convex approximation
with such randomized encodings
\cite{akibue2024unitary,akibue2024state}.

For convex channel classes, Ref.~\cite{gschwendtner2021programmability}
characterizes finite classical dictionaries with entries inside the
class by whether that class is a polytope, a convex hull of finitely
many points. It also exhibits covariant classes, whose channels obey
a prescribed group symmetry, with finite exact quantum programs and
nonpolytopal channel sets.
For the Pauli family, the rotating support of the Choi operator gives
the following obstruction to a finite classical dictionary.

\begin{proposition}[A rank-three obstruction]
\label{pauli:dictionary}
Let $q_I,q_X,q_Z>0$ and $q_Y=0$. There is no finite collection of
fixed qubit channels whose convex hull contains
$\{\mathcal E_{\theta,q}:\theta\in\mathbb R/2\pi\mathbb Z\}$.
\end{proposition}

\begin{proof}
Use the normalized Choi operator, with output--input basis order.
The sum of the $I$ and $Z$ contributions has full support on the
even subspace $\operatorname{span}\{|00\rangle,|11\rangle\}$.
The $X$ contribution has positive weight on the single line
\begin{equation}
 \ell_\theta
 =\operatorname{span}\{|10\rangle+e^{i\theta}|01\rangle\}
 \subset\operatorname{span}\{|10\rangle,|01\rangle\}.
 \label{pauli:odd-line}
\end{equation}
The Choi operator thus has rank three. If a positive Choi operator
occurs with positive coefficient in a convex representation, its
support lies in the target support: evaluation on any vector in
the target kernel makes each nonnegative summand zero.

A fixed dictionary operator with nonzero compression to the odd
subspace $\operatorname{span}\{|10\rangle,|01\rangle\}$ must
therefore have that compression supported on
$\ell_\theta$. Distinct phases modulo $2\pi$ give distinct lines,
so this operator can participate for at most one phase. Operators
with zero odd compression cannot supply the required positive odd
weight. A finite dictionary consequently cannot represent the entire
circle.
\end{proof}

For example, $q=(7/16,1/2,0,1/16)$ has $e=\beta=1/16$, with
$\beta$ defined in Eq.~\eqref{pauli:program-size}.
Its three-qubit program in Eq.~\eqref{pauli:program} is exact,
although Proposition~\ref{pauli:dictionary} rules out every finite
classical channel dictionary. This separates the two types of exact
channel representation within the capacity law of
Theorem~\ref{pauli:law}.

\subsection{Classical and quantum programs for a complete process}
\label{pauli:process-programs}

Consider the single-address family with $R=1$, $K=T$, and known
$q=(a,p,p,0)$, where $a=1-2p$ and $0<p\leq1/8$. Its $T$ phases
vary independently over $\mathbb R/2\pi\mathbb Z$, and its initialized
process $P_{\boldsymbol\theta}$ has propagated memory dimension one.
A classical process dictionary is a fixed collection
$\mathcal D_1,\ldots,\mathcal D_M$ of trace-preserving causal quantum
processes on the same interface, each of which may use arbitrary quantum
memory and correlations between slots. An encoder that knows
$\boldsymbol\theta$ sends one label sampled with probabilities
$\lambda_j(\boldsymbol\theta)$, and the decoder executes $\mathcal D_j$.
Only this label is passed or shared between encoder and decoder, while
decoder-local randomness and ancillary states are included in each fixed
entry.

Write $M_{\rm cl}^{(T)}(\epsilon;p)$ for the smallest dictionary size
such that every target admits one mixture satisfying
\begin{equation}
 \frac12D\!\left(P_{\boldsymbol\theta},
        \sum_j\lambda_j(\boldsymbol\theta)\mathcal D_j\right)
 \leq\epsilon,
 \qquad
 B_{\rm cl}^{(T)}(\epsilon;p)
 =\left\lceil\log_2 M_{\rm cl}^{(T)}(\epsilon;p)\right\rceil .
 \label{pauli:dictionary-cost}
\end{equation}
Thus the same mixture approximates all normalized adaptive testers.
If no finite dictionary meets the error requirement, set
$M_{\rm cl}^{(T)}(\epsilon;p)=B_{\rm cl}^{(T)}(\epsilon;p)=+\infty$.
For quantum programs, let $B_{\rm q}^{(T)}(0;p)$ be the minimum
number of qubits supplied once at the start to a fixed CPTP processor
that implements every target exactly from a target-dependent program state.
Both costs measure program
storage, with program preparation and decoder computation treated as separate
resources.

Visible quantum data compression studies classical and quantum
register costs when the encoder knows the source string, including
continuously parameterized, arbitrarily varying product
sources~\cite{hayden2002trading}. Lower bounds for randomized state
synthesis use high-fidelity decoder states to obtain covering
bounds~\cite{akibue2024state}. Here that argument is applied to
orthogonal noisy sectors, keeping the dependence on the number of
phase coordinates.

\begin{theorem}[Program costs for noisy phase sequences]
\label{pauli:whole-program-cost}
For the dictionary model above, let $T\geq1$, $0<p\leq1/8$,
$0<\epsilon\leq1/64$, and $m_0=\lceil T/2\rceil$. Then
\begin{align}
 &\frac{m_0}{2}\log_2\frac{m_0}{36\epsilon}-T
 \leq B_{\rm cl}^{(T)}(\epsilon;p)
 \leq\frac T2\log_2\frac T\epsilon+T+1,
 \label{pauli:classical-bits}\\
 &B_{\rm q}^{(T)}(0;p)
 \leq T\left\lceil\log_2\frac1{2p}\right\rceil .
 \label{pauli:quantum-bits}
\end{align}
In particular,
$B_{\rm cl}^{(T)}(\epsilon;p)=\Theta(T\log_2(T/\epsilon))$
with absolute constants on this parameter range.

For each fixed $T,p$,
$M_{\rm cl}^{(T)}(\epsilon;p)=\Theta_{T,p}(\epsilon^{-T/2})$
as $\epsilon\downarrow0$, and $M_{\rm cl}^{(T)}(0;p)=\infty$.
\end{theorem}

At fixed $0<p\leq1/8$, $0<\gamma\leq1/16$ and
$0<\epsilon\leq1/64$, this family has linear response capacity by
Theorem~\ref{pauli:law}, while the bounds above give classical storage
of order $T\log T$ and a linear upper bound on the storage of an exact
quantum program.

\begin{proof}
For the classical bounds, we collect output states and bound the phase
volume covered by each dictionary-state sector, then use polygon
interpolation for the upper bound. After establishing the fixed-parameter
asymptotics and the obstruction to finite exact classical simulation,
we use the exact phase processor to obtain the quantum upper bound.

Discard the known initial carrier, then at every slot prepare a fresh
Bell pair, send one half through the channel, and retain its reference
and output, using $O(T)$ private qubits and one call to each slot.
The collected target output is
\begin{equation}
 \rho_{\boldsymbol\theta}=\bigotimes_{t=1}^T\rho_{\theta_t},
 \qquad
 \rho_\theta=a|\Phi_\theta\rangle\langle\Phi_\theta|
       +p|01\rangle\langle01|+p|10\rangle\langle10|,
 \quad |\Phi_\theta\rangle
       =\frac{|00\rangle+e^{i\theta}|11\rangle}{\sqrt2}.
 \label{pauli:collected-state}
\end{equation}
Each dictionary entry produces a fixed density matrix $\sigma_j$,
possibly correlated across all slots. Since every final binary
effect is allowed, Eq.~\eqref{pauli:dictionary-cost} implies
$\tfrac12\|\rho_{\boldsymbol\theta}-\sum_j\lambda_j\sigma_j\|_1
\leq\epsilon$.

Let $P_e=|00\rangle\langle00|+|11\rangle\langle11|$ and
$P_o=I-P_e$. The orthogonal sector $P_S$, for $S\subseteq[T]$,
uses $P_e$ on slots in $S$ and $P_o$ on the other slots in their
original order, with target weight $a^{|S|}(1-a)^{T-|S|}$ and a conditional
state consisting of the pure product $|\Phi_{\boldsymbol\theta_S}\rangle$
on the even factors and $P_o/2$ on each odd factor. For
$G=\{S:|S|\geq m_0\}$ and $P_G=\sum_{S\in G}P_S$,
Bernoulli monotonicity and symmetry at probability $1/2$ give
$\beta_T=\Tr(P_G\rho_{\boldsymbol\theta})\geq1/2$.
To find a high-overlap sector, use $P_G$ to measure the mass in $G$ and
define a second effect to measure leakage from the required phase states
\begin{equation}
 Q_{\boldsymbol\theta}=\sum_{S\in G}
 \left(P_e^{\otimes|S|}
       -|\Phi_{\boldsymbol\theta_S}\rangle
        \langle\Phi_{\boldsymbol\theta_S}|\right)
       \otimes P_o^{\otimes(T-|S|)},
 \qquad 0\leq Q_{\boldsymbol\theta}\leq P_G\leq I.
 \label{pauli:sector-effect}
\end{equation}
Tensor factors retain their slot order. Each numerical phase vector
specifies a fixed allowed final effect with zero target expectation,
covered by the all-tester guarantee.

Put $w_{jS}=\Tr(P_S\sigma_j)$. For a positive weight, normalize
$P_S\sigma_jP_S$ and trace out the odd factors, obtaining a fixed
state $\tau_{jS}$ on the $|S|$ even factors. When $w_{jS}=0$,
choose any fixed density matrix $\tau_{jS}$ on those factors.
The two effects give
\begin{equation}
\sum_{\substack{1\leq j\leq M\\S\in G}}\lambda_jw_{jS}\geq\beta_T-\epsilon,\qquad \sum_{\substack{1\leq j\leq M\\S\in G}}\lambda_jw_{jS} \left(1-\langle\Phi_{\boldsymbol\theta_S}| \tau_{jS}|\Phi_{\boldsymbol\theta_S}\rangle\right) \leq\epsilon.
 \label{pauli:sector-leakage}
\end{equation}
Consequently some positive-weight pair $(j,S)$ has overlap at least
$1-t$, where $t=\epsilon/(\beta_T-\epsilon)
\leq64\epsilon/31\leq1/31$.

For any density matrix $\tau$ on $m$ even factors, consider the set
$A_\tau(t)$ of phase vectors whose product state has overlap at
least $1-t$ with $\tau$. When this set is nonempty, fix an anchor in it. The trace distance from each such pure state to $\tau$
is at most $\sqrt t$: apply convexity to a pure-state decomposition
of $\tau$, then use concavity of the square root. The triangle
inequality and the pure-state distance formula imply, for wrapped
phase differences $\Delta_i\in[-\pi,\pi]$ from the anchor,
\begin{equation}
 1-4t\leq\prod_{i=1}^m\cos^2(\Delta_i/2)
       \leq\exp\!\left(-\tfrac14\sum_i\Delta_i^2\right),
 \qquad \sum_i\Delta_i^2\leq32t.
 \label{pauli:phase-cap}
\end{equation}
The second inequality follows by integrating $\tan x\geq x$.
The last uses $-\ln(1-4t)\leq8t$ for $t\leq1/8$.
The unit $m$-ball volume satisfies
$V_m\leq(2\pi\mathrm e/m)^{m/2}$, as follows by bounding the
Gaussian integral below on a ball of radius $\sqrt m$.
Normalized Haar measure on the phase torus therefore obeys
\begin{equation}
 \mu_m(A_\tau(t))
 \leq\left(\frac{16\mathrm e\,t}{\pi m}\right)^{m/2}
 \leq\left(\frac{36\epsilon}{m}\right)^{m/2}.
 \label{pauli:cap-volume}
\end{equation}
Here $1024\mathrm e/(31\pi)<36$.

By Eq.~\eqref{pauli:sector-leakage}, the cylinders associated with the
at most $M2^T$ fixed pairs $(j,S)$ cover the full $T$-torus, each cylinder
having the measure of its $|S|$-dimensional phase set. Since
$(36\epsilon/m)^{m/2}$ decreases for $m\geq1$, this cover gives
\begin{equation}
 1\leq M2^T\left(\frac{36\epsilon}{m_0}\right)^{m_0/2},
 \label{pauli:cylinder-count}
\end{equation}
which proves the lower bound in Eq.~\eqref{pauli:classical-bits}.
It also gives $B_{\rm cl}^{(T)}\geq
(T/4)\log_2(T/\epsilon)-3T$ by monotonicity of
$x\log_2(x/(36\epsilon))$ for $x\geq1/2$.
If $\log_2(T/\epsilon)\geq24$, this is at least
$(T/8)\log_2(T/\epsilon)$. Otherwise $T<2^{18}$.
One dictionary entry cannot approximate both first-slot phases
$0$ and $\pi$: their Bell outputs have trace distance
$a\geq3/4>2\epsilon$. Hence $B_{\rm cl}^{(T)}\geq1$ in this
remaining range, giving an absolute positive lower constant for
the asserted uniform order.

For the upper bound, replacing the $T$ complete channels one at a time
allows a half-diamond error of $\epsilon/T$ per slot. Use
$M_0=\lceil(\pi/2)\sqrt{aT/\epsilon}\rceil$ equally spaced phases at each
slot and fill their regular polygon by randomized interpolation, using
the phase-circle convex geometry of probabilistic unitary
synthesis~\cite{akibue2024unitary}.
Channels with the same populations and coherence multipliers $az,aw$
have half-diamond distance $a|z-w|/2$: on a normalized positive
qubit--reference block matrix with off-diagonal block $B$, the
output difference has trace norm $2a|z-w|\|B\|_1$, while
$\|B\|_1\leq1/2$. An equatorial input attains equality.
The polygon's maximal local error is
$a\sin^2(\pi/(2M_0))\leq\epsilon/T$.

Telescoping this bound over the complete channels gives response error
at most $\epsilon$ for every adaptive tester. Sampling all local labels
together at the start gives a legal dictionary of $M_0^T$ fixed product processes.
Since $M_0\leq2\sqrt{T/\epsilon}$, taking the logarithm and its ceiling
gives the upper bound in Eq.~\eqref{pauli:classical-bits}.

For fixed $T,p$ and $\epsilon\leq a^T/9$, the all-even sector's mass $a^T$
and leakage give $t=\epsilon/(a^T-\epsilon)\leq1/8$. The first bound in
Eq.~\eqref{pauli:cap-volume} gives $M\geq c_{T,p}\epsilon^{-T/2}$,
with the opposite bound from the polygon construction.

To rule out a finite exact classical dictionary, discard the initial carrier
and ignore later outputs. This induces one fixed first-slot CPTP map per
dictionary entry, using its reduced initial memory state. First-slot tests with a
reference would therefore turn an exact process dictionary into an exact
channel dictionary. Composing the output with $X$ changes the target
Pauli weights to $(p,a,0,p)$, so Proposition~\ref{pauli:dictionary}
excludes a finite dictionary.

Finally, the exact processor in Eq.~\eqref{pauli:exact-program}
has $\beta=p$ and $W=X$ for this family. Equation
\eqref{pauli:table-program} supplies all $Tb$ program qubits once,
with $b=\lceil\log_2(1/(2p))\rceil$, and consumes the appropriate
factor at each slot. Its complete-reference identity proves
Eq.~\eqref{pauli:quantum-bits} under every adaptive tester.
\end{proof}

\begin{remark}[Phase alphabet and noise support]
\label{pauli:dictionary-boundaries}
A fixed $L$-point phase alphabet has an exact classical program of
at most $\lceil T\log_2L\rceil$ bits. For weights
$q=(a-\eta,p,p,\eta)$ with fixed $0<\eta\leq a/8$, the
coherence circle has radius $a-2\eta$ inside the available radius
$a$. An inscribed regular polygon with
$M=\lceil(\pi/2)\sqrt{a/\eta}\rceil$ vertices contains it,
since $a\cos(\pi/M)\geq a-2\eta$. Its product dictionary
gives an exact $O_\eta(T)$-bit classical program.
\end{remark}

\bibliography{capacity_references}
\end{document}